\documentclass[12pt]{amsart}
\usepackage{subcaption}
\usepackage{amsfonts,amsthm}
\usepackage[foot]{amsaddr}
\usepackage{subcaption}
\usepackage{amsmath}
\usepackage{mathrsfs}
\usepackage{enumerate}
\usepackage{xcolor}
\usepackage{graphicx}
\usepackage{natbib}

\usepackage{amssymb}
\usepackage{hyperref}
\usepackage{booktabs}
\usepackage[font=footnotesize]{caption}                    
\usepackage{lscape}
\usepackage{geometry}
\newgeometry{left=1.25in, right=1.25in,bottom=1.25in,top=1.25in} 
\usepackage{todonotes} %\usepackage[disable]{todonotes}
\usepackage{setspace}
\newcommand{\tr}{\mathrm{tr}}
\newcommand{\cov}{\mathrm{cov}}
\usepackage{xr}
\newtheorem {theorem}{Theorem}
\newtheorem {assumption}{Assumption}

\newtheorem {corollary}{Corollary}

\newtheorem {lemma}{Lemma}
\newtheorem {proposition}{Proposition}

\newtheorem* {E1}{E1}
\newtheorem* {E2}{E2}
\newtheorem* {E3}{E3}

\numberwithin {equation}{section} 
\numberwithin {theorem}{section} 
\numberwithin {proposition}{section}
\numberwithin {lemma}{section} 
\numberwithin {corollary}{section} 

\input{tcilatex}
\makeatletter
\renewcommand\subsection{\@startsection{subsection}{2}%
  \z@{0.2\linespacing\@plus 0.2\linespacing}{-.5em}%
  {\normalfont\bfseries}}
\makeatother

\begin{document}

\newgeometry{left=1.25in, right=1.25in,bottom=1.25in,top=1.25in}
\setstretch{1.25}
\title[Robust Inference]{Robust Inference on Infinite and Growing Dimensional Time Series Regression} 

\date{\today  \  \   \
We thank the co-editor and three referees for many insightful comments that substantially improved the paper, as well as Donald Andrews, Xiaohong Chen, Jean-Jacques Forneron, \v{S}\'{a}rka Hudecov\'{a}, Yuichi Kitamura, Joon Park, Pierre Perron, Zhongjun Qu and Yixiao Sun,  audiences at Yale, Boston University, GOFCP 2019 (Trento), MEG 2019 (Columbus), LAMES 2019 (Puebla), Cambridge, ESWC 2020 and Binghamton. We are grateful to Joel Tropp for guiding us to useful results and Ekaterina Oparina for excellent research assistance. \\
\ $^{\maltese}$ Department of Economics, University of Essex, Wivenhoe Park, Colchester, CO4 3SQ, UK. Email: a.gupta@essex.ac.uk. Research supported by British Academy/Leverhulme
Trust grant SRG$\backslash 170956$.\\
 $^{\S}$ Corresponding author.  
 $^{\dag}$ Department of Economics, Seoul National University, Gwan-Ak Ro 1, Seoul, Korea. Email: myunghseo@snu.ac.kr. Research supported by the Ministry of Education of the Republic of Korea, National Research Foundation of Korea (NRF-2018S1A5A2A01033487) and LG Yonam Foundation, and partly written while
visiting the Cowles Foundation at Yale University, whose hospitality is
greatly acknowledged.
}
\author{Abhimanyu Gupta$^{\maltese} $}
\author{Myung Hwan Seo$^{\S\dag}$ }

\begin{abstract}
        We develop a class of tests for time series models such as multiple regression with growing dimension, infinite-order autoregression and nonparametric sieve regression. Examples include the Chow test and general linear restriction tests of growing rank $p$. Employing such increasing $p$ asymptotics, we introduce a new scale correction to conventional test statistics which accounts for a high-order long-run variance (HLV) that emerges as $ p $ grows with sample size.
        %   \textcolor{red}{DELETE?: It echoes the heteroskedasticity-autocorrelation robust inference procedure in the low-dimensional time series models. }
        We also propose a bias correction via a null-imposed bootstrap to alleviate finite sample bias without sacrificing power unduly.      
        A simulation study shows the importance of robustifying testing procedures against the HLV even when $ p $ is moderate.     
        The tests are illustrated with an application to the oil regressions in \cite{Hamilton2003}.
        
        \bigskip
        
        
\end{abstract}

\maketitle   

%\listoftodos

\setstretch{1.5}

\section{Introduction}

\allowdisplaybreaks

\label{sec:intro}
\allowdisplaybreaks

This paper develops asymptotically valid tests for inference on infinite-order and growing dimensional time series regression models, revealing the presence of an hitherto undetected \emph{nonlinear serial dependence}  or \emph{high-order long-run variance} (HLV) factor. This factor depends on the model error and regressors in a nonlinear fashion, and can appear in limit distributions when the data exhibit dependence and the number of restrictions grows. Chow tests and tests for linear restrictions are both covered. Our theory, simulations and empirical results show the deleterious effect of ignoring the HLV term, and  we propose a testing procedure that is robust to its presence. This is shown to possess desirable finite sample properties. While the HLV factor is revealed by our increasing dimension asymptotics, it can contaminate inference even in multiple regressions with a moderate number of covariates. Such specifications are ubiquitous in practice. Thus, the findings and recommendations of this paper are important for practitioners wishing to make correct inferences when data are dependent.

Models of infinite or growing dimension have been widely studied in the recent econometric literature, reflecting modern applications with rich sets of variables. For example, the asset pricing literature has suggested hundreds of potential risk
factors to explain returns, see \cite{feng2019taming}. With a larger number of observations accumulating over time, it is natural to include more of these variables as covariates even without resorting to penalized estimation methods. In fact, an attitude that permits the number of covariates to grow as a function of sample size is tacitly adopted in the literature. 
In a survey, \cite{koenker1988asymptotic} observed that the number of regressors in empirical work increases as the sample size $n$ increases, roughly like $n^{1/4}$, suggesting that practitioners implicitly treat model complexity as a function of sample size. Finally, nonparametric methods such as  series estimation have found wide applications in the economics and finance literature, see e.g. \cite{jorda2005estimation}, \cite{chen2007}, \cite{Chen2015}. These methods involve the approximation of an infinite-order model with a sequence of growing dimensional models. Taken together, this proliferation of models highlights the importance of developing appropriate techniques for their study.

%\todo{need to make explicit our setting of growing number of restrictions and short discussio on why we do not pursue many restrcitions setting in \cite{kline2020leave}: we can mention that the covariates are not as diverse as in the cross-section to begin with and they often contain their own lagged variables and are more likely to be correlated to each other in the time series setting. Thus, the multicollinearity becomes more problematic as the number of variables grows than in the cross-section.}

Our approach is to develop tests for null hypotheses that involve a growing number of restrictions $p$ in time series regression, with $p$ increasing slower than sample size. 
As a leading example, we consider the Chow test, due to \cite{Chow1960}, to test for a structural break at a prescribed time. This has the advantage of being a simple exclusion restriction test with wide applicability. 
After examining the key issues in this simple context, we present results for the testing of general linear restrictions. This extends specification tests with slowly growing $p$, see e.g. \cite{Hong1995} and \cite{Gupta2018}, to time series regression. However, our testing problem is distinct from the so-called `many restrictions' setting in e.g. \cite{Calhoun2011}, \cite{Anatolyev2012,Anatolyev2019}, \cite{kline2020leave}, amongst others, where the number of restrictions grows proportionally to sample size.

%Allowing for dependent data, we approximate our infinite-dimensional testing problem with a sequence of finite-dimensional problems. In nonparametric regression, our approach is based on series estimation which entails approximating the nonparametric regression function with a linear combination of increasingly many basis functions. Because the sequence of basis functions that underlies series estimation is not of fixed dimension, one cannot simple apply the theory of finite-dimensional parametric problems. Instead, the conventional parametric test statistics must be adjusted to take into account the increasing dimensional nature of the problem. \cite{Hong1995} propose an appropriately standardized test statistic that provides a consistent specification test based on series estimation, while more recently \cite{Gupta2018} provides results that extend the classical trinity of tests to settings with increasing dimension.
%When a finite number of linear restrictions is concerned, a standard chi-squared test is still valid.  When the dimension of the restriction itself is infinite as in e.g. the nonparametric specification testing by Hong and White (1995) or  Gupta (2018), a suitable modification of the chi-squared test has been proposed.  

 We derive the asymptotic distribution of the Chow test Wald statistic centered by $ p $ and normalized by $ \sqrt{2p} $.
This yields asymptotic normality with an unknown asymptotic variance $ \mathcal{V} $, which we term the HLV, provided that $ p $ meets certain growth conditions.  The HLV factor $ \mathcal{V} $ captures high-order autocovariances of the regressors and disturbances, echoing the long-run variance that appears in fixed dimensional time series regression, and vanishes under simplifying assumptions that remove these high-order autocovariances. The new HLV factor $ \mathcal{V} $ does not appear in fixed $ p $ asymptotic regimes, nor does it appear in the independent data setting of  \cite{Hong1995}, who use the same transformation and obtain asymptotic standard normality. 
%We study $ \mathcal{V} $ in Section \ref{sec:V_robust} in greater detail and illustrate that a simple ARCH process in the error can drive $\mathcal{V} $ away from unity.
%Figure \ref{fig:V} reports various values that it can achieve for designs that correspond to the three leading cases that motivate our paper: growing dimensional linear regression, AR($\infty$) models and nonparametric sieve regression. 
%In our Monte Carlo simulation, we demonstrate that using the traditional critical values, thereby ignoring the factor that we discover, can seriously distort the size of the tests. Thus, we develop a correction for the nonlinear serial dependence factor based on the fixed-$b$ approach and provide a concrete testing procedure and recommendations for practitioners. 

We robustify the Chow test against the HLV by a random scaling because of numerical evidence that asymptotic normal inference based on consistent estimators of HLV performs poorly in finite samples, reported in Section \ref{sec:mc}. The random scaling is motivated by heteroskedasticity autocorrelation robust (HAR) inference, see e.g. \cite{kiefer2002heteroskedasticity}, \cite{Sun2014}, and \cite{lazarus2018har}, just to name a few. The resulting asymptotic distribution is pivotal. However, unlike conventional HAR inference, where the standard Wiener process characterizes mixed normality, our limit distribution is represented by two dependent centered Gaussian processes $ W(r) $ and $ \bar{W}(r) $ such that $ EW(r)^2 = r^2 $ and $ E\bar{W}(r)^2 = (1-r)^2 $.   
Although pivotal, the asymptotic distribution depends on the location of the hypothesized break date and thus we provide R code to compute the p-values. Similarly, we robustify the general linear restrictions Wald test to the HLV and provide suitable R code.

Finite sample bias in the Wald statistic, or in quadratic statistics more generally, is a serious issue when $ p $ is large. See e.g. \cite{kline2020leave} for more discussion and a bias correction proposal that works well even when $ p $ is proportional to the sample size but under independent sampling. Our simulations document that the problem is even worse in time series regression. Thus, we propose a bootstrap bias correction which imposes the null hypothesis in the resampling so as not to sacrifice power unduly. Even a small number of bootstrap iterations appear sufficient to reduce the bias, making computation easily manageable.  Based on these findings, we recommend a bias-corrected and HLV-robust test to practitioners.

 In simulations for a range of settings across regression with many covariates, long AR fits  and sieve regression, we demonstrate that our statistic exhibits excellent size control without sacrificing power excessively. Failure to correct for the HLV can seriously affect inference, in general leading to over-rejection and often severely so. Such a pattern is shown to persist for the two types of tests that we provide: Chow tests and exclusion restrictions. In an empirical example based on \cite{Hamilton2003,Hamilton2009}, we show that using our bias-corrected and HLV-robust tests can yield inferences that lead to new conclusions when considering the relation between oil prices and economic activity.

The paper is organized as follows. Section \ref{sec:model} introduces the model and the Chow test, along with some basic assumptions and examples. 
In Section \ref{sec:theory} we provide an asymptotic theory while Section \ref{sec:V_robust} introduces our HLV-robust and bias-corrected test statistic. The testing of general linear restrictions is covered in Section \ref{sec: extension}.
Section \ref{sec:mc} contains a Monte Carlo study of finite sample performance, and
Section \ref{sec:emp examples} demonstrates our test with real data. 
All the proofs of theorems and lemmas are collected in two further appendices, the second of which is available online. Throughout the paper, cross-referenced items prefixed with `S' can be found in this online supplementary appendix.
%, and specifically  Appendix \ref{sec:partial_sum}provides a result on partial sums of random matrices that may be useful more generally. 
An R-package to reproduce the simulations and empirical example is available in the replication files.

\section{Chow Test in Growing and Infinite-Order Regression}
\label{sec:model} 
We consider the issue of testing for a structural break at a known point in the
conditional mean function of $y_{t}$ given the information available up to $%
t-1$, ie $E\left( y_{t}|\mathcal{F}_{t-1}\right) , $
where $\mathcal{F}_{t-1}$ denotes the filtration up to time $t-1$. In
nonparametric regression, $\mathcal{F}_{t-1}$ typically consists of a finite
number of observable covariates $z_{t}$. In the context of the infinite
order autoregressive AR$\left( \infty \right) $ model,  $\mathcal{F}_{t-1}$
is the collection of all the lagged dependent variables, $\left\{y_{t-j}\right\}_{j\geq 1}$. Alternatively, it can be viewed as a genuine high-dimensional regression model
which may contain an infinite number of covariates and their lags. We allow for 
array structure but we do not introduce further notation to denote it unless necessary.

Given a sample of size $n$, we estimate the unknown regression function
via a growing-dimensional (or truncated) linear regression
\begin{equation}
y_{t}=x_{nt}^{\prime }\beta _{n}+e_{nt},  \label{model_n}
\end{equation}%
where $x_{nt}$ and $\beta _{n}$ are $p$-dimensional vectors and $%
p\rightarrow \infty $ as $n\rightarrow \infty $ to estimate $E\left( y_{t}|%
\mathcal{F}_{t-1}\right) $ consistently. To be more precise, let 
\[
\varepsilon _{t}=y_{t}-E\left( y_{t}|\mathcal{F}_{t-1}\right) ,
\] 
$\beta _{n}$
be the best linear predictor of $y_{t}$ given $x_{nt}$, and $r_{nt}=E\left(
y_{t}|\mathcal{F}_{t-1}\right) -x_{nt}^{\prime }\beta _{n}$. Then, $%
e_{nt}=r_{nt}+\varepsilon _{t}$. Throughout the paper, let $C$ ($c$) denote
a generic positive and finite constant, arbitrarily large (small) but independent of $n$, and `a.s.' stand for `almost surely'. Introduce the following assumptions:

\begin{assumption}
\label{ass:errors} The martingale difference sequence $\left\{ \varepsilon
_{t}\right\} $ satisfies  $\sigma _{t}^{2}\leq C,$ where $E\left(
\varepsilon _{t}^{2}|\mathcal{F}_{t-1}\right) =\sigma _{t}^{2},$ and $%
E\left( \varepsilon _{t}^{4}|\mathcal{F}_{t-1}\right) \leq C$, a.s.
\end{assumption}

The theory presented in the paper may not hold if in fact we only have $E\left(x_{nt}\varepsilon_t\right)=0$ as the long run variance of $x_{nt}\varepsilon_t$ will then appear in the type of quadratic statistics that we consider.

\begin{assumption}
\label{ass:aprx0} For $a=1,2$, 
\begin{equation}
\sup_{t}E\left( r_{nt}^{2a}\right) =o\left( n^{-1}\right) .  \label{r_order}
\end{equation}
\end{assumption}

We discuss this assumption on the negligibility of the approximation error in more detail in Section \ref{sec:mc}, where specific examples are introduced. The subscript $n$ will now usually be dropped, although we will emphasize this occasionally to remind the reader of the $n$-dependence of certain quantities. 
%\begin{example}[Sparse High-Dimensional Regression]
%       Certain high dimensional regressions with sparsity comes with a strong oracle procedure such as the SCAD or the adaptive lasso.  
%\end{example}

Introduce a potential structural break for these models at a given time, say $ t=[n \gamma] $, $\gamma \in \Gamma \subset (0,1)$, with  
$\Gamma$ compact and 
$[\cdot]$ denoting the integer part of the argument.  That is, $
\beta = {\beta _{1}\ \ {\text{if }}t/n\leq \gamma }$ and $\beta={\beta _{2}\ \ 
{\text{if }}t/n>\gamma }.
$
We write the model as
\begin{equation}
y_{t} =x_{t}^{\prime }\beta _{1}1\left\{ t/n\leq \gamma \right\}
+x_{t}^{\prime }\beta _{2}1\left\{ t/n>\gamma \right\} +e_{t} 
 =x_{t}^{\prime }\delta _{1}+x_{t}^{\prime }\delta _{2}1\left\{ t/n>\gamma
\right\} +e_{t},\label{our_strucb_model}
\end{equation}
where $\delta _{1}=\beta _{1}$, $\delta _{2}=\beta _{2}-\beta _{1}$, and $1\left\{\cdot\right\}$ denotes the indicator function. Consider the Wald test for the exclusion restriction $ \delta_{2} = 0 $, namely the Chow test for the presence of a structural break at a known date. 

Let $ \hat{\delta }\left( \gamma \right) $ and $ \hat{e}_t\left( \gamma \right) $ denote the OLS estimate and the OLS residuals, respectively, and  $x_{t}\left( \gamma \right) :=\left(x_{t}^{\prime },x_{t}^{\prime }1\left\{ t/n>\gamma \right\} \right) ^{\prime}$.  Also, let
$ 
\hat{M}\left( \gamma \right) ={n}^{-1}\sum_{t=1}^{n}x_{t}\left( \gamma \right) x_{t}\left( \gamma \right) ^{\prime },
$
and $\hat{\Omega}\left( \gamma \right) $ denote an estimator of $%
E\varepsilon _{t}^{2}x_{t}\left( \gamma \right) x_{t}^{\prime }\left( \gamma
\right) $. For instance, $\hat{\Omega}\left( \gamma \right) $ can be set as $n^{-1}\sum_{t=1}^{n}x_{t}\left( \gamma \right) x_{t}\left(
\gamma \right) ^{\prime }\hat{e}_{t}\left( \gamma \right) ^{2}$ (the Eicker-White formula) or, assuming conditional homoskedasticity, it can be
$\hat{\sigma}\left( \gamma \right) ^{2}%
\hat{M}\left( \gamma \right) $, where $\hat{\sigma}^{2}\left( \gamma \right)
=n^{-1}\sum_{t=1}^{n}\hat{e}_{t}\left( \gamma \right) ^{2}$. The choice
depends on the case being considered. 
Then, the Wald statistic for the familiar Chow test is defined as
\begin{equation}
W_{n}\left( \gamma \right) :=n\hat{\delta}_{2}\left( \gamma \right) ^{\prime
}\left( R\hat{M}\left( \gamma \right) ^{-1}\hat{\Omega}\left( \gamma \right) 
\hat{M}\left( \gamma \right) ^{-1}R^{\prime }\right) ^{-1}\hat{\delta}%
_{2}\left( \gamma \right) ,  \label{wald_def}
\end{equation}%
where $R=\left( 0_{p\times p}:I_{p}\right) $ is a selection matrix.

When the dimension $ p $ of $ x_t $ grows with the sample size $ n $, the Wald statistic diverges as it is approximately chi-squared distributed with $p$ degrees of freedom. Thus, a conventional approach, as used e.g. by \cite{DeJong1994} and \cite{Hong1995}  in the cross-sectional (independent data) framework is to introduce a new centering and scaling to define
\begin{equation}
{\mathcal{Q}}_{n}\left (\gamma \right ):=\left({W_{n}\left (\gamma \right )-p%
}\right)/{\sqrt{2p}}, \label{Q_n}
\end{equation}
since the mean and variance of a chi-square distribution with $ p $ degrees of freedom are $ p $ and $ 2p $, respectively. Furthermore, it has been established that the standard normal approximation of $ \mathcal{Q}_n $ is valid in their settings. Subsequent sections investigate how this conventional approach fails in the context of growing or infinite dimensional time series models, mirroring the failure of time series inference procedures without  heteroskedasticity and autocorrelation correction or robustification.

\section{Asymptotic Distribution of $ \mathcal{Q}_n $}

\label{sec:theory} This section provides the
asymptotic distribution of the Chow test statistic under the null and also shows
that the statistic has non-trivial power against local alternatives at an appropriate nonparametric rate.

There has been some recent interest in the so-called many regressor setting
where $p $ is allowed to be proportional to $n $, see e.g. \cite%
{cattaneo2018inference} and \cite{kline2020leave}. 
We do not permit such a large $p $ as our hypothesis of interest concerns a $p $-dimensional restriction and the design matrix of time series data faces more difficulties in satisfying the rank condition. 
In this regard, \cite{Chen2001} provide an interesting example from an
ANOVA design where the weak convergence of the empirical distribution of
residuals from the linear regression with growing dimension fails when the
dimension $p$ is of order $n^{1/3}$. They compare various growth conditions
for $p$ in the literature and conclude that $p^{3}\log ^{2}p=o\left(
n\right) $ is nearly necessary for a general stochastic design.
Heuristically, a\ hypothesis represented through the empirical distribution
function imposes an infinite number of restrictions, like our structural break
testing also does, and valid testing of such a hypothesis demands a tighter control
on the growth rate of $p$.

\subsection{Asymptotic Null Distribution}

Define $\left\Vert A\right\Vert =\left\{ \overline{\lambda }(A^{\prime
}A)\right\} ^{\frac{1}{2}}$ for a generic matrix $A$, where $\underline{%
\lambda }$ (respectively $\overline{\lambda }$) denotes the smallest
(largest) eigenvalue of a symmetric nonnegative definite matrix. 
Any limit stated as `$n\rightarrow\infty$' is taken as both $ n $ and $ p $ grow to infinity simultaneously unless specified otherwise. We also introduce the $p\times p$ non-stochastic matrix sequences $M$ and $ \Omega $ and define 
\begin{equation*}
	M(\gamma )=\left[ 
	\begin{array}{cc}
		M & (1-\gamma )M \\ 
		(1-\gamma )M & (1-\gamma )M%
	\end{array}%
	\right] ,\;\;\Omega (\gamma )=\left[ 
	\begin{array}{cc}
		\Omega & (1-\gamma )\Omega \\ 
		(1-\gamma )\Omega & (1-\gamma )\Omega%
	\end{array}%
	\right] .
\end{equation*}%

\begin{assumption}
\label{ass:M_diff} 
\sloppy (i) $\sup_{i,t}Ex_{ti}^{4}<\infty $. \\
(ii) For $ r\in\Gamma\cup \{1\} $, 
\begin{eqnarray*}
	%\sup_{\gamma\in\Gamma\cup\{1\}}
	\left\Vert n^{-1}\sum_{t=1}^{[nr]} x_{t}x_{t}^{\prime }-rM\right\Vert +
	%\sup_{\gamma\in\Gamma\cup\{1\}}
	\left\Vert n^{-1}\sum_{t=1}^{[nr]}x_{t}x_{t}^{\prime }\sigma _{t}^{2}-r \Omega \right\Vert
	&=&O_{p}\left(
	\varkappa _{p}\right),\\  
%	\sup_{\gamma\in\Gamma}
	\left\Vert \hat{\Omega}\left( r \right)
	-\Omega \left( r \right) \right\Vert &=&O_{p}\left( v_{p}\right) ,\\
	\underline {\lambda }
	\left ({M}\right )> \lambda_n,&\underline {\lambda }
	\left ({\Omega}\right )&> \lambda_n,
\end{eqnarray*}
for some positive sequences of numbers $ \varkappa_{p} $, $v_p$ and $\lambda_n$ satisfying
\begin{equation}
	\lambda_n^{-4} \sqrt{p}\left( \lambda_n^{-1}\varkappa _{p}+v_{p}\right) \rightarrow 0  
	\text{ and }  \lambda_n^6 p \to \infty.  \label{rate:Q_weak_conv}
\end{equation}
(iii)  $\varlimsup _{n\rightarrow \infty }\overline {%
        \lambda }\left (M\right )<\infty $, 
$\varlimsup _{n\rightarrow \infty }%
\overline {\lambda }\left (\Omega \right )<\infty $.

\end{assumption}

Several factors determine the bound $\varkappa _{p}$ for nonparametric series regression. It is proportional to $\sqrt{p/n}$ or $ p/\sqrt{n} $ up to logarithmic
factors with iid data, depending on the choice of basis functions. 
For dependent data, the mixing decay rate also contributes to $ \varkappa_{p} $.      
The exact rate $v_{p}$ depends on a particular example. We formally introduce our examples of multiple linear regression, AR($\infty$) and nonparametric sieve regression in Section \ref{sec:mc}. Primitive conditions and expressions for $\varkappa _{p}$ and $v_p$ are given in Propositions \ref{thm:Berk} and \ref{thm:CC}  in Appendix \ref{sec:primitive}, using the results of \cite{peligrad1982invariance}, \cite{Newey1997}, \cite{gonccalves2007bootstrapping} and \cite{Chen2015}.

%\textcolor{red}{With $\otimes$ denoting Kronecker product, note that 
%\[ 
%M(\gamma)=M\otimes  \left[\begin{array}{cc}
%	1 & 1-\gamma  \\ 
%	1-\gamma  & 1-\gamma 
%\end{array}\right],\;\;\;\Omega(\gamma)=\Omega\otimes  \left[\begin{array}{cc}
%1 & 1-\gamma  \\ 
%1-\gamma  & 1-\gamma 
%\end{array}\right],
%\]
%the second matrix factor above being symmetric and positive definite with eigenvalues $1\pm \sqrt{5\gamma^2-8\gamma+4}/2-\gamma/2$, where $1/\sqrt{5}\leq \sqrt{5\gamma^2-8\gamma+4}/2\leq 1$ on $[0,1)$ with the minimum attained at $\gamma=4/5$ and the maximum at $\gamma=0$. Thus, both eigenvalues are uniformly bounded above and below on the compact $\Gamma$ by positive constants. }

Recall that the eigenvalues of the Kronecker product of two symmetric matrices are the products of their eigenvalues, and $ \gamma $ is bounded away from zero and one. Thus, $M(\gamma)$ and $\Omega(\gamma)$ inherit the eigenvalue restrictions on $M$ and $\Omega$ in Assumption \ref{ass:M_diff} $(ii)$ and $(iii)$, up to positive constants.

%By the Cauchy interlacing theorem (see e.g. \cite{Smith1992}, Theorem 1), conditions (ii) and (iii) in Assumption \ref{ass:M_diff} imply that $\underline{\lambda }\left( {M}\right) > \lambda_n $, $\underline{\lambda }\left( {\Omega}\right) > \lambda_n $ and $\varlimsup_{n\rightarrow \infty }
%\overline{\lambda }\left( M\right) <\infty $, $\varlimsup_{n\rightarrow \infty }
%\overline{\lambda }\left( \Omega\right) <\infty $.  

To develop the distributional limit of ${\mathcal{Q}}_{n}(\gamma )$ where both $ n $ and $ p $
diverge simultaneously, we introduce more conditions. Now, for convenience we let $\xi_{t}=\Omega^{-1/2}x_{t}\varepsilon_{t}$,
$ \mathcal{G}_t $ denote a filtration for $ \xi_{t} $, $\Upsilon_t=E\left(\xi_{t} \xi_{t}'|\mathcal{G}_{t-1}\right)$,  and $\Xi_s=\sum_{t_1=1}^{s-1}\sum_{t_2=1}^{s-1} \xi_{t_1} \xi_{t_2}'$. The filtration $ \mathcal{G}_t $ need not be $ \mathcal{F}_t $ but a simpler one as long as it makes $ \xi_t $ a mds.  Indeed, some conditions may be easier to verify  under simpler filtrations. 
The next assumption introduces the HLV factor $\mathcal{V}$ formally.

\begin{assumption}\label{ass:MCLT}
       \sloppy Suppose that $\max_{1\leq t \leq n}\overline\lambda\left(\Upsilon_t\right)=o_p \left(n^\nu\right) $, for some $\nu\in[0,1/3)$, $\max_{1\leq t\leq n} E((\xi_t '\xi_t)^2|\mathcal{G}_{t-1}) = o_p\left(n^\omega\right)$, for some $\omega\in[0,1-\nu)$, $\sum_{t=1}^n\sum_{s=1}^{t-1}\cov\left(\tr\left(\Upsilon_{t}\Xi_{t}\right),\tr\left(\Upsilon_{s}\Xi_{s}\right)\right)=o(n^{4}p^{2})$,
and there exists $ \mathcal{V} $ such that for $ l=0 $ or $ [n\gamma] $ and for $ m $ that is proportional to $ n $
\begin{equation} 
\lim_{n \rightarrow \infty } \frac{1}{mp}tr\sum_{t_{1}=1}^{m-1}%
\sum_{t_{2}=1}^{m-1 }E\left( \xi_{m+l}\xi_{m+l}^{\prime
}\xi_{t_{1}+l}\xi_{t_{2}+l}^{\prime }\right)  
= {\mathcal{V}}. 
\label{V_partial_def}
\end{equation}%

\end{assumption}

 The first condition can be met if moments of $ \overline{\lambda}\left(\Upsilon_t\right) $ of an order higher than $ 1/\nu $  are bounded for all $ t $. 
The restriction on the summability rate of $\cov\left(\mathrm{tr}\left(\Upsilon_{t}\Xi_{t}\right),\mathrm{tr}\left(\Upsilon_{s}\Xi_{s}\right)\right)$ is related to the dimension $p$. To gain some insight, consider the case where the conditional moment $\Upsilon_{t}$ is homogeneous, so that $\Upsilon_{t}=I_{p}$ for all $t$. 
	% Then, $\cov\left(\mathrm{tr}\left(\Upsilon_{s}\Xi_{s}\right),\mathrm{tr}\left(\Upsilon_{l}\Xi_{l}\right)\right)=E\left(\sum_{t_{1},t_{2}<s}\left(\xi_{t_{1}}'\xi_{t_{2}}-\mu_{t_{1},t_{2}}\right)\sum_{t_{3},t_{4}<l}\left(\xi_{t_{3}}'\xi_{t_{4}}-\mu_{t_{3},t_{4}}\right)\right)$,	where $\mu_{t_{1},t_{2}}=p1\left\{ t_{1}=t_{2}\right\} $ and $\xi_{t}=\Omega^{-1/2}x_{t}\varepsilon_{t}$. If $t_{1}\neq t_{2}$, $t_{3}=t_{4}$, and $t_{1}$ or $t_{2}$ is bigger than $t_{3}$, then $E\xi_{t_{1}}'\xi_{t_{2}}\left(\xi_{t_{3}}'\xi_{t_{3}}-p\right)=0$ 	since $\xi_{t}$ is mds. If $t_{1}\neq t_{2}$, $t_{3}=t_{4}$, and both $t_{1}$ and $t_{2}$ are smaller than $t_{3}$, then $E\xi_{t_{1}}'\xi_{t_{2}}\left(\xi_{t_{3}}'\xi_{t_{3}}-p\right)=0$ 	due to the homogeneity of the conditional moment $E[\xi_{t}'\xi_{t}|\mathcal{F}_{t-1}]=p$. Similarly, 	$E\left(\xi_{t_{1}}'\xi_{t_{1}}-p\right)\left(\xi_{t_{3}}'\xi_{t_{3}}-p\right)=0$ for $t_{1}\neq t_{3}$. On the other hand, $E\xi_{t_{1}}'\xi_{t_{2}}\xi_{t_{3}}'\xi_{t_{4}}\neq 0 $ for $ t_1 = t_3 \neq t_2 = t_4 $. Thus, the covariance reduces to $$\sum_{t_{1}<\min\left(s,l\right)}E\left(\xi_{t_{1}}'\xi_{t_{1}}-p\right)^{2} + \sum_{t_{1}\neq t_2 <\min\left(s,l\right)}E\left(\xi_{t_{1}}'\xi_{t_{2}}\right)^{2}=O\left(n^2 p\right),$$ since $E\left(\xi_{t_{1}}'\xi_{t_{2}}\right)^{2}=p $ under the homogeneity of the conditional moment. 
Then, some tedious algebra yields that $\cov\left(\mathrm{tr}\left(\Upsilon_{t}\Xi_{t}\right),\mathrm{tr}\left(\Upsilon_{s}\Xi_{s}\right)\right)= O(n^2 p) $ uniformly over all $s,t$ with $s < t $. This implies that the double sum of the covariances is $O\left(n^4p\right)$ and thus meets the required condition as $ p\to \infty $.
	Our assumption says that more generally this double sum over covariances must be $o\left(n^4p^2\right)$ as $n,p\rightarrow\infty$.

Also, note that under the special case where $\left\{
x_{t}\varepsilon _{t}\right\} $ is an iid sequence, we have 
$
{\mathcal{V}}=\lim_{m,p\rightarrow \infty } {m}^{-1}\sum_{t=1}^{m-1}%
p^{-1}{\mathrm{tr}}E\left( \xi_{m}\xi_{m}^{\prime }\right) E\left( \xi_{t}\xi_{t}^{\prime }\right) =1,
$
thus ${\mathcal{V}}$ is an extra factor that appears in the limit due to nonlinear dependence in the data. In particular, it captures a {\textit{high-order serial correlation}} of $\xi_{t}$, while $\xi_{t}$ itself does not have serial correlation since it is a martingale difference sequence.

For mean zero random variables $a_{1i},a_{2j},a_{3k},a_{4l}$, let $\mathrm{cum}_{ijkl}\left(a_{1i},a_{2j},a_{3k},a_{4l}\right)$ denote the fourth cumulant. 
	\begin{assumption}\label{ass:autocovandcumulant}
        $\left\{x_{ti}\varepsilon_t\right\}_{t\in\mathbb{Z}}$ is fourth order stationary for all $i=1,\ldots,p$. Furthermore, $\sup_{i,j=1,\ldots,p}\sum_{t=-\infty}^{\infty}\left\vert c_{ij}(t)\right\vert<\infty$, where $c_{ij}(t)=E\left(x_{r,i}\varepsilon_{r}x_{r+t,j}\varepsilon_{r+t}\right)$ for integer $r$, and $\sup_{i,j,k,l=1,\ldots,p}\sum_{t_1,t_2,t_3=-n}^n \left\vert\mathrm{cum}_{ijkl}\left(x_{0,i}\varepsilon_{0},x_{t_1,j}\varepsilon_{t_1},x_{t_2,k}\varepsilon_{t_2},x_{t_3,l}\varepsilon_{t_3}\right)\right\vert=O\left(n^2\right)$.
\end{assumption}
This assumption controls the temporal dependence in $\left\{x_t\varepsilon_t\right\}$ and is discussed in \cite{Andrews1991a}, for example, wherein sufficient conditions for it to hold are also provided. The following theorem establishes distributional convergence for a given $ \gamma $.

\begin{theorem}
        \label{thm:null_Chow} Let Assumptions \ref{ass:errors}- \ref{ass:autocovandcumulant} %\ref{ass:MCLT},
        and $\mathcal{H}_{0}$ hold.
        Then $
        {\mathcal{Q}}_{n}(\gamma )\overset{d}\rightarrow \mathcal{N}(0,\mathcal{V})
        $, for a given $\gamma\in\Gamma$.
\end{theorem}
Theorem \ref{thm:null_Chow} highlights the distinctive feature of testing growing number of restrictions in time series regressions. Unlike the independent cross-sectional case, we have to robustify the test against the HLV term $ \mathcal{V} $. The provenance of this term can be illustrated by some formulae, details of which are contained in the full proofs. These proofs first establish (Theorem \ref{lemma:wald_approx}) that
\begin{equation}\label{WeqRop1}
	{\mathcal{Q}}_{n}\left (\gamma \right ):=\frac{W_n(\gamma)-p}{\sqrt{2p}}=\frac{\mathcal{R}_n(\gamma)-p}{\sqrt{2p}}+o_p(1), 
\end{equation}
where
\begin{eqnarray}
{\mathcal{R}_{n}}(\gamma )=\left[{ \gamma \left( 1-\gamma \right) n}\right]^{-1}
\left\Vert \sum_{t=1}^{[n\gamma ]}\xi_{t}-\gamma
\sum_{t=1}^{n}\xi_{t}\right\Vert^2= \left\Vert n^{-1/2}\sum_{t=1}^n\psi_t(\gamma)\xi_t\right\Vert^2,\label{fancyRdef}
\end{eqnarray}
and $\psi_t(\gamma)=\left({1}\left(t/n\leq \gamma\right)-\gamma\right)/\sqrt{\gamma(1-\gamma)}$. 
Also note that $ n^{-1}\sum_{t=1}^n\psi_t(\gamma)^2 \to 1 $. Thus, just as for the familiar Wald statistic, we have a quadratic form structure for $\mathcal{R}_n(\gamma)$. When $p$ is fixed and there is no approximation error, we note that (\ref{WeqRop1}) has also been established by \cite{Andrews1993}, \cite{Cho2017} and \cite{Sun2021}.  

This then yields the approximation
\begin{equation}\label{RSop1}
\frac{\mathcal{R}_n(\gamma)-p}{\sqrt{2p}}=\mathcal{S}_n(\gamma)+o_p(1),
\end{equation}
where
\[
\mathcal{S}_n(\gamma)=\frac{2}{\sqrt{2p}}\frac{1}{n}\sum_{t=2}^n\psi_t(\gamma)\xi_t'\sum_{s<t}\psi_s(\gamma)\xi_s=\frac{\sqrt{2}}{\sqrt{n}}\sum_{t=2}^n v_t(\gamma),
\]
say, by Lemma \ref{lemma:diag_terms_neg}. Then $\mathcal{V}=\lim_{n,p\rightarrow\infty}2n^{-1}\sum_{s,t=2}^n cov \left(v_s(\gamma),v_t(\gamma)\right)$, i.e. the limiting variance of $\mathcal{S}_n(\gamma)$. Note that the $v_t(\gamma)$ are defined as products of terms of the type $ x_t \varepsilon_t $ and the cumulative sum of their lags, implying that the variances of the $ v_t(\gamma) $ themselves contain high-order covariance terms. This explains why we call $ \mathcal{V} $ a HLV despite the mds property of the $ v_t (\gamma)$, which implies that $\{v_t(\gamma)\}$ is uncorrelated. 

The next section establishes that the test based on $ \mathcal{Q}_n $ has nontrivial local power under suitable sequences of local alternatives, following which we study more detailed characteristics of $ \mathcal{V} $ and develop a HLV-robust test.

\subsection{Local Alternatives}

We consider a sequence of local alternatives that converge to the null at $%
p^{1/4}/\sqrt{n}$-rate to study the local power properties of the test.
This is slower than the usual $1/\sqrt{n}$ parametric rate and has been employed by a number of other authors, e.g. \cite{DeJong1994}, \cite{Hong1995}, \cite{Gupta2018}. 
It is a cost of the growing-dimensional nature of the problem. Our sequence of local alternatives is: 
\begin{equation}
\mathcal{H}_{\ell }:\delta _{2\ell }=2^{1/4}\tau p^{1/4}/\sqrt{n},
\label{local_alternatives}
\end{equation}%
where $\tau $ is a unit length $p\times 1$ vector. 
\begin{theorem}
        \label{thm:local_power} \sloppy Suppose that Assumptions \ref{ass:errors}- \ref{ass:autocovandcumulant} and $\mathcal{H}_{\ell}$ hold and let $\tau _{\infty }=\lim_{n\rightarrow \infty }\tau ^{\prime
        }M\Omega ^{-1}M\tau $. Then, 
        $
        {\mathcal{Q}}_{n}(\gamma )\overset{d}\rightarrow
        \mathcal{N}\left(\tau _{\infty } {\gamma (1-\gamma )},\mathcal{V}\right)
        $.
\end{theorem}
Note that the noncentrality term is positive, implying nontrivial power of the test since the critical region is formed by $ \mathcal{Q}_n (\gamma) $ being greater than equal to a critical value.
Also, $
\left\vert \tau ^{\prime }M\Omega ^{-1}M\tau \right\vert \leq \left\Vert
\tau \right\Vert \left\Vert M\right\Vert ^{2}\left\Vert \Omega
^{-1}\right\Vert =\overline{\lambda }(M)^{2}/\underline{\lambda }(\Omega )
$
for any $n$. 
As Assumption \ref{ass:M_diff} assumes that the numerator $ \bar{\lambda}(M) $ is bounded but the denominator may not be bounded away from zero,  $\tau _{\infty }=\lim_{n\rightarrow \infty }\tau ^{\prime}M\Omega ^{-1}M\tau$ may diverge to positive infinity to imply more power. 

\section{$\mathcal{V}$ Robust Testing} \label{sec:V_robust}
%\todo{unify name for V, nonlinear serial dependence; high-order autocorrelation; long-run covariance ? }
In this section we provide a detailed study of the HLV $\mathcal{V}$ that our analysis has discovered. In particular, we present some alternative representations of $\mathcal{V}$ that shed more light on its structure. 
%We also provide some simulated values of $\mathcal{V}$ for the examples in Section \ref{sec:model} with various parametrizations and different degrees of temporal dependence. 

\subsection{Discussion}
\sloppy We first examine the relevance of $\mathcal{V}$. Specifically, we analyze the `pre-limiting' quantity 
$
\mathcal{V}_n  =2var\left(n^{-1}\sum_{t=2}^{n}\xi_{t}'p^{-1/2}\text{\ensuremath{\sum_{s=1}^{t-1}\xi_{s}}}\right)$. 
This can be rewritten as 
\[ 
\mathcal{V}_n =2n^{-1}\sum_{i=1}^{n-1}\left(\gamma\left(i,0\right)(n-i)/{n}+2\sum_{j=1}^{n-i}\gamma\left(i,j\right)(n-i-j)/{n}\right),
\] 
where 
\[
\gamma\left(i,j\right)=p^{-1}E\left(\xi_{t}'\xi_{t-i}\xi_{t}'\xi_{t-i-j}\right)=p^{-1}E\left(x_{t}'\Omega^{-1}\text{\ensuremath{x_{t-i}}}x_{t}'\Omega^{-1}\text{\ensuremath{x_{t-i-j}\varepsilon_{t}^2\varepsilon_{t-i}\varepsilon_{t-i-j}}}\right).
\]
This is a high-order autocovariance and captures a nonlinear serial
dependence in the sequence $ x_t \varepsilon_t $, which disappears
entirely for $ j>0 $ in independent cross sectional data. We encounter $\mathcal{V}_n\rightarrow \mathcal{V}\neq 1$ when $n^{-1}\sum_{i=1}^{n-1}\sum_{j=1}^{n-i}\gamma\left(i,j\right)(n-i-j)/{n}$ has a nonzero limit, with terms arising that are fourth-order cross-moments of the $\varepsilon_{t}$ and $ x_t $. Thus, the behaviour of such cross-moments is the key to obtaining non-unity $\mathcal{V}$. \cite{Robinson1991}, studying time series specification testing, encountered 
a term of the form $ E\left(\varepsilon_t^2 \varepsilon_{t-i} \varepsilon_{t-i-j} \right)$, somewhat different from ours albeit also of cross-moment type, but imposed conditions that nullify it when $ j>0 $. 

The Wald statistic is a quadratic form in the moment process. To establish the limit of the Wald statistic when the number of variables (i.e., the number of moments) grows with the sample size, we need to account for the variance of the quadratic form, hence the appearance of  fourth-order dependence of a certain type in the moment process. A form of fourth-order dependence has also been encountered in HAR testing, see e.g. \cite{Lobato2002}. In Section \ref{sec:mc}, we present some figures to show how $\mathcal{V}$ can vary for various designs and deviate significantly from unity.

\subsection{HLV-Robust Test Statistic} \label{sec:HAC}

This section propose a random scaling approach to robustify our test statistic against the unknown HLV term $\mathcal{V}$. We opt for this because our numerical experiments in Section \ref{sec:mc} (specifically Figure \ref{fig:size_HAC} and its discussion therein) reveal poor finite sample performance of the standard sample variance of $ q_t = \left( np\right)^{-1/2} x_{t}'\hat{\Omega}^{-1} \hat{e}_{t} \text{\ensuremath{\sum_{s=1}^{t-1}x_{s}\hat{e}_{s}}} $. The presence of the cumulative sums $ \sum_{s=1}^{t-1} $ and the estimated quantities $ \hat{\Omega}^{-1}$ and $ \hat{e}_{s} $  in the construction of $ q_t $ are likely to contribute to poor finite sample behavior of its sample variance. 

The random scaling approach has been employed when consistent estimators of the asymptotic variance perform poorly in finite samples. For instance, heteroskedasticity and autocorrelation consistent (HAC) estimators, e.g. \cite{newey1987simple}, \cite{Andrews1991a}, to name but two examples, have been followed by the fixed-bandwidth kernel approach to obtain an asymptotically pivotal and mixed-normal test, see e.g. \cite{kiefer2000simple} and \cite{lazarus2018har} for a recent review. 
In the machine learning literature, \cite{Lee_Liao_Seo_Shin_2022} also employ the random scaling approach for computationally efficient on-line inference based on the stochastic gradient descent algorithm. 
 
While a simple random scaling can be implemented by the integral of the square of the partial sum process of centered $ q_t $, 
that is, $ n^{-2} \sum_{t=1}^n \sum_{s=1}^{t} (q_t- n^{-1}\sum_{i=1}^n q_i )^2$, 
we present a class of more general random scaling methods following the heteroskedasticity and autocorrelation robust (HAR) inference literature. 
We note that the resulting pivotal distributions differ from the HAR literature, however. 

Introduce a kernel function $ k(\cdot) $ that meets the following conditions. 
\begin{assumption}
        \label{ass:kernel} (1) For all $x\in {\mathbb{R}}$, $k(x)=k(-x)$ and $%
        \left
        \vert k(x)\right \vert \leq 1$; $k(0)=1$; $k(x)$ is continuous at
        zero and almost everywhere on ${\mathbb{R}}$; $\int _{{\mathbb{R}}%
        }\left
        \vert k(x)\right \vert dx<\infty $. (2)
        For
        any $b\in(0,1]$ and $\rho\geq1$, $k_{b}\left(x\right)=k\left(x/b\right)$
        and $k^{\rho}\left(x\right)$ are symmetric, continuous, piecewise
        monotonic, and piecewise continuously differentiable on $\left[-1,1\right]$.
        (3) $\int_{[0,\infty )}\bar{k}%
        (x)<\infty $, where $\bar{k}(x)=\sup_{y\geq x}\left\vert k(y)\right\vert $.
\end{assumption}

% To motivate our HAR correction it is convenient to rewrite $\Gamma\left(i,j\right)=\left(np\right)^{-1}E\left[\varepsilon_{i}x_{i}'\Omega^{-1}\sum_{t=i+1}^{n-1}\sum_{s=j+1}^{n-1}x_{t}\varepsilon_{t}x_{s}'\Omega^{-1}\varepsilon_{s}x_{j}\varepsilon_{j}\right],$ using the martingale property of the $x_t\varepsilon_t$. We can then write $ \mathcal{V}=\lim_{n\to\infty}var\left(n^{-1}p^{-1/2}\sum_{t=2}^{n}x_{t}'\Omega^{-1}\varepsilon_{t}\text{\ensuremath{\sum_{s=1}^{t-1}x_{s}\varepsilon_{s}}}\right)=\lim_{n\to\infty}n^{-1}\sum_{i=1}^{n}\sum_{j=1}^{n}\Gamma\left(i,j\right) $.

% For a kernel function $k(\cdot)$ and bandwidth $h$, we now introduce a  pseudo-estimate of $\Gamma\left(i,j\right)$, $ \tilde{\Gamma}\left(i,j\right)=\varepsilon_{i}x_{i}'\Omega^{-1}\left[\sum_{t=i+1}^{n-1}\sum_{s=j+1}^{n-1}k\left(\frac{t-s}{n/h}\right)x_{t}\varepsilon_{t}\varepsilon_{s}x_{s}'\Omega^{-1}\right]x_{j}\varepsilon_{j}, $  and that of $\mathcal{V}$, $ \tilde{\mathcal{V}}=\frac{1}{n}\sum_{i=1}^{n}\sum_{j=1}^{n}\tilde{\Gamma}\left(i,j\right).  $   Exchanging the order of the summations yields $ \tilde{\mathcal{V}}=n^{-1}\text{\ensuremath{\sum_{t=1}^{n}\text{\ensuremath{\sum_{s=1}^{n}k\left(\frac{t-s}{n/h}\right)q_{t}q_{s}}}}}, $ where $q_{t}=\left(np\right)^{-1/2}x_{t}'\Omega^{-1}\varepsilon_{t}\text{\ensuremath{\sum_{s=1}^{t-1}x_{s}\varepsilon_{s}}}$. 

Since $\varepsilon_{t}$ and $\Omega$ are not directly observable in practice, we replace them with the least squares estimates as in Section \ref{sec:model} and introduce $ q_t = \left( np\right)^{-1/2} x_{t}'\hat{\Omega}^{-1} \hat{e}_{t} \text{\ensuremath{\sum_{s=1}^{t-1}x_{s}\hat{e}_{s}}} $ and its demeaned version,  $\bar{q}_{t}= q_t - n^{-1}\sum_{t=2}^{n} q_t$. Then, define a feasible estimate of $\mathcal{V}$ by 
\begin{equation} \label{eq:Vhat}
\hat{\mathcal{V}}=\frac{2}{n}\text{\ensuremath{\sum_{t=2}^{n}\text{\ensuremath{\sum_{s=2}^{n}k\left(\frac{t-s}{n b}\right)\bar{q}_{s}}}\bar{q}_{t}}}.
\end{equation}
Thus, we have a seemingly long-run variance estimate, analogous to traditional HAC/HAR inference, of a nonlinear transformation of the primitive variables.

The choice of bandwidth $ b $ has been a topic of much discussion in the
HAC literature. Since  $\mathcal{V}$ captures high-order autocovariances in the growing dimensional vector $x_{t}\varepsilon_{t}$,
the finite sample variation in the estimate $\hat{\mathcal{V}}$ is generally
larger than in more familiar long-run variances, and the moment
condition is more expensive. Motivated by this, we follow a fixed bandwidth approach, as in \cite{Sun2014}. 

Our estimator is based on the weighting function $ K_{h}\left(r,s\right) = k\left(h\left(r-s\right)\right)$, where $ h=1/b $. We present numerical results in this paper with  $k\left(u\right)=\left(1-\left|u\right|\right)^{h}1\left\{ \left|u\right|<1\right\} $, employing the Bartlett kernel case with $h=1$. \cite{Sun2014} terms this the sharp kernel estimator. Other options include the steep quadratic kernel estimator and the orthonormal series estimator with $K$ basis function, of which \cite{Sun2014} contains a more detailed discussion. \cite{Sun2014} also shows that the centering in $ \bar{q}_t $ can be conveniently represented through a centered version of $ K_h(\cdot) $, that is, $K_{h}^{*}\left(r,s\right)  =K_{h}\left(r,s\right)-\int_{0}^{1}K_{h}\left(\tau,s\right)d\tau-\int_{0}^{1}K_{h}\left(r,\tau\right)d\tau +\int_{0}^{1}\int_{0}^{1}K_{h}\left(\tau_{1},\tau_{2}\right)d\tau_{1}d\tau_{2} $. 

%Following \cite{Sun2014}, we may write \[ \hat{\mathcal{V}}  =\frac{2}{n}\text{\ensuremath{\sum_{t=1}^{n}\text{\ensuremath{\sum_{s=1}^{n}Q_{h}^{*}\left(\frac{t}{n},\frac{s}{n}\right)q_{s}}}q_{t}}}\\  =\sum_{j=1}^{\infty}\lambda_{j}\left(\frac{1}{\sqrt{n}}\ensuremath{\sum_{t=1}^{n}q_{t}}\Phi_{j}\left(\frac{t}{n}\right)\right)^{2}, \] where  $Q_{h}^{*}\left(r,s\right)  =Q_{h}\left(r,s\right)-\int_{0}^{1}Q_{h}\left(\tau,s\right)d\tau-\int_{0}^{1}Q_{h}\left(r,\tau\right)d\tau +\int_{0}^{1}\int_{0}^{1}Q_{h}\left(\tau_{1},\tau_{2}\right)d\tau_{1}d\tau_{2}, $
Building on the representation in Lemma 1 of \cite{Sun2014}, where the estimate $ \hat{\mathcal{V}} $ is not consistent, we characterize the joint weak limit of $ \hat{\mathcal{V}} $ and $ \mathcal{Q}_n (\gamma) $. For real numbers $a$ and $b$, let $a\vee b$ ($a\wedge b$) denote their maximum (minimum), and introduce a process 
\begin{equation}  \label{Q_gamma_def}
\mathcal{Q}(\gamma )=\frac{W(\gamma )}{\gamma }+
\frac{\bar { W}\left (\gamma \right )}{\left (1-\gamma \right )}-W(1),
\end{equation}
where $\left (W\left (r \right ),\bar {W}\left (r
\right
)\right
)^{\prime }$, $r\in[0,1]$, is a bivariate Gaussian process that does not depend on any model parameters including the break point $\gamma$, and has covariance kernel 
\begin{equation}  \label{C_def}
{\mathcal{C}}\left (r _{1},r _{2}\right )=\left ( 
\begin{array}{cc}
\left (r _{1}\wedge r_{2}\right )^{2} & 1\left \{r
_{1}>r _{2}\right \}\left (r _{1}-r _{2}\right )^{2} \\ 
1\left \{r _{1}<r _{2}\right \}\left (r _{1}-r
_{2}\right )^{2} & \left (1-\left (r _{1}\vee r _{2}\right )\right
)^{2}%
\end{array}
\right ).
\end{equation}
For any given $\gamma\in\Gamma$, the marginal distribution of $ \mathcal{Q}(\gamma ) $ is  standard normal. Thus, the conclusion of Theorem \ref{thm:null_Chow} can be expressed as $\mathcal{Q}_n(\gamma)\overset{d}{\rightarrow}\sqrt{\mathcal{V}}\mathcal{Q}(\gamma)$, pointwise in $\gamma\in\Gamma$.  By taking a suitable ratio, we obtain a pivotal variable as in the following theorem, which is the basis of our test statistic.

\begin{theorem}\label{thm:fixedbVdist}
Let Assumptions \ref{ass:errors}-\ref{ass:kernel} hold, together with 
\begin{equation}
 \lambda_n^{-2}p\left( v _{p}+ 
\frac{p}{\sqrt{n}}\right) \rightarrow 0{\text{ as }}n\rightarrow \infty .
\label{rate:Vhat}
\end{equation}
Under $\mathcal{H}_0$, we have $ \hat{\mathcal{V}}\overset{d}\rightarrow \mathcal{V}\int_{0}^{1}\int_{0}^{1}K_{h}^{*}\left(r,s\right)dW\left(r\right)dW\left(s\right) $ and 
\[ 
\mathcal{T}_n(\gamma):=
\frac{ \mathcal{Q}_n (\gamma)}{\sqrt{\hat{\mathcal{V}}}} 
\overset{d}\rightarrow \frac{\mathcal{Q}(\gamma)}{\sqrt{\int_{0}^{1}\int_{0}^{1}K_{h}^{*}\left(r,s\right)dW\left(r\right)dW\left(s\right)}}. \]
The numerator in the limit becomes $ \mathcal{Q}(\gamma) +\tau _{\infty } {\gamma (1-\gamma )} $ under $ \mathcal{H}_{\ell}$.
\end{theorem}
The asymptotic null distribution is mixed normal and pivotal. The critical values can be tabulated for each $ \gamma $ via Monte Carlo simulation and the replication files provide R code. Note that the same Gaussian process $W(\cdot)$ occurs in both the limiting numerator and denominator, and this process is different from the Brownian motion in \cite{Sun2014}. In fact, it can be represented by the partial sum of $ \sqrt{t/n}$ times an iid normal sequence. Since the limit also involves another variable $ \bar{W}(\cdot) $, the critical values will be different from those previously tabulated in the literature.

%\textcolor{red}{The noncentrality term under $ \mathcal{H}_{\ell} $ is positive since $ \tau_{\infty} $ is quadratic and $ \gamma $ lies within $ (0,1) $. As we reject the null if the test statistic is larger than a given critical value, the test is nontrivial. }

\subsection{Bias Correction} \label{sec:Bias} The degrees of freedom $p$  provide a correct centering for $ W_n (\gamma) $ in first order asymptotic analysis. However, in the finite sample experiments given in Section \ref{sec:mc}, e.g. Figure \ref{fig:TSbiasAR} and Figure \ref{fig:TSWaldbiasmult}, we find that the bias in $ \mathcal{Q}_n (\gamma) $ gets bigger for typical values of $ p $ in nonparametric regression. 
Therefore, we propose a bootstrap bias correction of $ \mathcal{Q}_n (\gamma) $. To estimate the bias, we implement the \textit{null-imposed} wild bootstrap by generating 
\begin{equation} \label{eq:y^star}
y_t^{\star} = x_t ' \hat{\delta}_1 (\gamma) + \hat{e}_t(\gamma) u_t , \quad t=1,...,n,
\end{equation}
where $ u_t $ is an iid sequence of centered and normalized variables, e.g. the Rademacher variables, to compute $ \mathcal{Q}_n^{\star} (\gamma) $. It is worthwhile to note that the bootstrap DGP \eqref{eq:y^star} imposes the null hypothesis $ \delta_2 =0, $ so as not to sacrifice the power of the test.
See also \cite{gonccalves2007bootstrapping} for a thorough discussion on the wild bootstrap for infinite order autoregression. 
Iterating this $ B $ times, we obtain $  \bar{\mathcal{Q}}_n^{\star} (\gamma) = B^{-1} \sum_j^{B} \mathcal{Q}_n^{\star, j} (\gamma)  $, the bootstrap estimate of the bias. 
In our experiment, $ B=200 $ suffices and thus the bootstrap is not computationally expensive. Therefore, we suggest the following bias corrected test statistic:
\begin{equation}\label{eq:Tnb}
\mathcal{T}_n^b (\gamma) := 
\frac{ \mathcal{Q}_n (\gamma) - \bar{\mathcal{Q}}_n^{\star} (\gamma) }{\sqrt{\hat{\mathcal{V}}}}.
\end{equation}
The numerical experiments in Section \ref{sec:mc} show that the bootstrap bias corrected test controls the type I error reliably without sacrificing power unduly.\footnote{It is worth noting  that the wild bootstrap may not be valid to approximate the quantiles of $ \mathcal{Q}_n (\gamma) $ as it does not capture the high-order dependence embodied in $ \mathcal{V} $.} Now, with the superscript $ \star $ indicating the bootstrap analogue, we have the following result. 
\begin{theorem}\label{thm:bootstrap} Under Assumptions \ref{ass:errors}-\ref{ass:M_diff} and $ \mathcal{H}_0 $,
	\begin{equation}
		%\sup_{\gamma \in \Gamma}
		\left\vert	E^{\star}W_{n}^{\star}(\gamma )-p \right\vert = o_p\left(p^{1/2}\right), \label{bootthmorig}
	\end{equation}
and
	\begin{equation} 
\mathcal{T}^b_n(\gamma)\overset{d}\rightarrow \frac{\mathcal{Q}(\gamma)}{\sqrt{\int_{0}^{1}\int_{0}^{1}K_{h}^{*}\left(r,s\right)dW\left(r\right)dW\left(s\right)}}. \label{bootthmnew}
	\end{equation}
\end{theorem}
	Theorem \ref{thm:bootstrap} implies that the bootstrap bias correction is first-order correct	but does not imply any higher-order improvement. We demonstrate its merits not analytically but numerically in Section \ref{sec:mc}, which is common with the wild bootstrap, see e.g. \cite{gonccalves2004bootstrapping} and references therein. Details of the components of $W_{n}^{\star}(\gamma )$ are left to Section \ref{sec:bootstrap_proof} of the online supplement.

% % % % % % % % % % % % % % % % % % % % % % % % % % % % % % % % % % % % % % % % % % % % % % 
% % % % % % % % % % % % % % % % % % % % % % % % % % % % % % % % % % % % % % % % % % % % % % 

\section{Wald test for general linear restrictions of growing rank}\label{sec: extension}

For a linear regression model $y_t=x_t' \beta+\varepsilon_t$, we consider testing linear restrictions  $\mathcal{H}_0^e:R^e\beta=r $, where $R^e$ is a matrix of rank $p\leq dim(\beta)$. For the usual Wald statistic 
\begin{equation}
W^e_{n}:=n\left(R^e\hat\beta-r\right) ^{\prime
}\left( R^e \hat{M} ^{-1}\hat{\Omega}
\hat{M}^{-1}R'^e \right) ^{-1}\left(R^e\hat\beta-r\right)
,  \label{wald_def_excl}
\end{equation}
 $\hat{M}=n^{-1}\sum_{t}x_tx_t'$ and $\hat\Omega$ is an estimator of $E\varepsilon_t^2x_tx_t'$, define ${\mathcal{Q}}_{n}^e:=\left({W_{n}^e-p}\right)/{\sqrt{2p}}$. Although the test statistic appears to be very similar to the Chow test,  the next theorem shows that the numerator and denominator in our corrected test statistic are related differently, calling for different critical values.  Furthermore, the HLV is now obtained by replacing $\Omega^{-1}$ in Assumption \ref{ass:MCLT} with $L=M^{-1}R^{e\prime}\left(R^e M^{-1}\Omega M^{-1}R^{e\prime}\right)^{-1}R^eM^{-1}$, and the resulting limit denoted $\mathcal{V}^e$.

To estimate $\mathcal{V}^e$ and employ bootstrap bias correction, it is convenient to reformulate the restriction as an exclusion restriction of growing dimension, without loss of generality. Indeed,  let $ S $ be the orthogonal complement of $ R^e $ and $ Q=(S, R^e) $. Then, let $ \mathring{x}_t = Q'^{-1}x_t $, $ \delta = Q \beta - (0',r')' $ such that $ \delta = (\delta_1',\delta_2')' $, with $\mathring x_t=\left(\mathring x_{1t}',\mathring x_{2t}'\right)'$ conformably partitioned, $ \delta_2 = R^e \beta - r $ and $ \tilde{y}_t = y_t - \mathring x_{2t} ' r $. We can now test the null hypothesis $ \mathcal{H}_0^e : \delta_2 = 0 $ in the regression of $ \tilde{y}_t $ on $ \mathring x_t $.

This transformation makes it particularly convenient to impose the null in the bootstrap resampling at the bias correction stage.
 Let $\bar{\mathcal{Q}}_n^\star$ denote the bootstrap bias correction factor for $W_n^e$. This yields the bias-corrected statistic $\mathcal{T}_n^{e,b} =    {\left(\mathcal{Q}_n^e  - \bar{\mathcal{Q}}_n^{\star}\right) }/{        \sqrt{\hat{\mathcal{V}^e}}}$,  where $\hat{\mathcal V}^e$ is defined analogously to $\hat{\mathcal{V}}$, but now with  
$ q_t = \left( np\right)^{-1/2} \tilde{x}_{2t} '\hat{\Omega}^{e-1} \hat{e}_{t}
\text{\ensuremath{\sum_{s=1}^{t-1}\tilde{x}_{2t} \hat{e}_{s}}} $, where $ \tilde{x}_{2t} $ denotes the residuals from the regression of $\mathring x_{2t} $ on $\mathring x_{1t} $ and $\hat\Omega^e=n^{-1}\sum_{t}\tilde{x}_{2t} \tilde{x}_{2t} '\hat{e}_t^2$. Then, with $ W(\cdot) $ defined in \eqref{C_def}, we have the following theorem: 
\begin{theorem}\label{thm:exclusion_thm}
        Let Assumptions \ref{ass:errors}-\ref{ass:kernel} hold with the following modifications: (1) $L$ replacing $\Omega^{-1}$ in Assumption \ref{ass:MCLT} and the resulting limit denoted $\mathcal{V}^e$, (2) The conditions in Assumption \ref{ass:M_diff}$(ii)$ hold for $ r=1 $. Also suppose that         \eqref{rate:Vhat} holds. 
        % $ p\left( v_{p}+\frac{p}{\sqrt{n}}\right) \rightarrow 0        $. 
        Then, under $\mathcal{H}_{0}^e$, 
               \begin{equation}
                        \mathcal{T}_n^{e,b} =    \frac{\mathcal{Q}_n^e-\bar{\mathcal{Q}}_n^\star   }{        \sqrt{\hat{\mathcal{V}^e}}} \overset{d}{\to} \frac{W (1)}{\sqrt{\int_{0}^{1}\int_{0}^{1}K_{h}^{*}\left(r,s\right)dW\left(r\right)dW\left(s\right)}}.
                \end{equation}
                Under $\mathcal{H}_{\ell}^e:R^e\beta-r=2^{1/4}\tau p^{1/4}/\sqrt{n}$,  the numerator in the limit becomes 
        $ W (1) +\tau _{\infty }^e$, where $\tau _{\infty }^e=\lim_{n\rightarrow \infty }\tau ^{\prime}\left(R^e M^{-1}\Omega M^{-1}R^{e\prime}\right)^{-1}\tau $.
\end{theorem}

The limiting distribution is mixed normal and pivotal but different from the limit in Theorem \ref{thm:null_Chow}. This is because the Chow test considers a quadratic form in $n^{-1/2}\sum_{t=1}^n\psi_t(\gamma)\xi_t$, which differs from this section by introducing a trend into the regressors via the factor $\psi_t(\gamma)$. Due to this difference, the partial sum processes converge to Gaussian processes with different covariance kernels. An R code to compute the critical values is available in the replication files.

\section{Monte Carlo Experiments}\label{sec:mc}

This section examines the finite sample properties of our bias corrected HLV-robust test $ \mathcal{T}_n^b $ compared to the standard chi-square test $ W_n $, which does not account for growing $ p $, and the unscaled $\mathcal{Q}_{n}$ statistic with standard normal critical values, which does not account for $ \mathcal{V} $, in terms of bias, size and power. 

We will consider the examples below in our Monte Carlo experiments. In Appendix \ref{sec:primitive} we check our assumptions for these settings.

\begin{E1}[Multiple Regression of Growing Dimension]
	% Stock Return Predictive Regression ? Chen and Hong, others. 
	\cite{koenker1988asymptotic} found through his metastudy that it is common practice in
	econometrics to increase the number of regressors as the sample size $n$ grows,
	at a rate of roughly $O\left(n^{1/4}\right) $. In this case, the approximation error $r_t$ is not explicitly modeled and may be set as zero. Practitioners thus adopt a flexible approach to modelling, where the assumed model becomes richer with more covariates and with more lagged terms to account for the dynamic effect in the spirit of the \textit{distributed lag model}, as illustrated in e.g. \cite{Stock2015introduction}. 
\end{E1}

\begin{E2}[Infinite-Order Autoregression]
	This model is one of the most fundamental models in time
	series analysis, see e.g. \cite{brockwell1991time} or \cite{Hamilton2020time}. For the process to be stationary, the coefficients $\left\{
	b _{j}\right\} $ in the AR$\left( \infty \right) $ model 
	$
	y_{t}=b _{0}+\sum_{j=1}^{\infty }b _{j}y_{t-j}+\varepsilon _{t}
	$
	are assumed to obey a certain decay rate. Specifically, the tail sum of the
	coefficients satisfies Assumption \ref{ass:aprx0} if  $\sum_{j=p}^{\infty}\left\vert b_{j}\right\vert =o\left( n^{-1/2}\right) $. While we take $p$ as given in our analysis, for practical purposes various methods based on information criteria are available to choose the truncation lag $p$, see e.g. \cite{Shibata1980} and references therein. \cite{wang2007regression} propose a lasso-based autoregressive order selection rule while \cite{lee2018asymptotic} propose a lag selection rule in an infinite order panel autoregression. 
	For expositional ease, we assume that the observations begin from $ t=1-p $ and $ x_1 = (1,y_0,...,y_{1-p}) $.
	% and $ \varepsilon_t $ is a stationary martingale difference sequence. 
\end{E2}

\begin{E3}[Nonparametric Series Regression]
	In case of the nonparametric series least squares estimation of $E\left( y_{t}|z_{t}\right) $,  there exists a sequence of transformations of the covariates $z_{t}$ given by 
	$
	x_{nt}:=x_{n}\left( z_{t}\right) :{\mathbb{R}}^{k}\mapsto {\mathbb{R}}^{p},
	$
	and coefficients $\beta _{n}$ such that 
	$
	E\left( y_{t}|\mathcal{F}_{t-1}\right) =f\left( z_{t}\right) =x_{nt}^{\prime
	}\beta _{n}+r_{n}\left( z_{t}\right),
	$
	where $r_{nt}=r_{n}\left( z_{t}\right) $ meets Assumption \ref{ass:aprx0} for a broad class of functions $ f $, see e.g. \cite{Andrews1991}, \cite%
	{Newey1997}, \cite{chen2007} and \cite{Lee2016}. 
	By Lemma 1 of \cite{Lee2016}, it is met if 
	$ |r_t|_{\infty} = O(p^\alpha) $ for some $ \alpha < 0 $ and $p^{2\alpha }\leq
	n^{-1}$.  Depending on the smoothness of the
	nonparametric function $f(\cdot )$, the regressor support dimension $k$, and
	the type of basis functions used, different values of $\alpha $ may be
	implied, see e.g. \cite{Newey1997}, \cite{chen2007}, p. 5573, for examples
	and further references. 
	Often, the condition (\ref{r_order}) holds under the so-called undersmoothing selection
	of $p$. Another closely related example is the partially linear regression model, e.g. \cite{engle1986semiparametric} and \cite{robinson1988root}. Again, while we do not consider data-dependent $p$, for practical purposes the literature proposes methods for the choice of $p$ using cross validation or information criteria, see e.g. p. 5623 of \cite{chen2007} for a list of references.
\end{E3}

 The tests are applied to the setting of the Chow test and testing general linear restrictions. 
We consider various sample sizes $ n $ and dimensions $ p $ from the three examples, E1-E3, with the error generated from a bounded ARCH process 
\begin{equation}\label{eq:ARCH}
\varepsilon_{t}  =\sigma_{t}\eta_{t},\;\;\;\;\;\;\;\;\;\;\;\;\;\sigma_{t}^{2}  =\left(1-\alpha\right)+\alpha \phi ( \varepsilon_{t-1} ) ,
\end{equation}
where $ \phi(x)  = x^2 1\{|x|\leq c\} +c^2 1\{|x|> c\} $, $\eta_{t}=\left(u_{t}-Eu_{t}\right)/\sqrt{var\left(u_{t}\right)}$,
and $\left\{ u_{t}\right\} $ is an iid sequence from the \cite{Marron1992} normal mixture distributions of type 1-3, which we refer to as error 1, 2 and 3. Their error 1 is standard normal. For a standard normal vector $\left(Z_{1},...,Z_{k}\right)$ and multinomial vector $\left(d_{1},...,d_{k}\right)$ with probability $\left(1/5,1/5,3/5\right)$, error 2 skewed unimodal variate is $u_{t}=Z_{1}d_{1}+\left(2Z_{2}/3+1/2\right)d_{2}+d_{3}\left(5Z_{3}/9+13/12\right)$, while error 3 strongly skewed variate is $u_{t}=\sum_{l=0}^{7}d_{l+1}\left(Z_{l+1}\left(2/3\right)^{l}+3\left(\left(2/3\right)^{l}-1\right)\right)$ with equally likely $ d_i $'s. 
We report results using \eqref{eq:ARCH} with $\alpha\in\left\{ 0.3,0.4,0.5,0.55,0.57\right\} $ and $ c = 2.5 $. Results from $ c=3 $ and $ \infty $ are similar and omitted.

More specifically, for multiple regression, E1, the regressors $ x_t $ consist of independent AR(1) processes with coefficient $ \alpha_x $ and ARCH innovation as in \eqref{eq:ARCH} and their lags of order up to 3. That is, we consider the distributed lag model with a growing number of variables. The first five elements of $\beta$ in \eqref{model_n} are set as $d_0\left(5^{-1/2},...,5^{-1/2}\right)p^{1/4}n^{-1/2}$ and the others as zeros. When there is a break, all the values become zero after the break so that the value $ d_0 $ controls the magnitude of the change. We vary  $ p \in \{5,9,13\} $ to examine the effect of the dimension on our tests.

For the infinite order AR regression, E2, we generate the sample from the MA(1) model $ y_{t}=\varepsilon_{t}+\theta1\left\{ t\leq\mu\right\} \varepsilon_{t-1} $, with $\mu=n$ for the size experiment and $\mu=[n\gamma]$ for the power evaluation, and estimate the AR($ p $) model with $p=9$ for $ n = 250 $ and $ p = 13 $ for $ n = 500 $.
For the sieve regression, E3, we consider two variables $ \zeta_{1t} $ and $ \zeta_{2t} $ and their lags $ \zeta_{1,t-1} $ and $ \zeta_{2,t-1} $ as regressors, denoted by $ z_{1t},\cdots,z_{4t}  $, after transforming them as $2\arctan\left(\zeta_{it}\right)/\pi$. Each $\zeta_{it}$ follows an AR(1) process with ARCH error. The regression function is set as $f\left(z_{1},\cdots,z_{4}\right) =d_0\left(1,z_{1},\cdots,z_{4},z_{1}^{2},...,z_{4}^{p_{1}}\right)\left(1^{-2},...,p_2^{-2}\right)' +\sqrt{\left|z_{1}\right|/n}$ with $p_{1}=\left\lfloor n^{1/4}\right\rfloor $ and $ p_2=1+4(p_1-1) $.  To estimate the regression function, we construct $ x_t $ from polynomial basis functions and its dimension $ p=p_2 $.

We first employ these DGPs to simulate pre-limit values of $ \mathcal{V} $ in \eqref{V_partial_def} with $ n=m=500, l=0 $ and $ p $ as described above for each case, which are plotted in Figure \ref{fig:V}, reporting averages from 10000 iterations. This serves as a useful illustration to observe visually that pre-limiting $ \mathcal{V} $ deviates from unity for various specifications. A broad observation we make is that the deviation is bigger with larger ARCH coefficients and bigger autocorrelation in $ x_t $, although this feature is not always monotone. To conclude, we observe that the nonlinear serial autocorrelation factor can induce serious distortion in inference without a suitable robustifying treatment, as we provide in section \ref{sec:HAC}.  

Also, Appendix \ref{sec:primitive} gives the verification of the high-level conditions in Assumptions 1-4 for these examples. 

\subsection{Chow Test} We consider three candidate break points as proportions of the sample sizes, $ \gamma \in \{0.2,0.3,0.5\} $. We begin by examining the bias of $ \mathcal{Q}_n (\gamma) $, conventionally centered by the degrees of freedom $ p $, under the null hypothesis. Note that a severe bias in $ \mathcal{Q}_n (\gamma) $ also implies that the size of the Wald test $ W_n (\gamma) $ can be distorted severely. We report the results in Figure \ref{fig:TSbiasmult}, in which the line with dot markers shows the bias in $\mathcal{Q}_n(\gamma)$ for $n=250$ and $n=500$. For E1, (Figures \ref{fig:Chowbiasmult250}, \ref{fig:Chowbiasmult500}), each vertical partition (marked by a dotted vertical line) corresponds to a specific value of $p$. Within each vertical partition the  DGP parameters change along the horizontal axes as $(\text{error type}, \alpha)$, in lexicographic order. As $p$ grows, we observe that the $\mathcal{Q}_{n}\left(\gamma\right)$ statistic exhibits severe finite sample bias for all values of the DGP parameters.

A similar visualization of bias in $\mathcal{Q}_n(\gamma)$ for E2 is presented in Figures \ref{fig:ChowbiasARMA250}, \ref{fig:ChowbiasARMA500}. Rather than report values for different $p$, here we focus on the case $p=9$ for $ n = 250 $ and $ p = 13 $ for $ n = 500 $ and allow the values of $\alpha$ and $\theta$ to vary along the horizontal axis lexicographically as $(\text{error type}, \alpha\text{ or }\theta)$, as detailed in the caption. A substantial bias in $\mathcal{Q}_n(\gamma)$ is observed for all cases, regardless of $n=250$ or $n=500$, albeit the biases are generally smaller in the latter case. Finally, Figures \ref{fig:ChowbiasSieve250}, \ref{fig:ChowbiasSieve500} show the bias in $\mathcal{Q}_n(\gamma)$ for E3, with the same $ p $ as for E2,  to mimic the asymptotic regime of a sieve regression, and parameters as in E2. We observe a similar pattern of substantial bias for both sample sizes.

As discussed above, Figures \ref{fig:Chowbiasmult250}-\ref{fig:ChowbiasSieve500} clearly show that the biases present in $\mathcal{Q}_{n}\left(\gamma\right)$
are severe. In these figures we also plot the bias of the bias-corrected HLV-robust statistic $\mathcal{T}^b_n(\gamma)$, shown in black with square markers. The bootstrap bias correction seems to work well for
all the cases, substantially alleviating bias. In Figure \ref{fig:TSbiasSieve}, we observe that $\mathcal{T}^b_n(\gamma)$ can still exhibit some bias for specific cases but for E1 and E3, unlike the bias of $\mathcal{Q}_n(\gamma)$, this is centered around zero, while for E2 it is generally smaller in absolute value. Thus we recommend the use of the bootstrap bias correction in practice especially when faced with large values of $p$.

We now study the finite sample rejection frequencies of four competing
tests: $\mathcal{T}_{n}^{b}(\gamma),\mathcal{T}_{n}(\gamma),\mathcal{Q}_{n}(\gamma),$ and $W_{n}(\gamma)$, with specific parameter values as given in the respective figure captions. As shown earlier, the unknown HLV scaling factor $\mathcal{V}$ varies along different ARCH
parameters. This motivates our approach of experimenting with different $\alpha$ values and innovations. The Monte Carlo sizes resulting from the experiment are plotted in Figure \ref{fig:TSsizemult}, wherein we place a horizontal dotted line to mark the nominal size of 5\%. We report results for $\gamma=0.3$. The vertical partitions in each panel of Figure \ref{fig:TSsizemult} correspond, as discussed earlier, to increasing values of $p$ from left to right in E1. We cover multiple regression (Figures \ref{fig:Chowsizemult250gm3}, \ref{fig:Chowsizemult500gm3}), AR fits (Figures \ref{fig:ChowsizeARMA250gm3}, \ref{fig:ChowsizeARMA500gm3}) and sieve regression (Figures \ref{fig:ChowsizeSieve250gm3}, \ref{fig:ChowsizeSieve500gm3}) for $n=250, 500$.

For all DGPs, the usual Wald statistic $W_n(\gamma)$ (diamond markers) over-rejects. Simply standardizing the test statistic $W_n(\gamma)$ to $\mathcal{Q}_n(\gamma)$, hence ignoring the HLV $\mathcal{V}$, does not improve matters. In fact, it usually worsens the problem of over-rejection. This can be seen in the lines with triangle markers. Our HLV-robust statistic $\mathcal{T}_n(\gamma)$ does much better, as the lines with dot markers indicate. While this shows the importance of the correction for $\mathcal{V}$ that we stress in the paper, there is still a tendency to over-reject. On the other hand, applying the bootstrap bias correction and using the bias corrected HLV-robust statistic $\mathcal{T}_n^b(\gamma)$ achieves excellent size control, as can be seen in the line with square markers. The discussion holds regardless of whether $n=250$ or $n=500$. Thus the importance of our proposed testing procedure is clearly visible.

We now analyze the power features of the competing test statistics for the proposed DGPs, allowing for breaks of different magnitudes and setting $\gamma=0.5$. After the break all the coefficients become zero so that the values of $d_0$ govern the size of the breaks in E1 and E3, while the  values of $\theta$ do so for E2. The power performance is plotted in Figure \ref{fig:TSpowerARMAandSieve}, where to conserve space we report results only for $n=400$. Again, we use $p=9$ for E2 and E3, while a range of $p$ is employed for E1. The line marker schemes for each of the competing tests are as described earlier. Examining the figure, the power of our HLV-robust statistics $\mathcal{T}_n(\gamma)$ (dots) and $\mathcal{T}_n^b(\gamma)$ (squares) tracks that of the uncorrected ones as the break size increases for both E2 (center panel) and E3 (right panel). For E1 (left panel), we only report results for $d_0=2$ for clarity. We observe that $W_n(\gamma)$ tends to have the highest power but our statistics still perform reasonably well with power in excess of 80\% even for large $p$. Recall that our size experiments earlier indicate that $W_n(\gamma)$ over-rejects, a phenomenon of which high power is likely an artefact. Thus we conclude that our test is able to control size without sacrificing power to an undue extent.

Finally, Figure \ref{fig:size_HAC} reports the size distortions when the sample variance of $ q_t $ multiplied by 2 is employed instead of $\hat{\mathcal{V}} $, noting that $ q_t $ is a martingale difference array. While we discussed the potential reasons for this severe size distortion in Section \ref{sec:HAC}, the investigation on a more precise  approximation to the finite sample distribution of this statistic is an interesting issue but out of the scope of this paper due to the complex nature of the statistic. Nevertheless, Figure  \ref{fig:size_HAC} provides numerical evidence that the typical scaling by variance is not sufficient to control size. Specifically, define the test statistic $\mathcal{T}_n^{b,2}(\gamma)$ exactly like $\mathcal{T}_n^{b}(\gamma)$ but with random scaling replaced by twice the variance of $q_t$ and the standard normal approximation. It is clear that the test (triangle markers) is oversized relative to our recommended test $\mathcal{T}_n^{b}(\gamma)$. 

\subsection{Testing Linear Restrictions}\label{sec:extension}
This section presents the outcomes of bias, size and power experiments for testing general linear restrictions, analogous to those for the Chow test in the preceding discussion. We use the reparameterization of the linear restrictions to the exclusion restrictions ${\delta}_2=0$, as discussed in Section \ref{sec: extension}. We focus on E1 with $n=400$, $p=8,12,16$, $d_0=1$, error 1 and 2 disturbances and $\alpha=0.4,0.55$. The results are displayed in Figure \ref{fig:TSWaldbiasmult}, with the same marking scheme as before and three test statistics employed: $W_n^e$, $\mathcal{Q}_n^e$ and $\mathcal{T}_n^{e,b}$. In all three figures, each vertical partition marks a different value of $p$, increasing from left to right.

The left panel of Figure \ref{fig:TSWaldbiasmult} shows that the bootstrap bias correction indeed improves matters, as was the case for the Chow test. The center panel again demonstrates the importance of our proposed corrections for size control. $W_n^e$ and $\mathcal{Q}_n^e$ tend to over-reject, becoming worse as $p$ increases. $\mathcal{T}_n^{e,b}$ controls size very well for medium to large $p$, while still outperforming $W_n^e$ and $\mathcal{Q}_n^e$ for smaller $p$. The right panel shows that $\mathcal{T}_n^{e,b}$ sacrifices some power relative to $W_n^e$ and $\mathcal{Q}_n^e$, but not unduly so.

\section{Empirical example}

\label{sec:emp examples} 

We revisit structural stability in the \cite{Hamilton2003} study of the effect of oil shocks on economic activities. The autoregressive distributed lag model, ADL$(p,p)$, with quarterly
time series of outputs and several oil price measures is employed. For real output, the quarterly growth rate of chain-weighted real GDP is used, while the oil price is the nominal crude oil producer price index, seasonally unadjusted. As in Hamilton, three oil price measures were considered: the growth rate $o_{t}$ from the previous quarter, the rectified linear unit, $o_{t}^{+}=o_{t}1\left\{ o_{t}>0\right\} $, and the net oil price increase, $o_{t}^{n}$, defined as the amount by which log oil prices in quarter $t$ exceed their peak value over the previous 12 months. If it does not exceed the previous peak, then $o_t^n$ is taken to be zero. We extend the original sample using the FRED database at the St. Louis Fed to obtain a sample from January 1949 to October 2019. 

First, we re-evaluate structural stability of the GDP dynamics using AR$(p)$ fits, and that of the regression function of GDP growth on oil price change using an ADL$(p,p)$ model with the three alternative measures of oil price change. Following Table 4 in \cite{Hamilton2003}, we investigate four exogenous disruptions in world petroleum supply. These are: the Arab–Israel War (November 1973), the Iranian Revolution (November 1978), the Iran–Iraq War (October 1980)
and the Persian Gulf War (August 1990). 

The p-values of the tests are reported in sub-tables (a) and (b) in Table \ref{table:examplestrucbreakGDP}, where for $ W(\gamma)$ these are computed using the ordinary Chow test. We observe that in many cases the usual Chow test supports a structural break in both regressions more strongly than our recommended $\mathcal{T}^b_n(\gamma)$ test.  Thus, the evidence for structural instability is often no longer as strong. In fact, the p-values that we calculate using our test exceed those of the standard Chow test in 22 out of 32 cases. 
\begin{table}[ht]
	\centering
	\small{        \begin{tabular}{c|cc|cc|cc|cc|c|cc|cc}
			&\multicolumn{8}{l|}{Stability test}   &\multicolumn{5}{l}{Exclusion test} \\
			&\multicolumn{2}{l}{(a) AR($p$)}&\multicolumn{6}{l|}{(b) ADL($p,p$) } &\multicolumn{5}{l}{(c) ADL($p,p$) }\\
			\midrule        
			\textbf{GDP}&&&\multicolumn{2}{c}{$o_t$}&\multicolumn{2}{c}{$o_t^+$}&\multicolumn{2}{c}{$o_t^n$}& \multicolumn{3}{|c}{\ \ \ all oil}&\multicolumn{2}{c}{NL}\\
			\midrule
			lags $p$ & 4 & 6 &  4 & 6 & 4 & 6 & 4 &6&& 4 &6& 4 &6\\ 
			\midrule 
			&\multicolumn{8}{c|}{Arab-Israel War, November 1973}\\
			\midrule
			$\mathcal{T}_n^{b}(\gamma)$ & 51.3	&43.9&	3.5	&0.52	&1.64	&1.33&	1.58	&0
			&$\mathcal{T}_n^{e,b}$&8.7&26.9&7.6&19 \\ 
			${W}_n(\gamma)$ & 52.9&	38	&0.5	&0	&0.58&	0&	0.7	&0
			 &$W_n^e$&9&4.4&2.8&4.4\\
			\midrule
			&\multicolumn{8}{c|}{Iranian Revolution, November 1978}\\
			\midrule 
			$\mathcal{T}_n^{b}(\gamma)$ & 41.8	&52.8	&1.04	&0.03&	7	&0.1	&0	&0
			 \\ 
			${W}_n(\gamma)$ & 48.5	&53.2&	0.6&	0&	7.57	&0	&0.1	&0	
			\\
			\midrule
			&\multicolumn{8}{c|}{Iran-Iraq War, October 1980}\\
			\midrule 
			$\mathcal{T}_n^{b}(\gamma)$ & 13.1	&28.2	&0.3	&0.4&	0.33	&0.64	&0.33	&0
			 \\ 
			${W}_n(\gamma)$ & 10.6	&19.7	&0	&0&	0.1	&0	&0	&0
			\\
			\midrule
			&\multicolumn{8}{c|}{Persian Gulf War, August 1990}\\
			\midrule 
			$\mathcal{T}_n^{b}(\gamma)$ & 54.6	&53.6&	73.3	&43.7&	17.8&	31.2	&4.53&	4.85
			 \\ 
			${W}_n(\gamma)$ &56.2	&65	&39.5	&16	&17.1	&14.4&	1.9	&1.17
			\\
			\bottomrule \end{tabular}}
	\caption{$ 100 \times $p-values of Chow tests and exclusion restriction tests for full sample. (a) Tests for stability of GDP dynamics via AR$(p)$ fits. (b) Tests for stability of ADL$(p,p)$ regressions of GDP on $o_t$, $o_t^+$ or $o_t^n$. (c) Tests for exclusion restrictions on all oil price measures ($o_t$, $o_t^+$, $o_t^n$) or nonlinear oil price measures ($o_t^+$, $o_t^n$.) in  ADL$(p,p)$ regressions of GDP on oil prices. \label{table:examplestrucbreakGDP}}
\end{table}
\begin{table}[ht]
        \centering
\small{        \begin{tabular}{c|cc|cc|cc|cc|c|cc|cc}
			&\multicolumn{8}{l|}{Stability test}   &\multicolumn{5}{l}{Exclusion test} \\
&\multicolumn{2}{l}{(a) AR($p$)}&\multicolumn{6}{l|}{(b) ADL($p,p$) } &\multicolumn{5}{l}{(c) ADL($p,p$) }\\
\midrule
\textbf{IP}&&&\multicolumn{2}{c}{$o_t$}&\multicolumn{2}{c}{$o_t^+$}&\multicolumn{2}{c}{$o_t^n$}& \multicolumn{3}{|c}{\ \ \ all oil}&\multicolumn{2}{c}{NL}\\
\midrule
lags $p$ & 12 & 18 &  12 & 18 & 12 & 18 & 12 &18&& 12 &18& 12 &18\\ 
 \midrule 
  &\multicolumn{8}{c|}{Arab-Israel War, November 1973}\\
  \midrule
$\mathcal{T}_n^{b}(\gamma)$  & 20.9	&32.9	&33.9	&28.4	&24.9	&18.9	&10.6	&17.8
&$\mathcal{T}_n^{e,b}$&27&30.9&11.9&18.5 \\ 
${W}_n(\gamma)$ & 8.46	&26.4	&5.18	&1.71	&3.93&	0.78	&1.23&	1.65
&$W_n^e$&1.7&0.8&0.5&0.5\\
\midrule
&\multicolumn{8}{c|}{Iranian Revolution, November 1978}\\
\midrule 
$\mathcal{T}_n^{b}(\gamma)$ & 11	&21.5&	1.39	&0.74	&0.7	&0.26&	0.15&	0.04
\\ 
${W}_n(\gamma)$ & 0.96	&4.3	&0&	0&	0&	0&	0	&0
\\
\midrule
&\multicolumn{8}{c|}{Iran-Iraq War, October 1980}\\
\midrule 
$\mathcal{T}_n^{b}(\gamma)$ &3.31	&9.34	&1.96	&1.87	&1.15	&1.68&	0.08	&0.46
 \\ 
${W}_n(\gamma)$ & 0 & 0.04&0&0&0&0&0&0\\
\midrule
&\multicolumn{8}{c|}{Persian Gulf War, August 1990}\\
\midrule 
$\mathcal{T}_n^{b}(\gamma)$ & 5.82	&14.6	&3.95	&19.4	&2.03	&12.8&	0.94	&8.04
 \\ 
${W}_n(\gamma)$ & 0.08	&1.13	&0.21	&2.98	&0.23	&2.43	&0.02	&0.22
\\
\midrule
 \bottomrule \end{tabular}}
 \caption{$ 100 \times $p-values of Chow tests and exclusion restriction tests for full sample. (a) Tests for stability of IP dynamics via AR$(p)$ fits. (b) Tests for stability of ADL$(p,p)$ regressions of IP on $o_t$, $o_t^+$ or $o_t^n$. (c) Tests for exclusion restrictions on all oil price measures ($o_t$, $o_t^+$, $o_t^n$) or nonlinear oil price measures ($o_t^+$, $o_t^n$.) in  ADL$(p,p)$ regressions of IP on oil prices. \label{table:examplestrucbreakIP}}
\end{table}

Second, we explore the relevance of the oil price measures and of the nonlinear transformations ($o_t^+$, $o_t^n$) by testing two exclusion restrictions in the ADL$(p,p)$ regression that include all the three oil price measures as covariates. The first exclusion restriction is to set the coefficients of all the measures as zero and the second is to set those of the nonlinear transformations ($o_t^+$, $o_t^n$) to zero. This yields 12 and 8 df, respectively, when $ p=4 $ and 18 and 12 df, respectively, when $ p=6 $. As shown in sub-table (c) of Table \ref{table:examplestrucbreakGDP}, our recommended test $T^{e,b}_n$ produces p-values bigger than 5\% for all cases, suggesting the effect of oil price as measured by these transformations is not statistically significant, nor are the nonlinear transformations. The standard Wald test for the exclusion restrictions is more supportive of their inclusion but may lack robustness with large df. Overall, our tests produce larger p-values than the standard Chow or Wald tests in 72\% of cases (26 out of 36) cases in Table \ref{table:examplestrucbreakGDP}.

As another measure of economic activity we now consider the industrial production (IP) index. This is available at monthly frequency and thus we consider ADL(12,12) and ADL(18,18) to include lags of one year and one and a half years, respectively. 
With monthly data, the dimensionality becomes more important: 
the number of restrictions we test varies from 13 and 25 in the structural break test for the AR(12) and ADL(12,12) regressions to 36 and 48 for the exclusion tests in the ADL(18,18) regression.  

The results in Table \ref{table:examplestrucbreakIP} illustrate much stronger differences in the conclusions of our test versus the standard Chow test, compared to the GDP study in Table \ref{table:examplestrucbreakGDP}. The rejection of the null of no structural break is now often overturned at reasonable significance levels. For example, for the Arab-Israel War, the ADL$(p,p)$ model fails to reject the null of no structural break for any significance level below 10.6\% when using our test $\mathcal{T}^b_n(\gamma)$, in contrast to the standard Chow test ${W}_n(\gamma)$. Conclusions are likewise overturned for the AR$(p)$ model and the Iranian Revolution and Iran-Iraq War, and indeed for both the ADL$(p,p)$ and AR$(p)$ in several cases for the Persian Gulf War.

The results in Tables \ref{table:examplestrucbreakGDP} and \ref{table:examplestrucbreakIP} are not surprising given our simulation evidence. Indeed, our Monte Carlo simulation illustrates the effect of the degrees of freedom (df) on finite sample properties of the two tests, $ W_n(\gamma)$ tends to have larger p values in the AR case ($p+1$ df) than in the ADL case ($2p+1$ df), while $\mathcal{T}^b_n(\gamma)$ would be the opposite. These are exactly the patterns that we also observe. Finally, sub-table (c) of Table \ref{table:examplestrucbreakIP} shows even stronger differences than Table \ref{table:examplestrucbreakGDP}(c), with all conclusions on the exclusion restrictions overturned at significance levels below 11.6\%. Overall, our tests produce larger p-values in every single case considered in Table \ref{table:examplestrucbreakIP} when compared to the standard Chow or Wald tests.

%Building on studies on major oil shocks such as \cite{Hamilton2009}, we also consider two subsamples: SS1 starts after the 1980 oil shock while SS2 ends before the 2007 oil shock. As for the lag order $ p $, we try both $ p=4 $ and $ p=6 $ for robustness. Furthermore, we also test subsample time series and regression stability as well as exclusion restrictions for the dependent variable IP, using $p=12$ and $p=18$ as before. The p-values are reported in Table \ref{table:examplesubsamp}, with the each panel covering GDP and IP and containing subpanels for SS1 and SS2. Again, we note that evidence against the null becomes weaker when using our recommended tests $\sup \mathcal{T}^b_n(\gamma)$ or $\mathcal{T}_n^{e,b}$, as compared to the standard Wald statistic approach. Very often the conclusion of the test is changed when using our approach, but even when this is not the case there can be large differences in p-values.

\bibliographystyle{chicago}
\setstretch{1}
\bibliography{AGmasterref}

\newpage \newgeometry{left=.25in, right=.25in,bottom=.05in,top=.05in}
% censored/figR3/

\allowdisplaybreaks 
\begin{figure}
	\centering 
	\begin{subfigure}{0.3\textwidth}
		\includegraphics[width=\linewidth]{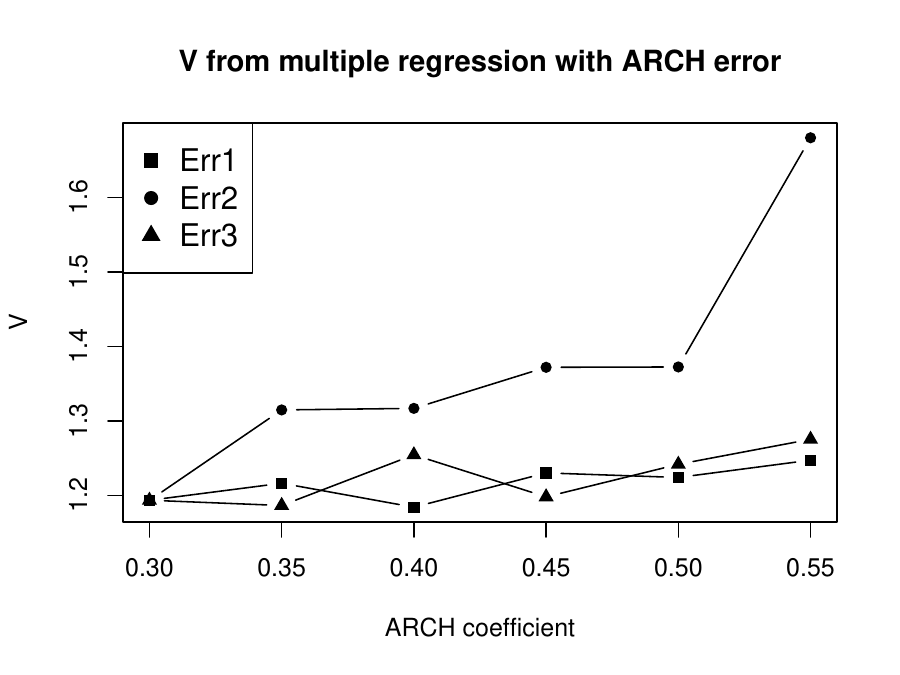}
		\subcaption{\tiny{E1 with $p=12, \alpha_x =0.7 $}\label{fig:Vmult}}
	\end{subfigure}
	\begin{subfigure}{0.3\textwidth}
		\includegraphics[width=\linewidth]{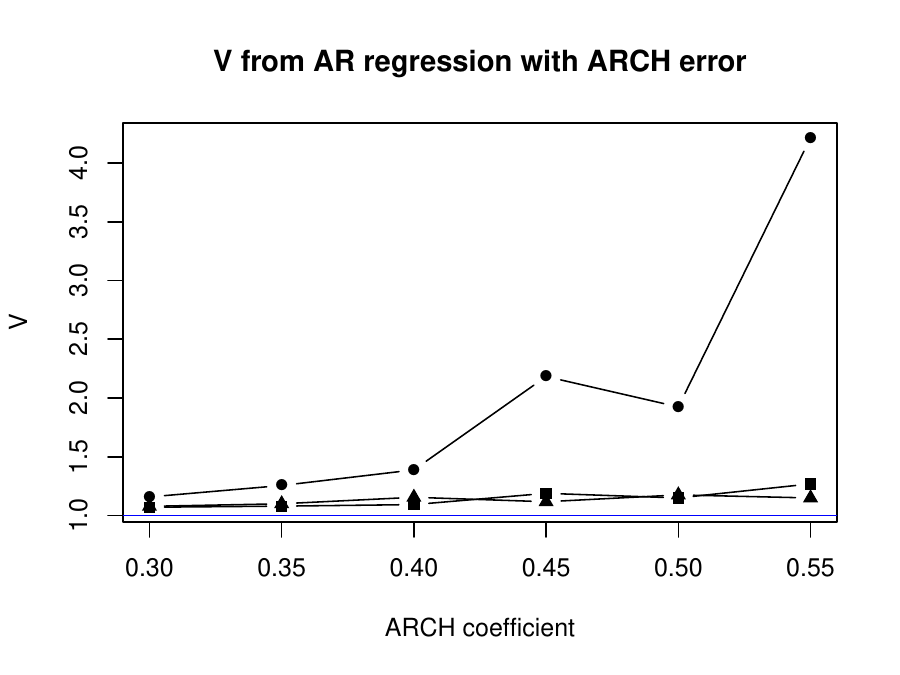}
		\subcaption{\tiny{E2 with $p=12, \theta=-0.5$.}\label{fig:Var1}}
	\end{subfigure}
	\begin{subfigure}{0.3\textwidth}
		\includegraphics[width=\linewidth]{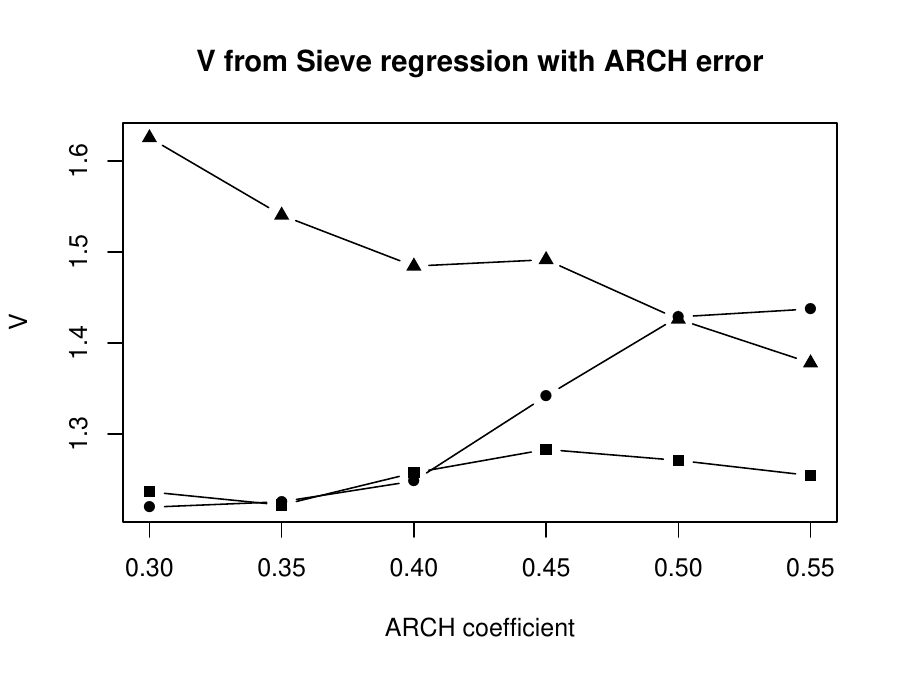}
		\subcaption{\tiny{E3 with $p=13, \alpha_x =0.7$}\label{fig:Var5}}
	\end{subfigure}
	\caption{\scriptsize{Simulated pre-limit of $\mathcal{V}$ for $ n=m=500 $ and $l=0$. Error 1: square; Error 2: dot; Error 3: triangle.}} \label{fig:V}
\end{figure}

\begin{center}
	\begin{figure}
		\begin{subfigure}{0.3\textwidth}
			\includegraphics[width=\linewidth]{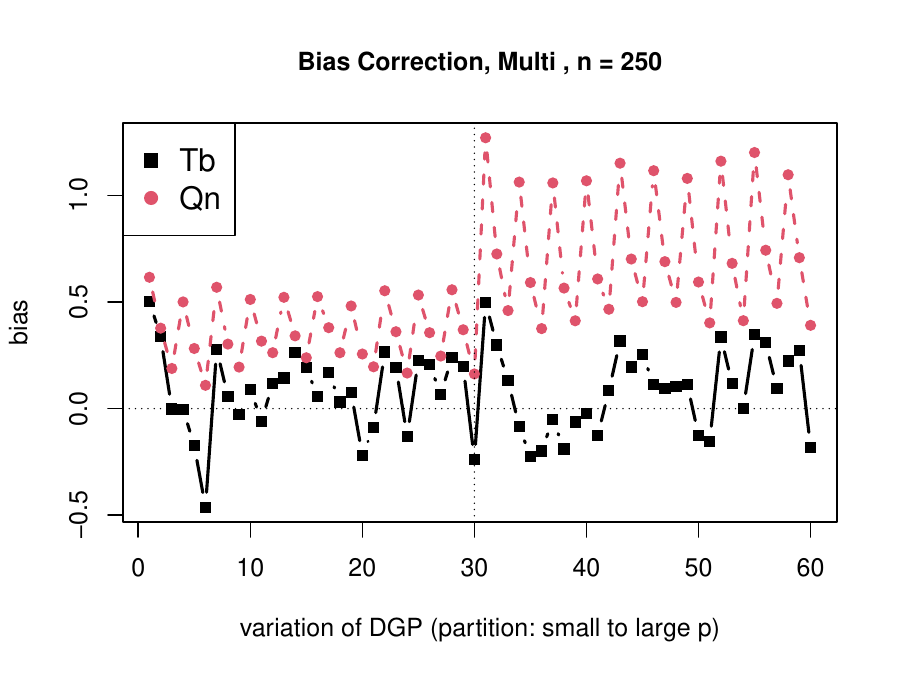}
			\subcaption{\tiny{Chow test bias, E1: $n=250$.}\label{fig:Chowbiasmult250}}
		\end{subfigure}
		\begin{subfigure}{0.3\textwidth}
			\includegraphics[width=\linewidth]{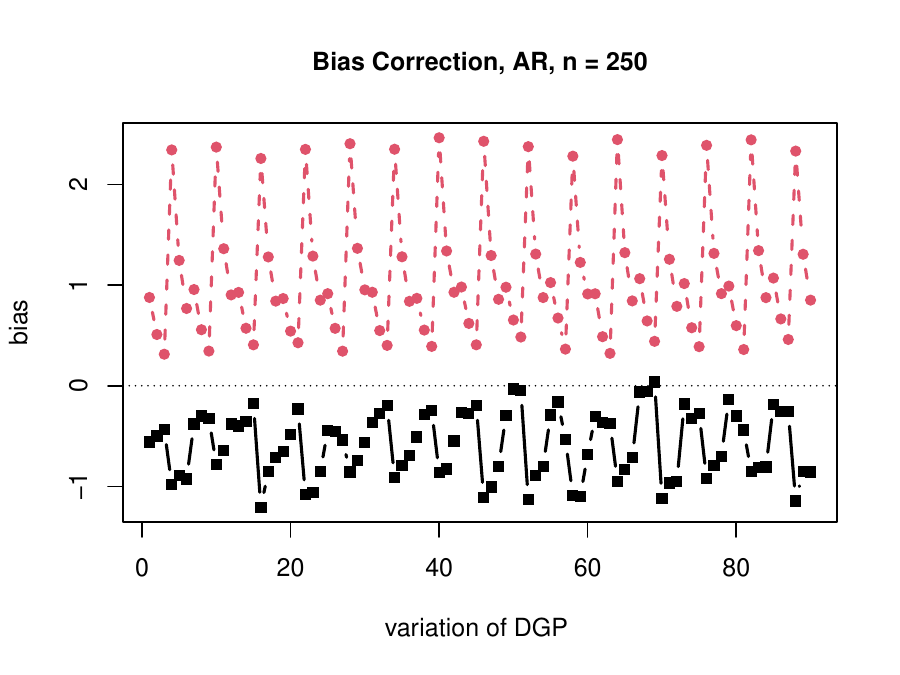}
			\subcaption{\tiny{Chow test bias, E2: $n=250$.}\label{fig:ChowbiasARMA250}}
		\end{subfigure}
		\begin{subfigure}{0.3\textwidth}
			\includegraphics[width=\linewidth]{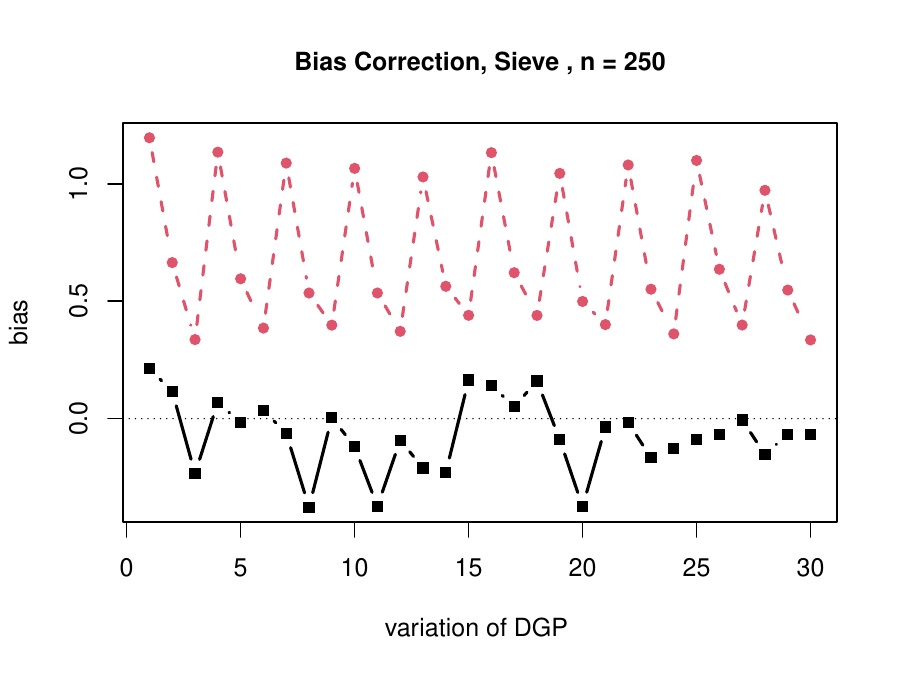}
			\subcaption{\tiny{Chow test bias, E3: $n=250$.}\label{fig:ChowbiasSieve250}}
		\end{subfigure}
		\begin{subfigure}{0.3\textwidth}
			\includegraphics[width=\linewidth]{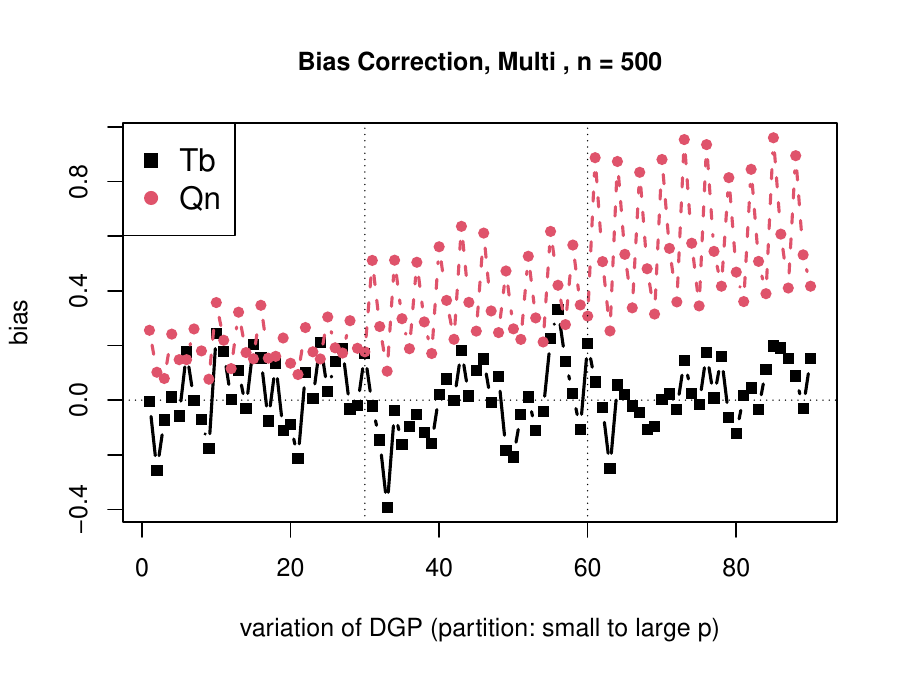}
			\subcaption{\tiny{Chow test bias, E1: $n=500$.}\label{fig:Chowbiasmult500}}
		\end{subfigure}
		\begin{subfigure}{0.3\textwidth}
			\includegraphics[width=\linewidth]{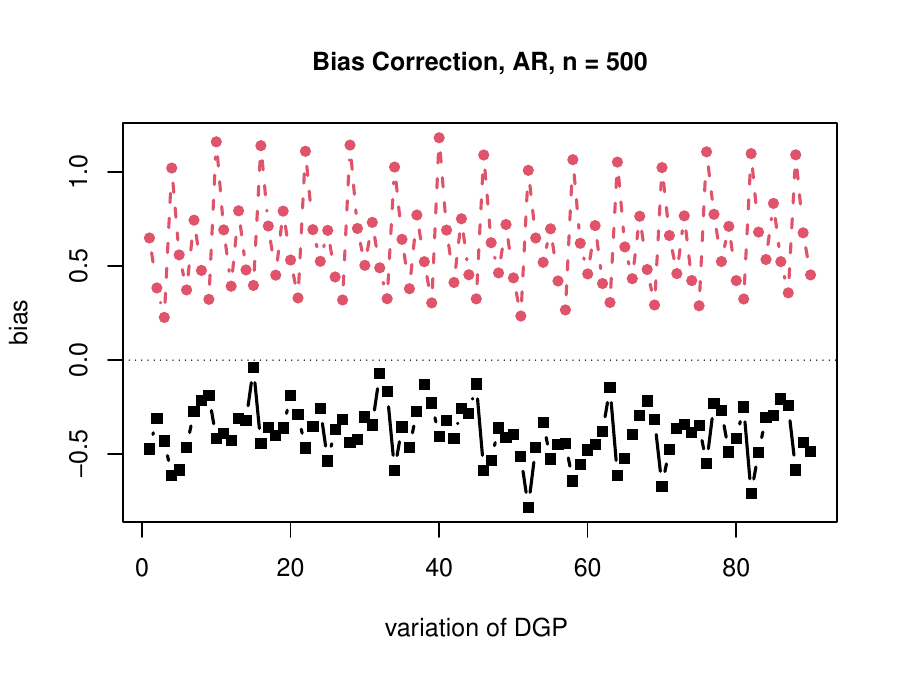}
			\subcaption{\tiny{Chow test bias, E2: $n=500$.}\label{fig:ChowbiasARMA500}}
		\end{subfigure}
		\begin{subfigure}{0.3\textwidth}
			\includegraphics[width=\linewidth]{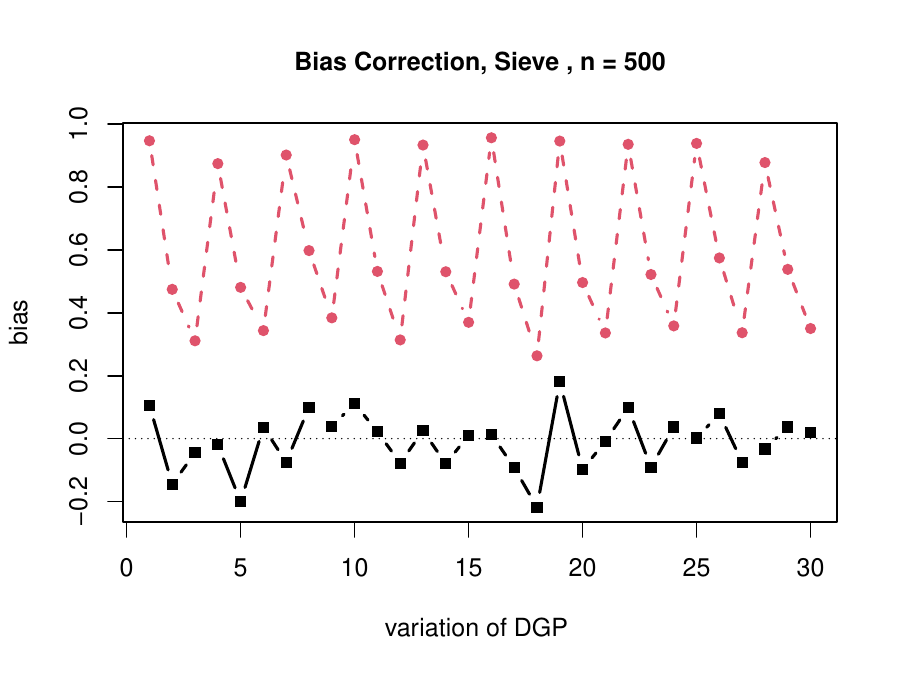}
			\subcaption{\tiny{Chow test bias, E3: $n=500$.}\label{fig:ChowbiasSieve500}}
		\end{subfigure}
		\caption{\scriptsize{Bias in $\mathcal{Q}_n(\gamma)$ (dot) and $\mathcal{T}_n^b(\gamma)$ (square, black). For E1, the vertical partitions in (a)  and (d) correspond to $ p=5,9 $ and $p=5,9,13$, respectively. Within each vertical partition results are ordered lexicographically as $(\gamma\in\{0.2,0.3,0.5\}, \text{error } \in\{1,2\},\alpha \in \{0.3,0.57\})$. E2 and E3: $p=9$ for $ n=250 $ and $ p=13 $ for $ n=500 $. Results horizontally ordered lexicographically as $(\gamma, \text{error },\alpha\text{ or }\theta)$.  \label{fig:TSbiasSieve}\label{fig:TSbiasAR}\label{fig:TSbiasmult}}}
	\end{figure}
\end{center}

%\ref{fig:TSbiasAR}

\begin{center}
	\begin{figure}
		\begin{subfigure}{0.3\textwidth}
			\includegraphics[width=\linewidth]{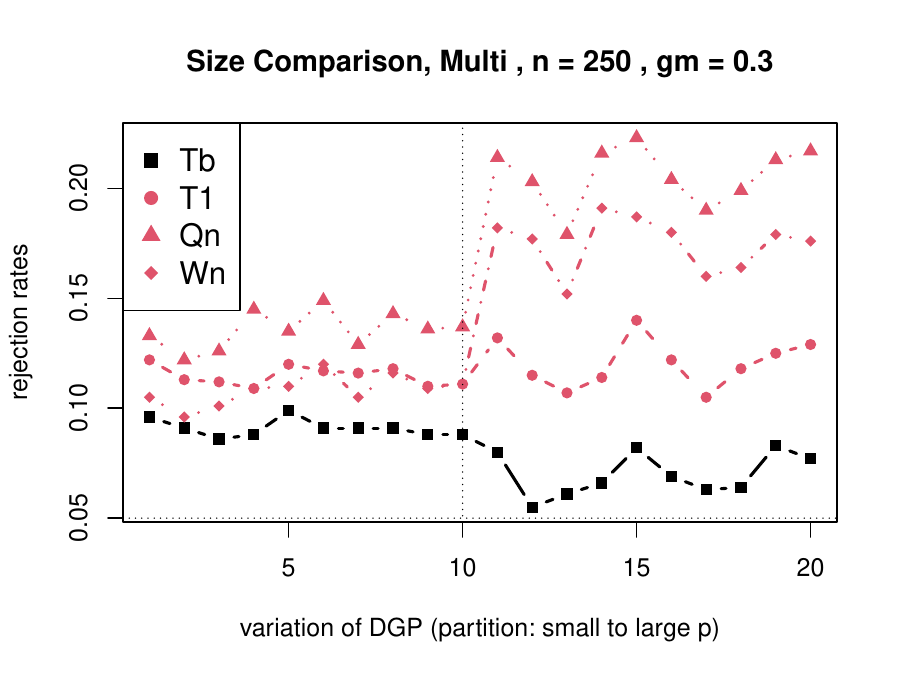}
			\subcaption{\tiny{Chow test size, E1: $n=250$, $\gamma=0.3$.}\label{fig:Chowsizemult250gm3}}
		\end{subfigure}
		\begin{subfigure}{0.3\textwidth}
			\includegraphics[width=\linewidth]{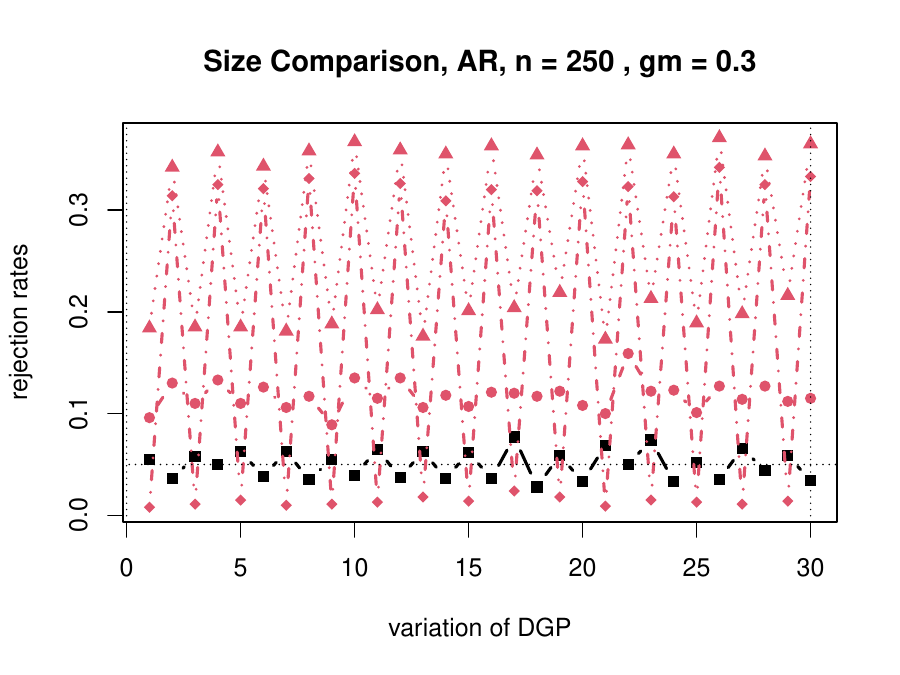}
			\subcaption{\tiny{Chow test size, E2: $n=250$, $\gamma=0.3$.}\label{fig:ChowsizeARMA250gm3}}
		\end{subfigure}
		\begin{subfigure}{0.3\textwidth}
			\includegraphics[width=\linewidth]{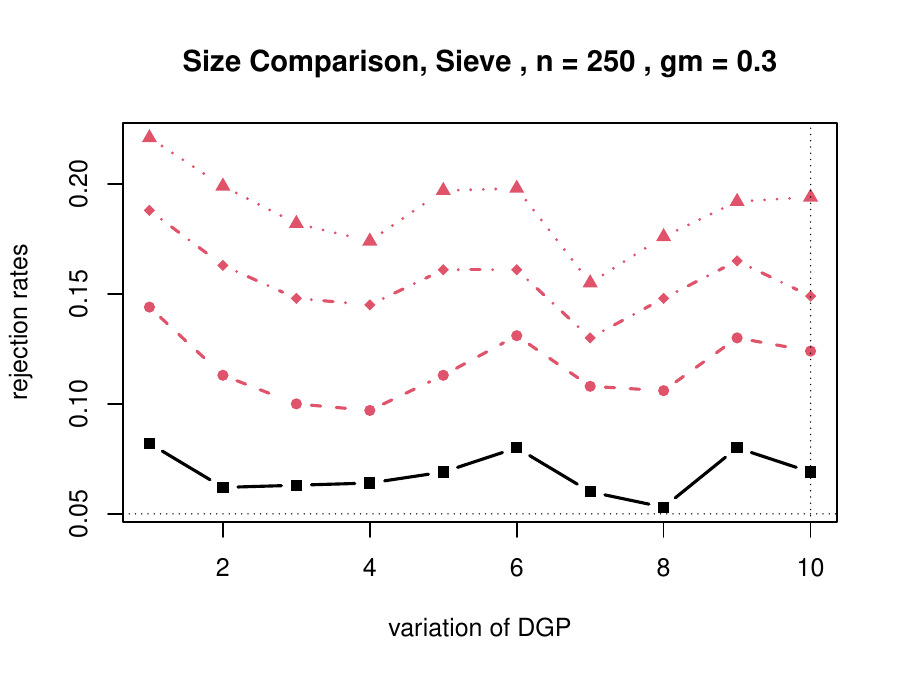}
			\subcaption{\tiny{Chow test size, E3: $n=250$, $\gamma=0.3$.}\label{fig:ChowsizeSieve250gm3}}
		\end{subfigure}	
		\begin{subfigure}{0.3\textwidth}
			\includegraphics[width=\linewidth]{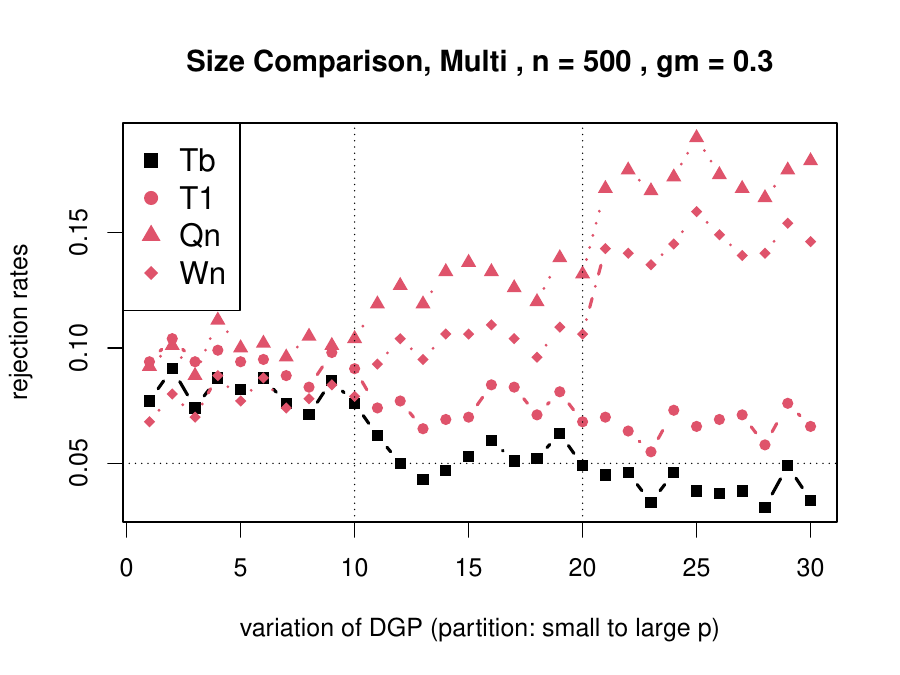}
			\subcaption{\tiny{Chow test size, E1: $n=500$, $\gamma=0.3$.}\label{fig:Chowsizemult500gm3}}
		\end{subfigure}
		\begin{subfigure}{0.3\textwidth}
			\includegraphics[width=\linewidth]{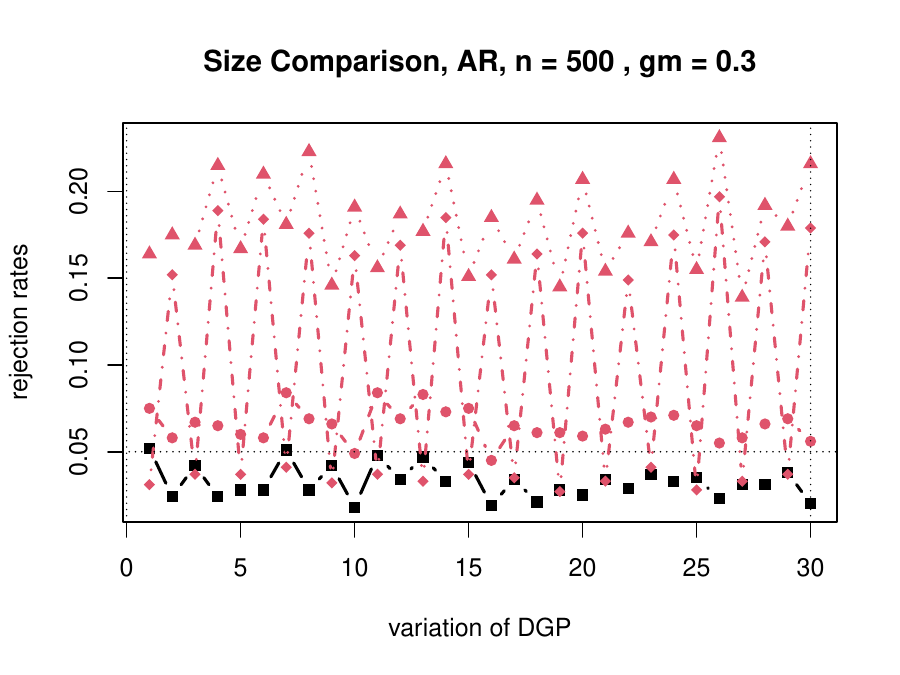}
			\subcaption{\tiny{Chow test size, E2: $n=500$, $\gamma=0.3$.}\label{fig:ChowsizeARMA500gm3}}
		\end{subfigure}
		\begin{subfigure}{0.3\textwidth}
			\includegraphics[width=\linewidth]{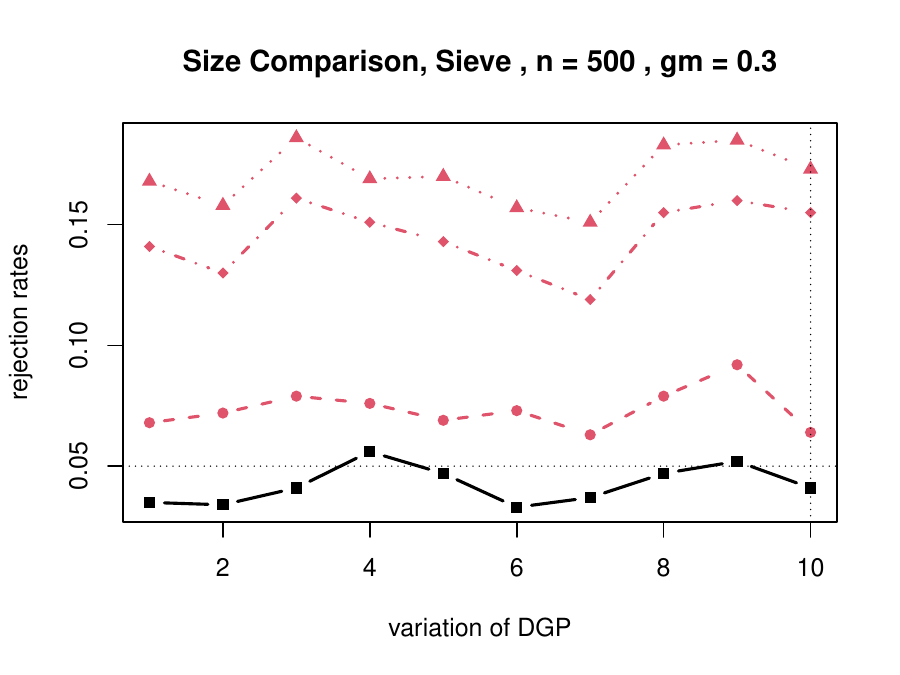}
			\subcaption{\tiny{Chow test size, E3: $n=500$, $\gamma=0.3$.}\label{fig:ChowsizeSieve500gm3}}
		\end{subfigure}
		
		\caption{\scriptsize{Size of Chow tests with $\gamma = 0.3 $ : ${W}_n(\gamma)$ (diamond), $\mathcal{Q}_n(\gamma)$ (triangle), $\mathcal{T}_n(\gamma)$ (dot) and $\mathcal{T}_n^b(\gamma)$ (square, black). Nominal size is 5\%. For E1, vertical partitions in (a) correspond to $n=250$ and $p=5,9$ and those in (d) correspond to $n=500$ and $p=5,9,13$. Within each vertical partition results are ordered lexicographically as $(\text{error } \in\{1,2\},\alpha \in \{0.3,0.57\})$. 
				For E2,
				$p=9$ for $ n=250 $ and $ p=13 $ for $ n=500 $. 
				Results horizontally ordered lexicographically as $(\text{error } \in\{1,2\},\alpha \in \{0.3,0.57\}\text{ or }\theta \in \{-0.5,-0.1,0.5\})$.
				For E3, 
				$p=9$ for $ n=250 $ and $ p=13 $ for $ n=500 $.  Results horizontally ordered lexicographically as $(\text{error } \in\{1,2\},\alpha \in \{0.3,...,0.57\})$.				
				\label{fig:TSsizemult}}
		}
	\end{figure}
\end{center}

\begin{center}
	\begin{figure}		
		\includegraphics[width=0.3\linewidth]{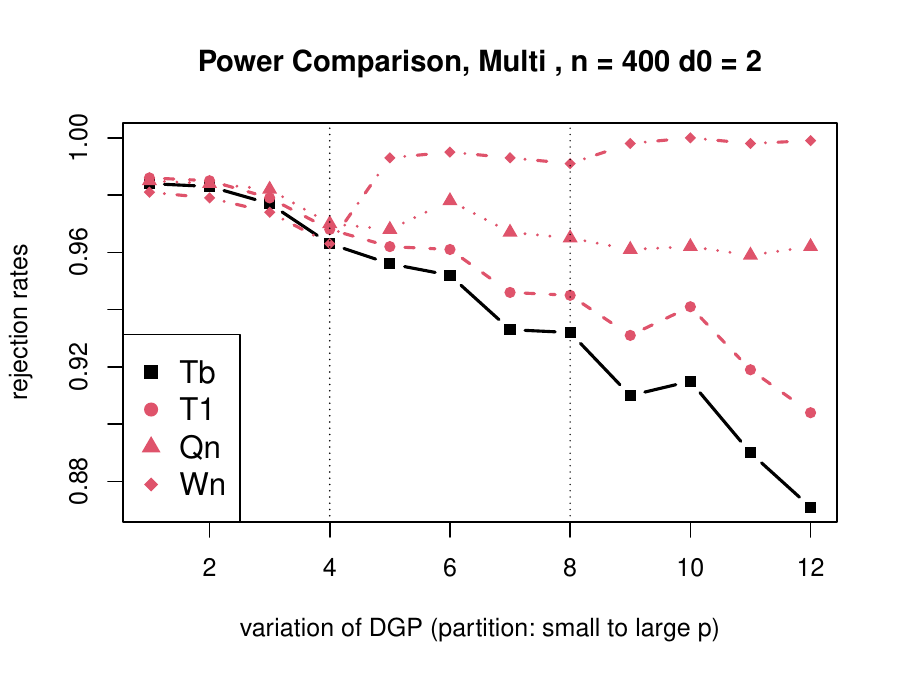}
		\includegraphics[width=0.3\linewidth]{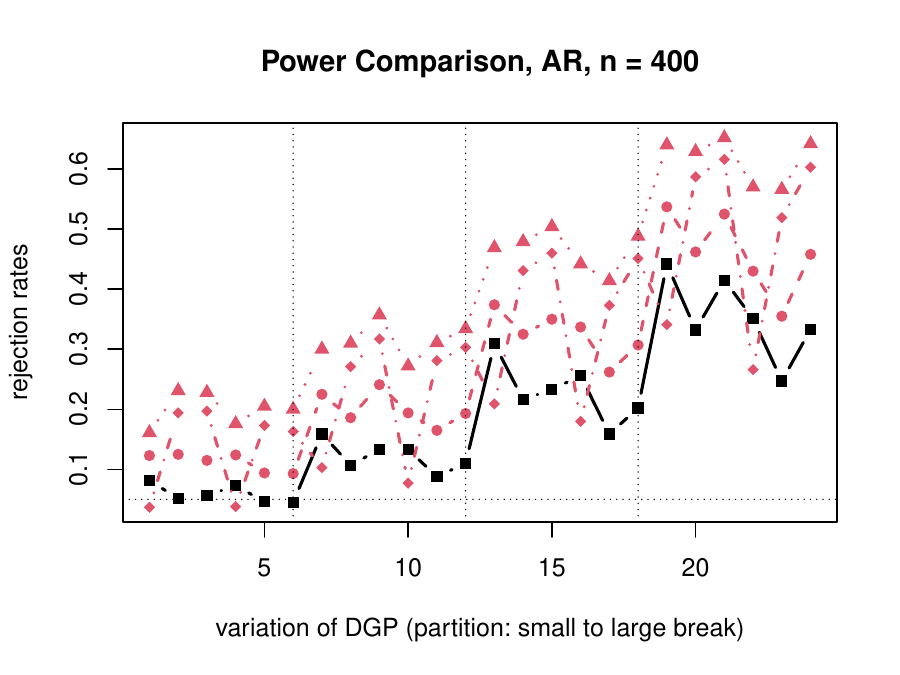}
		\includegraphics[width=0.3\linewidth]{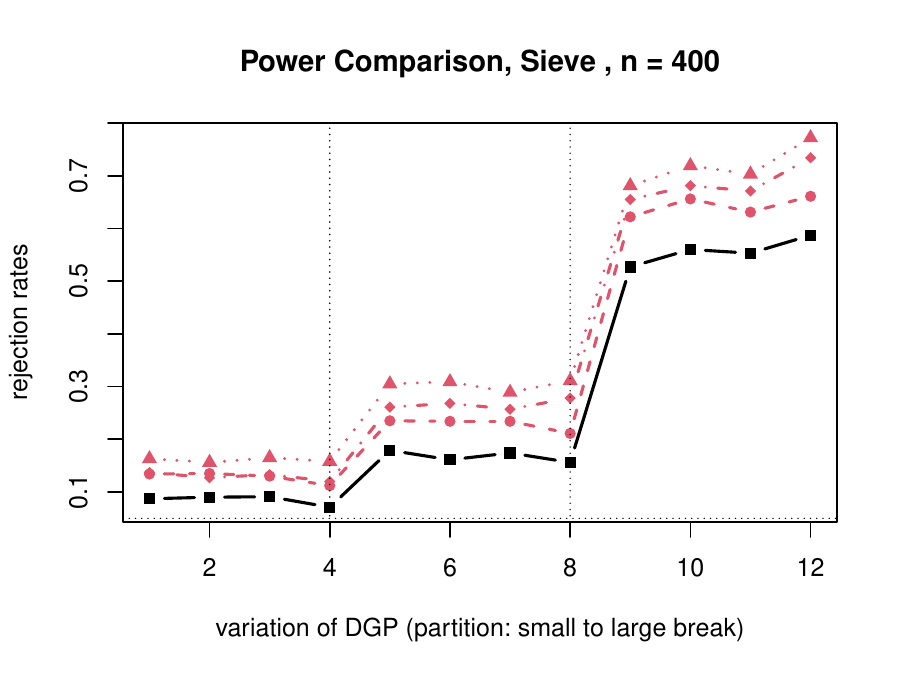}
		\caption{\scriptsize{Power of Chow tests, ${W}_n(\gamma)$ (diamond), $\mathcal{Q}_n(\gamma)$ (triangle), $\mathcal{T}_n(\gamma)$ (dot) and $\mathcal{T}_n^b(\gamma)$ (square, black): 
				E1 (left), E2 (center) and E3 (right), $n=400$, $\gamma=0.5$. 
				Vertical partitions correspond to $p=5,9,13$ (left), $\theta=0.2,0.4,0.6,0.8$ (center) and $\delta=0.5p^{1/4}/n^{1/2}(1,5,10)$ (right). Within each vertical partition results are ordered lexicographically as  $(\text{error } \in\{1,2,3\},\alpha \in \{0.3,0.5\})$ for E1 and $(\text{error } \in\{1,2,3\},\alpha \in \{0.3,0.5\}\text{ or }\theta)$ for E2 and E3.
				\label{fig:TSpowerARMAandSieve}}}
	\end{figure}
\end{center}

\begin{figure}
	\centering
%	\captionsetup{width=.8\linewidth}
	\includegraphics[width=.3\linewidth]{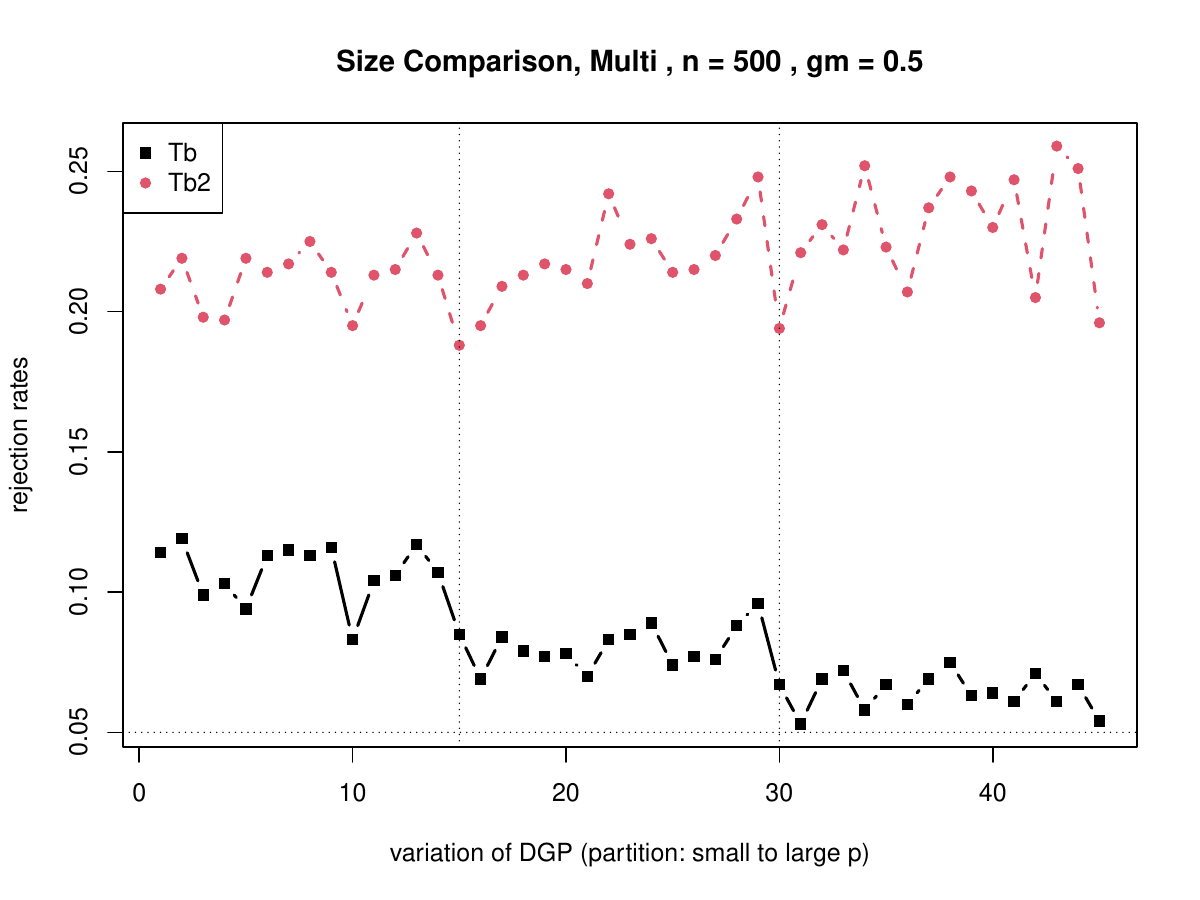}	
	\includegraphics[width=.3\linewidth]{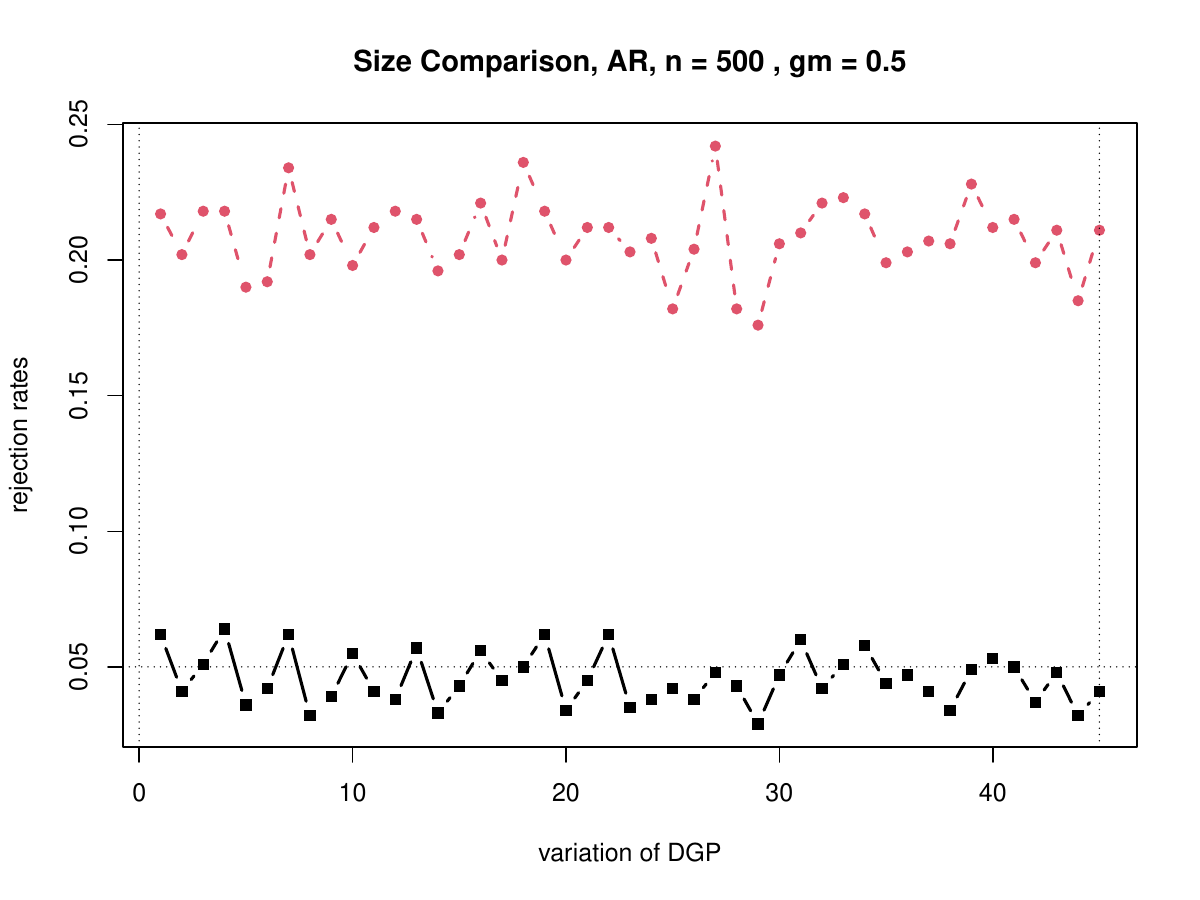}	
	\includegraphics[width=.3\linewidth]{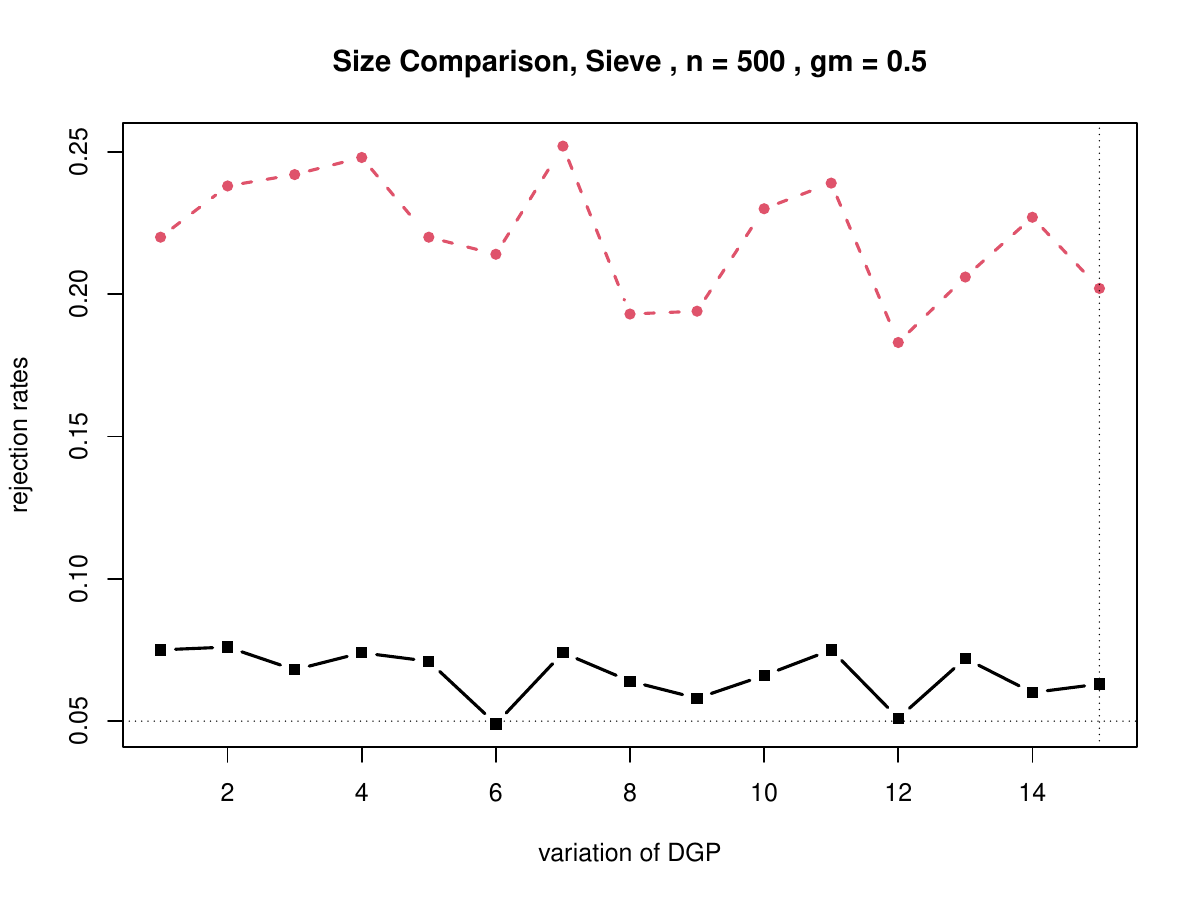}		
	\caption{Size distortions in the test based on $\mathcal{T}_n^{b,2}(\gamma)$, which is computed exactly like $\mathcal{T}_n^b(\gamma)$ except that the scaling is done by the sample variance of $ q_t $ along with the standard normal approximation.}
	\label{fig:size_HAC}
\end{figure}

\begin{center}
	\begin{figure}
		\includegraphics[width=0.3\linewidth]{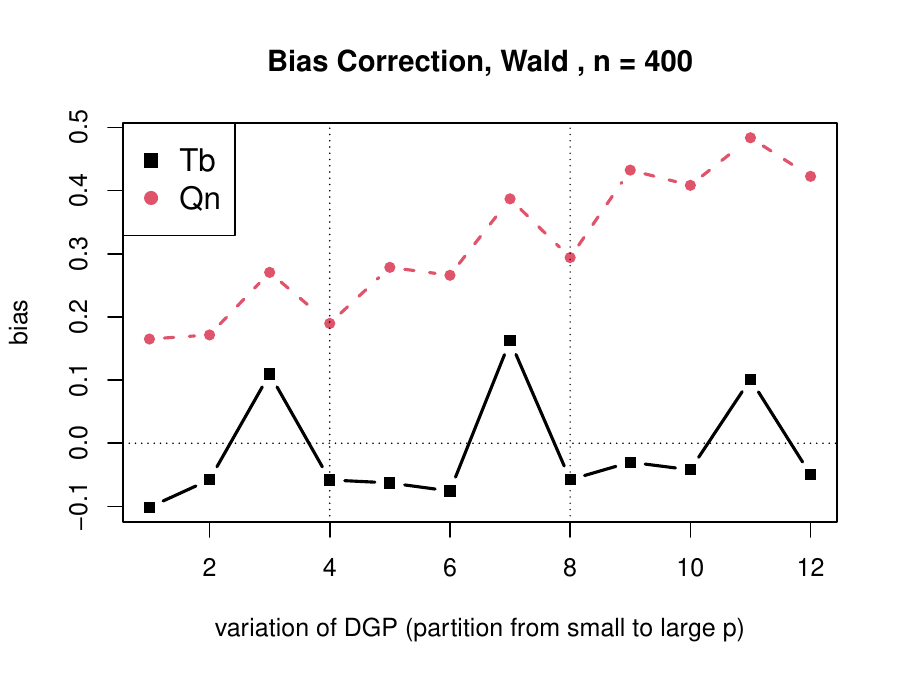}
		\includegraphics[width=0.3\linewidth]{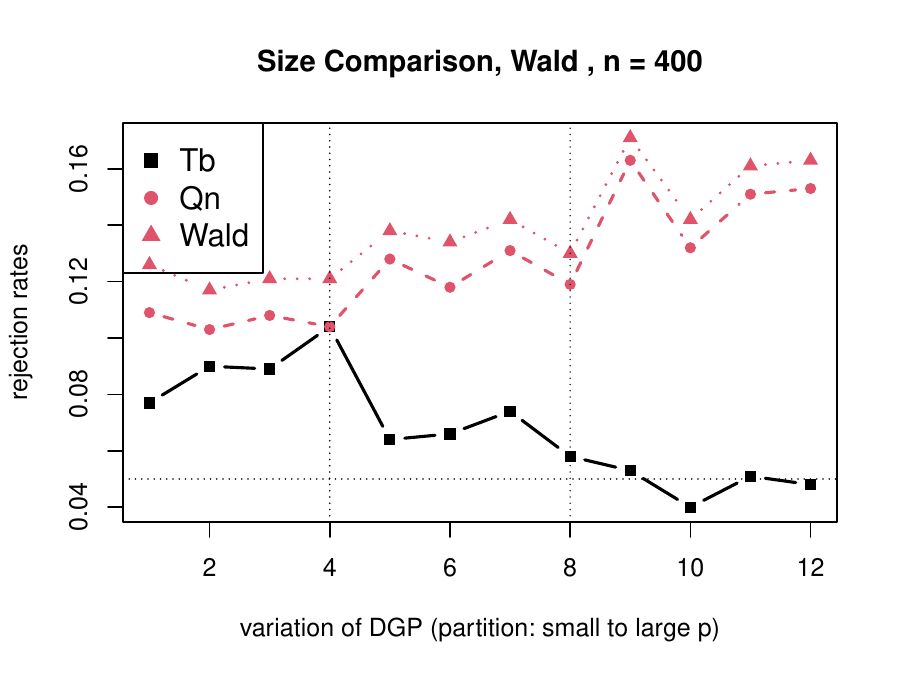}
		\includegraphics[width=0.3\linewidth]{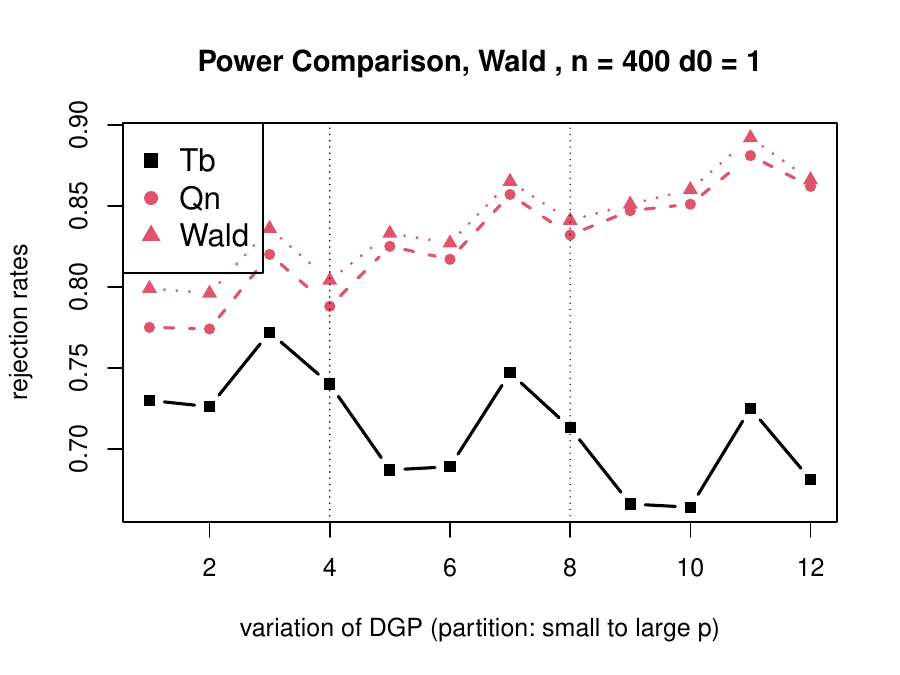}
		\caption{\scriptsize{Bias, Size, and Power of Exclusion Tests in E1: $\mathcal{Q}_n^e$ (dot) and $\mathcal{T}_n^{e,b}$ (square, black). 
				Left to right: vertical partitions correspond to $p=8,12,16$. Within each vertical partition results are ordered lexicographically as $(\text{error } \in\{1,2\},\alpha \in \{0.4,0.55\})$.\label{fig:TSWaldbiasmult}}}
	\end{figure}
\end{center}

\newpage

\appendix
\restoregeometry
\section{Proofs of theorems}\label{sec:app_thms}

We begin with some notation. Let 
\begin{equation*}
A(\gamma )=\left( X^{\ast }(\gamma )^{\prime }M_{X}X^{\ast }(\gamma )\right)
^{-1}X^{\ast }(\gamma )^{\prime }M_{X}
\end{equation*}%
with $X^{\ast }(\gamma )$ having $t$-th row $x_{t}^{\ast }(\gamma )^{\prime
}=x_{t}^{\prime }1\left\{ t/n>\gamma \right\} $, $M_{X}$ the residual
maker for the matrix $X$ with $t$-th row $x_{t}^{\prime }$, and 
\begin{equation*}
B(\gamma )=RM(\gamma )^{-1}\Omega (\gamma )M(\gamma )^{-1}R^{\prime }.
\end{equation*}%
Also, let $\bar{\Omega}(\gamma )=n^{-1}\sum_{t=1}^{n}x_{t}(\gamma
)x_{t}^{\prime }(\gamma )\sigma _{t}^{2}$ and $x_{t}(\gamma )=\left(
x_{t}^{\prime },x_{t}^{\prime \ast }(\gamma )\right) ^{\prime }$. It is also
convenient to recall that $\hat{M}=n^{-1}X^{\prime }X$ and define $\hat{S}(\gamma
)=n^{-1}X^{\prime \ast }(\gamma )X(\gamma )$. Recall that cross-referenced items prefixed with `S' can be found in the online supplementary appendix.

\subsection{Proofs for Section \ref{sec:theory}}

\begin{proof}[Proof of Theorem \protect\ref{thm:null_Chow} and \protect\ref{thm:local_power}] 
For the convenience of exposition, these are combined with the proof of Theorem \ref{thm:fixedbVdist} in the next section. 
\end{proof}

% \begin{proof}[Proof of Theorem \protect\ref{thm:Berk}]
%         Assumption \ref{ass:errors} is trivially met. Next, verify Assumption  \ref{ass:aprx0}.
%          Due to Jensen's inequality and
%         condition (ii),%
%         \begin{eqnarray*}
%                 E\left\vert r_{t}\right\vert ^{2a} &=&\left( \sum_{j=p}^{\infty
%                 }b_{j}\right) ^{2a}E\left\vert \sum_{j=p}^{\infty }\frac{b_{j}}{%
%                         \sum_{j=p}^{\infty }b_{j}}y_{t-j}\right\vert ^{2a} \\
%                 &\leq &\left( \sum_{j=p}^{\infty }b_{j}\right) ^{2a}\sum_{j=p}^{\infty }%
%                 \frac{b_{j}}{\sum_{j=p}^{\infty }b_{j}}E\left\vert y_{t-j}\right\vert ^{2a} = o\left( n^{-a}\right).
%         \end{eqnarray*}%
%         To verify Assumption \ref{ass:M_diff}, note that Lemma 3 in Berk (1974)
%         showed the whole sum convergence and it is readily extended to the partial
%         sums by an extension of Kolmogorov maximal inequality for $\alpha $-mixing,
%         see e.g. Theorem 3.1 of Rio (2017). Note that $\left\{ y_{t}\right\} $ is $%
%         \beta $-mixing under condition (ii) as shown by e.g. \cite{Doukhan1994}, p.
%         79, Theorem 2. Similarly, Assumption \ref{ass:Omega_diff} and \ref%
%         {ass:Omega_hat} follows due to the conditionally homoskedastistic error
%         variance estimation and the consistency proof for $\hat{\sigma}^{2}$ in
%         Berk's Theorem 1.
%         
%         Finally, Assumption \ref{ass:M_eigs} and \ref{ass:Omega_eigs} are met due to
%         Berk's (2.14).
%\end{proof}

\subsection{Proof of Theorem \protect\ref{thm:fixedbVdist}:}
For the result under the null,  Section \ref{subsec:Q} first establishes  the asymptotic normality for $ {\mathcal{Q}}_n(\gamma ), $ and then Section \ref{subsec:V} proves $\hat{\mathcal{V}}\overset{d}{\to}\int_{0}^{1}\int_{0}^{1}K_{h}^{*}\left(r,s\right)dW\left(r\right)dW\left(s\right)$, where $ W(r) $ denotes the same limit Gaussian process as in Theorem \ref{thm:T_weak_conv}.  Then, the claim follows by Theorem \ref{thm:T_weak_conv} and the continuous mapping theorem.
After  completing the proof under the null, we prove convergence under the local alternative in Section \ref{subsec:H1}.
        
\subsubsection{Asymptotic normality of $ {\mathcal{Q}}_n(\gamma )$ under $\mathcal{H}_0$. }\label{subsec:Q}
\begin{proof}
 This step is quite involved and we delegate proofs of many intermediate steps to Section \ref{sec:weak convergence}.  Summarizing these steps, 
        Theorem \ref{theorem:struc_break_approx} therein develops the initial approximation $ {\mathcal{Q}}_n(\gamma )=\left({{\mathcal{R}}
                _n(\gamma )-p}\right)/{\sqrt{2p}}+o_p(1) $, where $\mathcal{R}
            _n(\gamma )$ is defined in (\ref{fancyRdef}). Then,   
         (\ref{quad_form_full}) and Lemma \ref{lemma:diag_terms_neg} yield the second approximation 
         \[
        \frac{{\mathcal{R}}
                _n(\gamma )-p}{\sqrt{2p}}={{\mathcal{S}}
                _n(\gamma )}+o_p(1),\] 
where
        \begin{equation}
        {\mathcal{S}}_{n}(\gamma )=\frac{n^{-1}\sum_{s\neq t}g_{t}(\gamma )^{\prime
                }\Omega ^{-1}g_{s}(\gamma )\varepsilon _{t}\varepsilon _{s}}{\gamma \left(
                1-\gamma \right) \sqrt{2p}},  %\label{dist_target}
        \end{equation}
        and $g_{t}(\gamma )=x_t 1\left\{t/n\leq\gamma\right\}-\gamma x_t$. 
         The claim now follows by a CLT for $ \mathcal{S}_{n}(\gamma )$ established in Theorem \ref{thm:T_weak_conv}. 
\end{proof}
        
\subsubsection{Weak convergence of $ \hat{\mathcal{V}} $}       \label{subsec:V}        

\begin{proof}
First, we establish tightness of the stochastic process 
\[
	\mathcal{A}_{n}(\gamma ) =\frac{1}{n\sqrt{p}}{\sum_{s=2}^{[n\gamma ]}}{%
		\sum_{t=1}^{s-1}}\xi _{t}^{\prime }\xi _{s},
\]
with $\xi _{t}=\left\{ \xi _{ti}\right\} _{i=1}^{p}=\Omega
^{-1/2}x_{t}\varepsilon _{t}$ being an mds.

Note that $\mathcal{A}_{n}\left(\gamma\right)$
is a partial sum process of a heterogeneous martingale difference
array $w_{ns}=\xi_{s}^{\prime}{\sum_{t=1}^{s-1}}\xi_{t}/\sqrt{np}$,
and thus it is sufficient to show 
\begin{eqnarray}
E\left\vert \mathcal{A}_{n}\left(\gamma_{1}\right)-\mathcal{A}_{n}\left(\gamma_{2}\right)\right\vert ^{4} & = & E\left\vert \frac{1}{\sqrt{n}}{\sum_{s=\left[n\gamma_{1}\right]+1}^{[n\gamma_{2}]}}w_{ns}\right\vert ^{4}\nonumber \\
& \leq & E\left(\sum_{s}\mathrm{E}\left(w_{ns}^{2}|\mathcal{G}_{s-1}\right)/n\right)^{2}+n^{-1}\max_{s}E\left\vert w_{ns}\right\vert ^{4}O\left(\left\vert \gamma_{2}-\gamma_{1}\right\vert \right)\label{eq:max_w_s}\\
& = & O\left(\left\vert \gamma_{2}-\gamma_{1}\right\vert \right),\nonumber 
\end{eqnarray}
where we apply the Rosenthal inequality, e.g. \cite{Hall1980}, for the inequality and a calculation similar to (\ref{clt:var_cond}) and \eqref{cltvarbd1}  for the last equality. Specifically, 
\begin{equation}\label{tight_argument}
n^{-1}\max_{s}E\left\vert w_{ns}\right\vert ^{4}\leq \max_s E\left( E((\xi_{s}'\xi_{s})^2 | \mathcal{G}_{s-1}) \left( \sum_{t_1 ,t_2 <s} \xi_{t_1}'\xi_{t_2}\right)^{2}\right)n^{-3}p^{-2}=o(1),
\end{equation} 
by Assumption \ref{ass:MCLT} and the same reasoning as for \eqref{target} and  \eqref{cltvarbd1}.        
        
        Having established tightness, by Lemma 1 (c) of \cite{Sun2014} weak convergence follows if 
\begin{equation}
\hat{\mathcal{V}}-\tilde{\mathcal{V}}=o_p(1) \label{eq:VhatVtil},
\end{equation} 
where $\tilde{\mathcal{V}}=\frac{2}{n}\text{\ensuremath{\sum_{t=2}^{n}\text{\ensuremath{\sum_{s=2}^{n}k\left(\frac{t-s}{n/h}\right)\bar{q}^{\star}_{s}}}\bar{q}^{\star}_{t}}},$ $ q^{\star}_t = \left( np\right)^{-1/2} x_{t}'{\Omega}^{-1} {\varepsilon}_{t} \text{\ensuremath{\sum_{s=1}^{t-1}x_{s}{\varepsilon}_{s}}} $,  $\bar{q}^{\star}_{t}= q^{\star}_t - n^{-1}\sum_{t=2}^{n} q^{\star}_t$. Strictly speaking, Sun's Lemma 1 (c) is stated for the case where the partial sums of $ q_{t}$ are approximated by the partial sums of $ e_t $, which is iid normal, but it also holds when it is approximated by the partial sums of $ a_{nt }e_t $ for any real bounded array $ a_{nt} $ by repeating the same argument in the proof. In our case, $ a_{nt} = \sqrt{t/n} $. 

 Let $\varsigma=n/h$, $\hat\zeta_t=\hat{h}_t'\sum_{s<t}\hat{h}_s/\sqrt{p}=\sqrt{n}q_t$, $\hat{h}_t=\hat\Omega^{-1/2}x_t\hat{e}_t$, $\bar{\hat\zeta}/\sqrt{n}=n^{-1}\sum_{t=2}^{n} q_t=n^{-1}\sum_{t=2}^{n}\hat\zeta_t/\sqrt{n}$,  with analogous definitions using $\Omega$ and $\varepsilon_t$ for $\zeta_t$, $h_t$ and $\bar{\zeta}$. Then \begin{eqnarray}
\hat{{\mathcal{V}}}-\tilde{{\mathcal{V}}}&=&n^{-2}\sum_{j=-(n-1)}^{n-1}k\left(j/\varsigma \right)n^{-1}\sum_{t=1+\left(j\vee0\right)}^{n-\left\vert j\wedge 0\right\vert}\left\{\left(\hat\zeta_t\hat\zeta_{t+|j|}-\zeta_t\zeta_{t+|j|}\right)+2\bar{\hat\zeta}\left(\hat\zeta_t-\zeta_t\right)\right.\nonumber\\
&+&\left.\left(\bar{\hat\zeta}-\bar\zeta\right)\hat\zeta_t+\left(\bar{\hat\zeta}^2-\bar\zeta^2\right)\right\}.    \label{Vhat_Vtilde1}
\end{eqnarray}
We obtain a bound for 
\begin{equation}\label{zetapdiff}
\hat\zeta_t\hat\zeta_{t+|j|}-\zeta_t\zeta_{t+|j|}=\left(\hat\zeta_t-\zeta_t\right)\hat\zeta_{t+|j|}+\left(\hat\zeta_{t+|j|}-\zeta_{t+|j|}\right)\hat\zeta_{t},
\end{equation}
while omitting similar details for the other three terms.
To find a bound for (\ref{zetapdiff}), first note that $\left\Vert\hat{h}_t\right\Vert=O_p\left( \left\Vert x_t\right\Vert\right)=O_p\left( \sqrt{p}\right)$, 
by Assumption \ref{ass:M_diff}$(ii)$ and finite fourth moments of $x_t$ components (Assumption \ref{ass:M_diff}$(i)$), and because
\begin{equation}
\hat{e}_{t}=y_{t}-x_{t}^{\prime }\hat{\delta _{1}}(\gamma
)=x_{t}^{\prime }\left( \hat{\delta}_{1}(\gamma )-\delta _{1}\right)
+x_{t}^{\prime }1\left( t/n>\gamma \right) \delta _{2\ell
}+r_{t}+\varepsilon _{t}=O_{p}(1).  \label{epsilonhatorder}
\end{equation} Hence
\begin{equation}\label{zetabound}
\hat\zeta_t=\hat{h}_t'\sum_{s<t}\hat{h}_s/\sqrt{p}=O_p\left(n\sqrt{p}\right).
\end{equation}
By the same argument, $\left\Vert h_t\right\Vert=O_p(\sqrt{p})$ and $\zeta_t=O_p\left(n\sqrt{p}\right)$  as well. 

Now recall that $\hat{e}_{t}-\varepsilon _{t}=x_{t}^{\prime }\left( 
\hat{\delta}_{1}(\gamma )-\delta _{1}\right) +x_{t}^{\prime }1\left(
t/n>\gamma \right) \delta _{2\ell }+r_{t}$ and $\left\Vert \hat{\delta}%
_{1}(\gamma )-\delta _{1}\right\Vert =O_{p}\left( \left\Vert \hat{\delta}%
(\gamma )-\delta \right\Vert \right) =O_{p}\left(\lambda_n^{-1} \sqrt{p}/\sqrt{n}\right) $
implying that 
\begin{equation}
\hat{e}_{t}-\varepsilon _{t}=O_{p}\left( \max \left\{ \lambda_n^{-1}p/\sqrt{n}%
,p^{3/4}/\sqrt{n}\right\} \right).  \label{epsdifforder}
\end{equation}
 Thus we obtain
\begin{equation}\label{hdiffbound}
\left\Vert\hat{h}_t-h_t\right\Vert=\left\Vert\Omega^{-1}\left(\Omega-\hat\Omega\right)\hat\Omega^{-1}x_t\hat{e}_t+\Omega^{-1}x_t\left(\hat{e}_t-\varepsilon_t\right)\right\Vert=O_p\left(\lambda_n^{-2}\sqrt{p}\max\left\{v_p,p/\sqrt{n}\right\}\right),
\end{equation}
using Assumption \ref{ass:M_diff}$(iii)$. Using (\ref{hdiffbound}), we get
 \begin{equation}\label{zetadiffbound}
\hat\zeta_t-\zeta_t =\left(\hat{h}_t-h_t\right)'\sum_{s<t}\hat{h}_s/\sqrt{p}+\hat{h}_t'\sum_{s<t}\left(\hat{h}_s-h_s\right)/\sqrt{p}=O_p\left(\lambda_n^{-2}n\sqrt{p}\max\left\{v_p,p/\sqrt{n}\right\}\right).
\end{equation}

Using (\ref{zetabound}) and (\ref{zetadiffbound}) in (\ref{zetapdiff}), we obtain $\hat\zeta_t\hat\zeta_{t+|j|}-\zeta_t\zeta_{t+|j|}=O_p\left(\lambda_n^{-2}n^2p\max\left\{v_p,p/\sqrt{n}\right\}\right)$. This, along with similarly obtained bounds for the remaining terms in (\ref{Vhat_Vtilde1}) and Lemma 1 of \cite{Jansson2002}, yield
\[
\hat{{\mathcal{V}}}-\tilde{{\mathcal{V}}} =O_{p}\left(\varsigma \left( \int_{%
{\mathbb{R}}}\left\vert k(x)\right\vert dx\right) \lambda_n^{-2}p\max\left\{v_p,p/\sqrt{n}\right\}
\right) =O_{p}\left(\lambda_n^{-2} \max\left\{pv_p,p^2/\sqrt{n}\right\} \right) ,
\]
which is negligible by (\ref{rate:Vhat}).
\end{proof}

\subsubsection{Characterize the relationship between the numerator and denominator}
For this, we derive the covariance kernel of $\left( \mathcal{A}_{n}\left( \gamma
\right) ,\bar{\mathcal{A}}_{n}\left( \gamma \right) \right) ^{\prime }$, where \[
\bar{\mathcal{A}}_{n}(\gamma )=\frac{1}{n\sqrt{p}}{\sum_{s=[n\gamma
		]+1}^{n}\sum_{t=[n\gamma ]+1}^{s-1}}\xi _{t}^{\prime }\xi _{s}.
\] 
	Note that $E\left( \mathcal{A}_{n}\left( \gamma _{2}\right) -\mathcal{A}_{n}\left( \gamma _{1}\right)
\right) \mathcal{A}_{n}\left( \gamma _{1}\right) =0$ for any $\gamma _{1}<\gamma _{2}$. From the proof of Theorem \ref{thm:T_weak_conv} in the supplementary material, we have
\begin{align*}
E\left\vert \mathcal{A}_{n}\left( \gamma \right) \right\vert ^{2}&=\frac{\gamma ^{2}{\mathcal{V}}}{2}+o\left( 1\right) ,\\
E\left\vert \bar{\mathcal{A}}_{n}\left( \gamma \right) \right\vert ^{2}
& =\frac{\left( 1-\gamma \right) ^{2}{\mathcal{V}}}{2}+o\left( 1\right) ,
\end{align*}
where ${\mathcal{V}}$ is given in (\ref{V_partial_def}). Thus 
\begin{equation*}
	E\left( \mathcal{A}_{n}\left( \gamma _{1}\right) \mathcal{A}_{n}\left( \gamma _{2}\right)
	\right) \rightarrow \frac{\left( \gamma _{1}\wedge \gamma _{2}\right) ^{2}}{2%
	}{\mathcal{V}},
\end{equation*}%
and, similarly noting that $E\left( \bar{\mathcal{A}}_{n}\left( \gamma _{2}\right) -\bar{\mathcal{A}}%
_{n}\left( \gamma _{1}\right) \right) \bar{\mathcal{A}}_{n}\left( \gamma _{1}\right)
=0 $ for any $\gamma _{1}>\gamma _{2}$, 
we have
\begin{equation*}
	E\left( \bar{\mathcal{A}}_{n}\left( \gamma _{1}\right) \bar{\mathcal{A}}_{n}\left( \gamma
	_{2}\right) \right) \rightarrow \frac{\left( 1-\left( \gamma _{1}\vee \gamma
		_{2}\right) \right) ^{2}}{2}{\mathcal{V}}.
\end{equation*}%
Finally
\begin{align*}
	E\left( \mathcal{A}_{n}\left( \gamma _{1}\right) \bar{\mathcal{A}}_{n}\left( \gamma _{2}\right)
	\right) & =\frac{1\left\{ \gamma _{1}>\gamma _{2}\right\} }{n^{2}p}\sum_{s=%
		\left[ n\gamma _{2}\right] +1}^{\left[ n\gamma _{1}\right] }{\mathrm{tr}}E%
	\left( \xi _{s}\xi _{s}^{\prime }\sum_{t=1}^{s-1}\sum_{u=\left[ n\gamma _{2}%
		\right] +1}^{s-1}\xi _{t}\xi _{u}^{\prime }\right) \\
	& =\frac{1\left\{ \gamma _{1}>\gamma _{2}\right\} }{n^{2}}\sum_{s=\left[
		n\gamma _{2}\right] +1}^{\left[ n\gamma _{1}\right] }\left( s-1-\left[
	n\gamma _{2}\right] \right) {\mathcal{V}}+o\left( 1\right) \\
	& =\frac{1\left\{ \gamma _{1}>\gamma _{2}\right\} }{2}\left( \gamma
	_{1}-\gamma _{2}\right) ^{2}{\mathcal{V}}+o\left( 1\right) .
\end{align*}

\subsubsection{Under the alternative}\label{subsec:H1}
\begin{proof}
We present a general proof where the true break point is $\gamma_0$, and setting $\gamma=\gamma_0$ gives our claim in the paper. Under $\mathcal{H}_{\ell }$, we have $\hat{\delta}_{2}(\gamma )=A(\gamma )X^{\ast
}\left( \gamma _{0}\right) \delta _{2\ell }+A(\gamma )\varepsilon +A(\gamma
)r$, so that, writing $D\left( \gamma ,\gamma _{0}\right) =A(\gamma )X^{\ast
}\left( \gamma _{0}\right) $ and $\hat{B}(\gamma )=R\hat{M}(\gamma )^{-1}\hat{\Omega}(\gamma )\hat{M}(\gamma )^{-1}R^{\prime }$, similar algebra to that used in the online appendix and Lemmas \ref{lemma:wald_approx}-\ref{lemma:fancyS_lemma2}
yields 
\begin{eqnarray}
{\mathcal{Q}}_{n}(\gamma ) &=&{{\mathcal{S}}_{n}(\gamma )}+\frac{2n\delta
_{2\ell }^{\prime }D\left( \gamma ,\gamma _{0}\right) ^{\prime }\hat{B}%
(\gamma )^{-1}A(\gamma )\varepsilon }{\sqrt{2p}}+\frac{2n\delta _{2\ell
}^{\prime }D\left( \gamma ,\gamma _{0}\right) ^{\prime }\hat{B}(\gamma
)^{-1}A(\gamma )r}{\sqrt{2p}}  \notag \\
&+&\frac{n\delta _{2\ell }^{\prime }D\left( \gamma ,\gamma _{0}\right)
^{\prime }\left( \hat{B}(\gamma )^{-1}-{B}(\gamma )^{-1}\right) D\left(
\gamma ,\gamma _{0}\right) \delta _{2\ell }}{\sqrt{2p}}  \label{local_alt_1} \\
&+&\frac{n\delta _{2\ell }^{\prime }D\left( \gamma ,\gamma _{0}\right)
^{\prime }\hat{B}(\gamma)^{-1}D\left( \gamma ,\gamma _{0}\right) \delta _{2\ell }}{\sqrt{2p}}%
+o_{p}(1). \notag
\end{eqnarray}%
For the second term on the RHS of (\ref{local_alt_1}), note that this equals 
\begin{eqnarray}
&&\frac{2n\delta _{2\ell }^{\prime }D\left( \gamma ,\gamma _{0}\right)
^{\prime }{B}(\gamma )^{-1}A(\gamma )\varepsilon }{\sqrt{2p}}+\frac{2n\delta
_{2\ell }^{\prime }D\left( \gamma ,\gamma _{0}\right) ^{\prime }\left( \hat{B%
}(\gamma )^{-1}-{B}(\gamma )^{-1}\right) A(\gamma )\varepsilon }{\sqrt{2p}} 
\notag  \label{local_alt_2} \\
&=&\frac{2n\delta _{2\ell }^{\prime }D\left( \gamma ,\gamma _{0}\right) {B}%
(\gamma )^{-1}A(\gamma )\varepsilon }{\sqrt{2p}}+O_{p}\left(\lambda_n^{-2} n\left\Vert
\delta _{2\ell }\right\Vert \left\Vert n^{-1}X^{\prime }\varepsilon
\right\Vert \left\Vert \hat{B}(\gamma )-{B}(\gamma )\right\Vert /\sqrt{p}%
\right)  \notag \\
&=&\frac{2n\delta _{2\ell }^{\prime }D\left( \gamma ,\gamma _{0}\right)
^{\prime }{B}(\gamma )^{-1}A(\gamma )\varepsilon }{\sqrt{2p}}+O_{p}\left( 
\lambda_n^{-4}p^{1/4} \max \left\{ \lambda_n^{-1}\varkappa _{p},v_p \right\} \right) ,  \notag
\end{eqnarray}%
proceeding like (\ref{Omegadifflater}), the second stochastic order above
being negligible by (\ref{rate:Q_weak_conv}). By Assumption \ref{ass:errors}%
, the first term  has mean zero and variance equal to
a constant times 
\begin{equation*}
\frac{\tau ^{\prime }D\left( \gamma ,\gamma _{0}\right) ^{\prime } B(\gamma)^{-1}A(\gamma
)A(\gamma)'B(\gamma )^{-1}D\left( \gamma ,\gamma _{0}\right) \tau }{\sqrt{p}}=O_{p}(1/%
\sqrt{p}),
\end{equation*}%
uniformly in $\gamma $ by Lemmas \ref{lemma:B_eigs} and \ref{lemma:Bhat_norm}
and the calculations therein.

By Assumption \ref{ass:aprx0}, the third term on the RHS of (\ref{local_alt_1}) is 
\begin{equation*}
O_{p}\left( n\left\Vert \delta _{2\ell }\right\Vert \left\Vert
n^{-1}X^{\prime }r\right\Vert /\sqrt{p}\right) =O_{p}\left( p^{-1/4}%
\right) .
\end{equation*}%

The fourth term on the RHS of (\ref{local_alt_1}) is readily seen to be $%
O_{p}\left( \left\Vert \hat{B}(\gamma )^{-1}-{B}(\gamma )^{-1}\right\Vert
\right) =O_{p}\left( \lambda_n^{-4}\max \left\{\lambda_n^{-1} \varkappa _{p},v_p \right\} \right) 
$, which is negligible by (\ref{rate:Q_weak_conv}). Thus, using similar steps to replace $\hat{B}(\gamma )^{-1}$ by ${B}(\gamma )^{-1}$ in the fifth term on the RHS and by (\ref{B_gamma1}%
), (\ref{local_alt_1}) becomes 
\begin{eqnarray}
{\mathcal{Q}}_{n}(\gamma ) &=&{\mathcal{S}}_{n}(\gamma )+\frac{n\delta
_{2\ell }^{\prime }D\left( \gamma ,\gamma _{0}\right) ^{\prime }{B}(\gamma
)^{-1}D\left( \gamma ,\gamma _{0}\right) \delta _{2\ell }}{\sqrt{2p}}%
+o_{p}(1)  \notag \\
&=&{{\mathcal{S}}_{n}(\gamma )}+\gamma (1-\gamma ){\tau ^{\prime }D\left(
\gamma ,\gamma _{0}\right) ^{\prime }M \Omega^{-1} M D\left( \gamma ,\gamma _{0}\right)
\tau }+o_{p}(1).  \notag
\end{eqnarray}%
Now, by the definition of its components and steps similar to those
elsewhere in the paper, it is readily seen that 
$\left\Vert D\left( \gamma ,\gamma _{0}\right) -\left\{{\left( \gamma +\gamma
_{0}(1-\gamma )-\left( \gamma \vee \gamma _{0}\right) \right) }/{\gamma
(1-\gamma )}\right\}I_{p}\right\Vert =o_{p}(1),
$ uniformly on $\Gamma $ and that $\gamma +\gamma _{0}(1-\gamma )-\left(
\gamma \vee \gamma _{0}\right) =-\gamma \gamma _{0}+\left( \gamma \wedge
\gamma _{0}\right) $ as $\gamma +\gamma _{0}-\left( \gamma \vee \gamma
_{0}\right) =\left( \gamma \wedge \gamma _{0}\right) .$ Thus, 
\begin{equation}
{\mathcal{Q}}_{n}(\gamma )\overset{d}{\rightarrow} Q(\gamma )+\frac{\left( \gamma \gamma
_{0}-\left( \gamma \wedge \gamma _{0}\right) \right) ^{2}}{\gamma (1-\gamma )%
}\lim_{n\rightarrow \infty }{\tau ^{\prime}M \Omega^{-1} M\tau },  \label{local_alt_3}
\end{equation}%
by Theorem \ref{thm:T_weak_conv}, which
gives the distribution of ${\mathcal{Q}}_{n}(\gamma )$ under $\mathcal{H}_{\ell }$.
\end{proof}
\begin{proof}[Proof of Theorem \ref{thm:bootstrap}]
In Section \ref{sec:bootstrap_proof} of the online supplement.
\end{proof}	

\subsection{Proofs for Section  \ref{sec: extension}}

\begin{proof}[Proof of Theorem \ref{thm:exclusion_thm}]
The proof proceeds exactly as that of Theorem \ref{thm:fixedbVdist}, but without $\gamma$. We give a brief summary and omit the details. Because $R^e\hat\beta-r=Rn^{-1}\hat M^{-1}\sum_{t=1}^nx_t\varepsilon_t$ under $\mathcal{H}_0^e$, we can obtain the approximation 
\[
\mathcal{Q}_n^e=\frac{n^{-1}\left(\sum_{t=1}^nx_t\varepsilon_t\right)'L\left(\sum_{t=1}^nx_t\varepsilon_t\right)-p}{\sqrt{2p}}+o_p(1).
\]
Then, the proof of asymptotic normality follows with $w_{ns}=\xi_s'\sum_{t=1}^{s-1}\xi_t/\sqrt{np}$ as in Theorem \ref{thm:T_weak_conv}, but now defining $\xi_t=\left(R^e M^{-1}\Omega M^{-1}R^{e\prime}\right)^{-1/2}R^eM^{-1}x_t\varepsilon_t$. From this  it is readily seen that $E\left(\sum_{s=1}^nw_{ns}\right)^2=\mathcal{V}^e+o(1)$. 
\end{proof}

\section{Verification of high-level conditions for Examples E1, E2, and E3}\label{sec:primitive}

This section verifies some of the high level conditions for our examples analytically, while others are verified numerically. 

First, we show that Assumptions \ref{ass:errors} and \ref{ass:aprx0} are satisfied for all the examples. For this, note that the innovations $\eta_{t}$ in the DGPs are mixed-normal with finite moments of all order. Recall that an AR(1) or MA(1) process with bounded ARCH innovations whose AR or MA coefficients satisfy $\left\vert \alpha\right\vert <1$ and $\left\vert \alpha_{x}\right\vert <1$ (as in our DGP for $x_{it}$) is strictly stationary and $\beta$-mixing, see e.g. Theorem 15.0.1 of \cite{Meyn1993}. 
Thus, $x_{it}$ is strictly stationary and $\beta$-mixing with a finite moment generating function. This implies that Assumptions \ref{ass:errors} and \ref{ass:aprx0} are met for all the examples.

\begin{figure}
	\centering
	\begin{subfigure}{0.3\textwidth} \includegraphics[width=\linewidth]{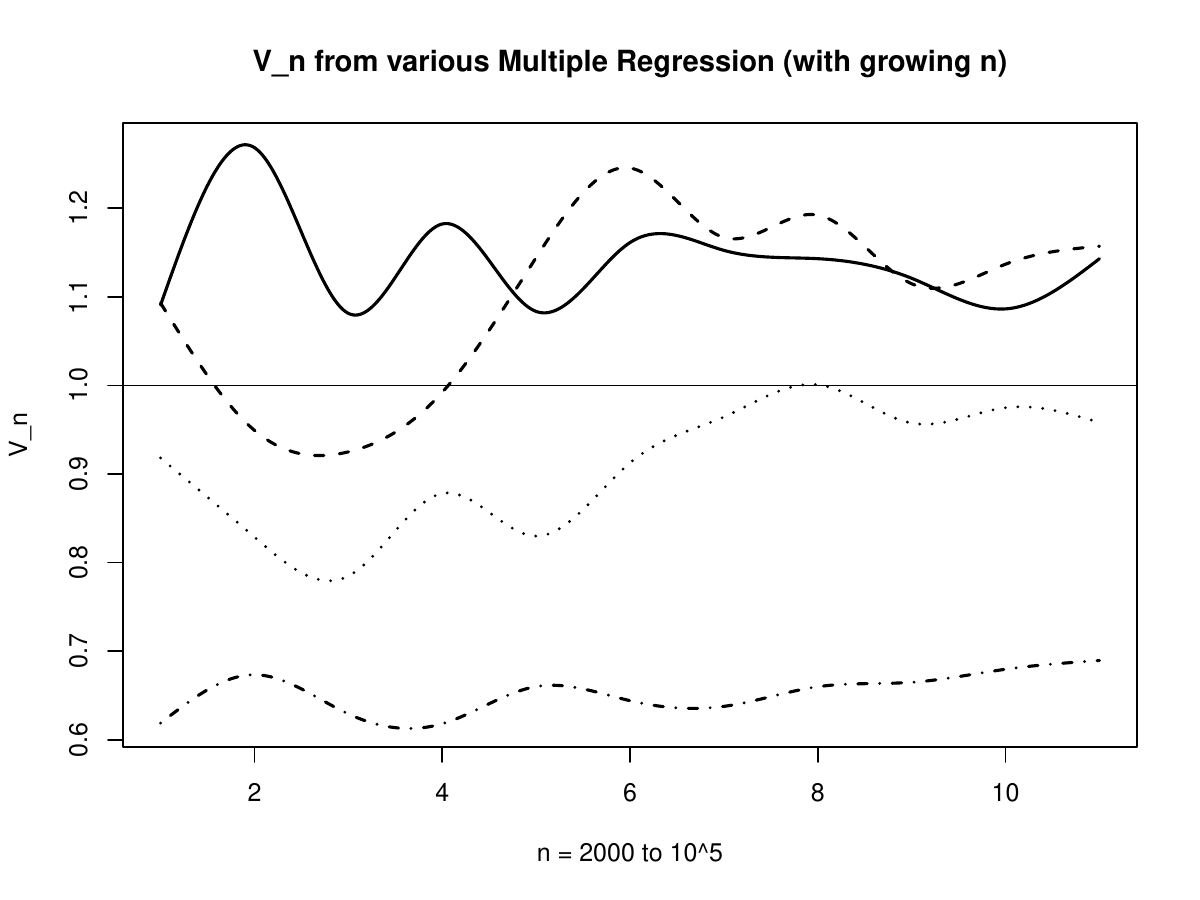}	
		\label{fig:V_multi}
	\end{subfigure}
	\begin{subfigure}{0.3\textwidth} \includegraphics[width=\linewidth]{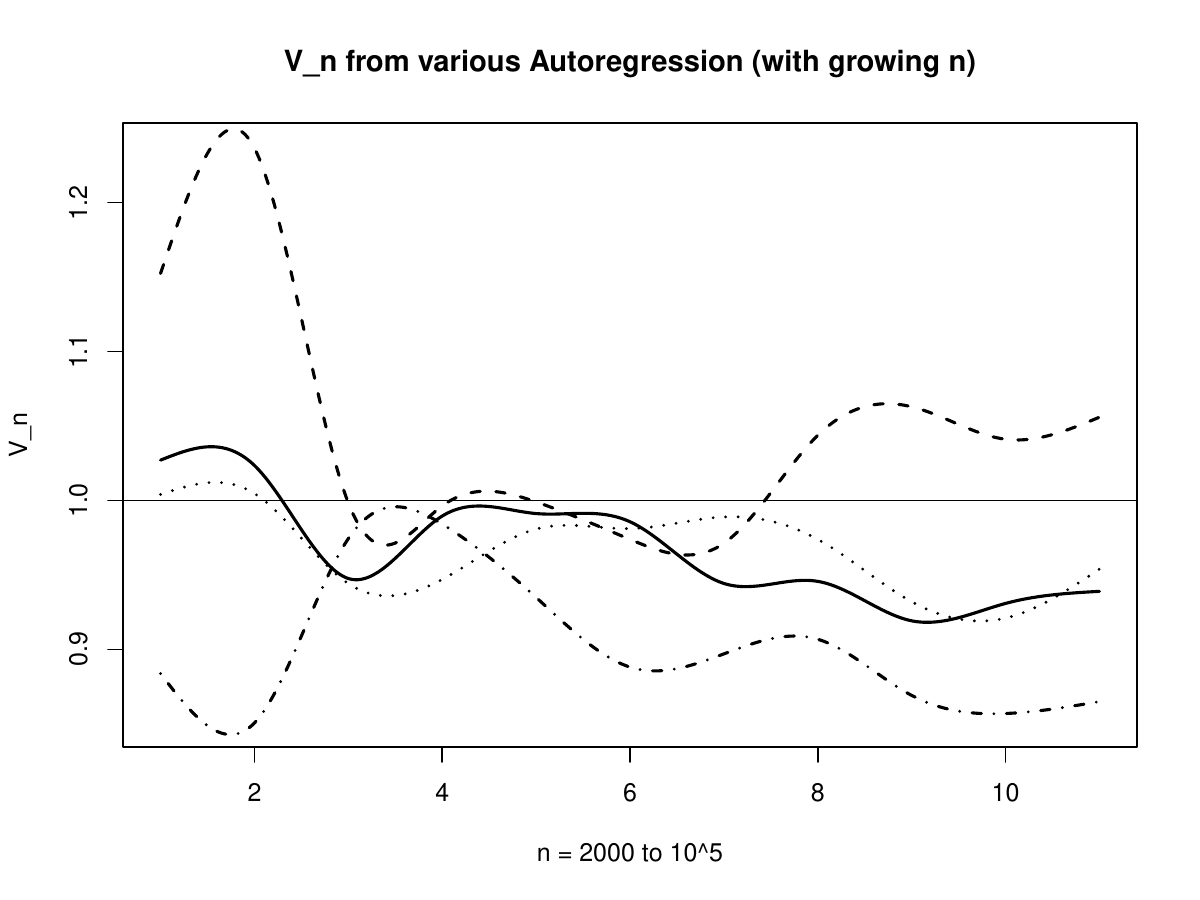}	
		\label{fig:V_ARMA} 
	\end{subfigure}
	\begin{subfigure}{0.3\textwidth} \includegraphics[width=\linewidth]{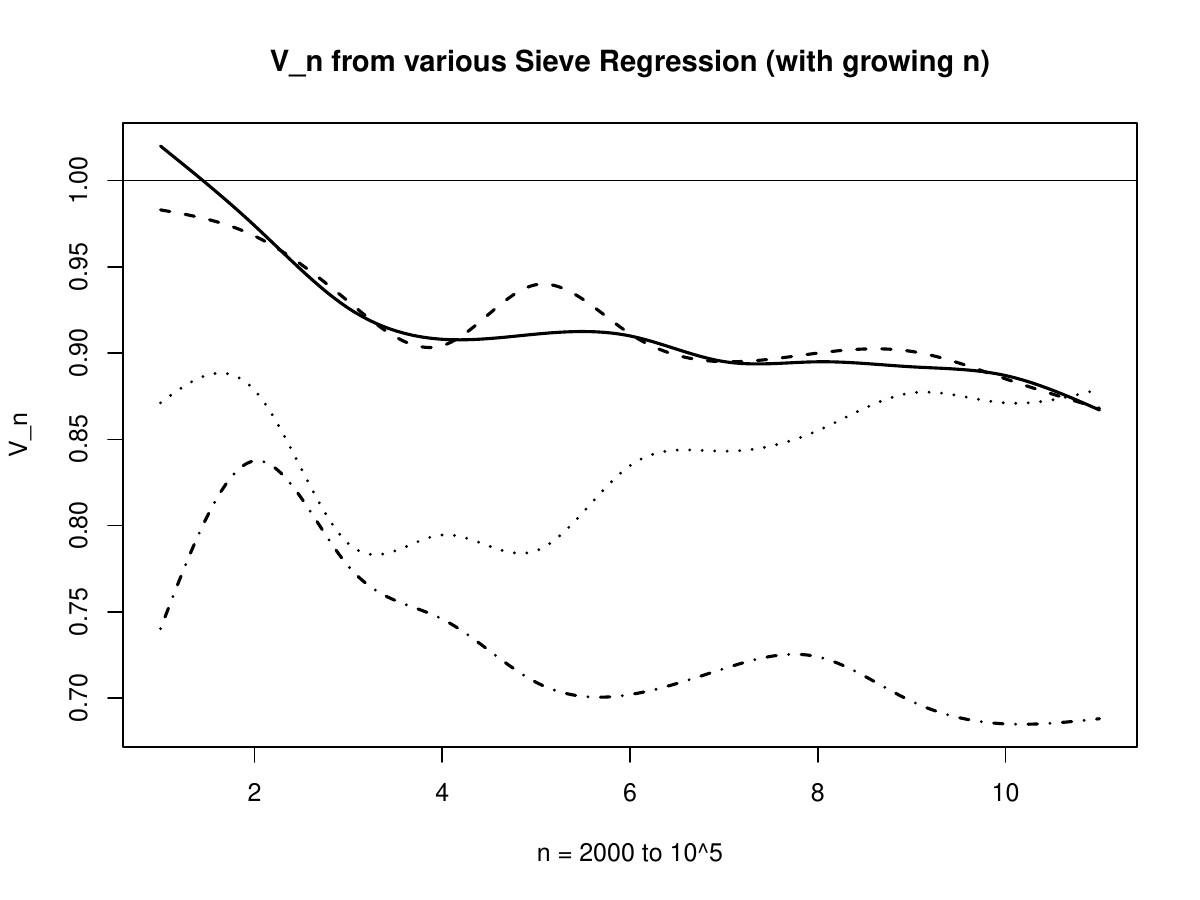}	
		\label{fig:V_sieve} 
	\end{subfigure}
	\caption{Simulated $\mathcal{V}_n$ for various DGPs with increasing $ n=m=2000,\ldots,10^5$ on horizontal axis}
	\label{V_n}
\end{figure}

\subsection*{E1: Multivariate Regression}
	The regressor $x_{t}$ is a collection of independent
	centered stationary AR(1) processes $x_{it}$ and their lags up to
	order 3. Thus, Assumption \ref{ass:M_diff} (i) is trivially satisfied. 
	For (ii), we note that $M$ and $\Omega$ are block diagonal, 
	implying that the minimum eigenvalues are bounded away from zero. 
	As for Assumption \ref{ass:M_diff} (iii), the usual maximal inequality for mixing, as in e.g. Lemma 7 
	in \cite{Linton2022}, 
	holds to yield the bounds $O_{p}\left(pn^{-1/2}\log n\right)$ on
	$\left\Vert n^{-1}\sum_{t=1}^{\left[nr\right]}x_{t}x_{t}'-r M\right\Vert $
	and $\left\Vert n^{-1}\sum_{t=1}^{\left[nr\right]}x_{t}x_{t}'\sigma_{t}^{2}-r\Omega\right\Vert $. 
	
	Turning to Assumption \ref{ass:MCLT}, first recall that the maximum of $ n $ random variables is $ O_p (\log n) $ if  the moment generating function of each random variable exists. Thus,
	\begin{eqnarray}
		E \max_{1\leq t \leq n} \bar{\lambda}(\xi_t \xi_t') 
	\leq  E \max_{1\leq t \leq n}(\xi_t' \xi_t)  
	\leq \sum_{i=1}^{p} E \max_{1\leq t \leq n} \xi_{it}^2
	= pO(\log n),  \label{max_lambda}
	\end{eqnarray}
	as desired for $ p=o\left(n^{1/3}/\log n\right) $ and similarly we can duduce the bound for $ E \left\vert E\left(\max_{1\leq t \leq n} (\xi_t' \xi_t)^2 | \mathcal{G}_{t-1} \right)\right\vert $, which equals $ E \left(\max_{1\leq t \leq n} (\xi_t' \xi_t)^2 \right) = p^2 O(\log n) $, due to the non-negativity of the square and the law of iterated expectations.
	
	Next, the two conditions on convergence in Assumption \ref{ass:MCLT} are difficult to derive analytically, so we present some numerical evidence in Figure \ref{V_n} and Table \ref{Table:cov} in the online appendix.
	The coefficient $ \alpha $  is chosen from $ 0, 0.3 $ or $ 0.8 $ for the experiments. 
	The AR coefficient in E1 is set as $ -0.5 $ or $ 0.5 $, the MA coefficient in E2 as $-0.5$ or $ -0.1 $, and the AR coefficient in E3 as 0.3, 0.5, or 0.7. 
	The expectations in \eqref{V_partial_def} are approximated by the average of 10000 iterations. 
	For each coefficient combination, we experiment with growing sample sizes from $ n=m=2000 $ to $ 10^5 $, $l=0$, and  $ p=n^{1/3} $ for E1 and E2 and $p=2n^{1/4}$ for E3. 
	The results are plotted as lines in the figure. 
	We note that the values do not diverge even for the most persistent case of $ \alpha = 0.8 $ and tend to stabilize for larger $ n $.
	On the other hand, the degeneracy of
	$(n^{4}p^{2})^{-1}\sum_{t=1}^{n}\sum_{s=1}^{t-1}cov\left(tr\left(\Upsilon_{t}\Xi_{t}\right),tr\left(\Upsilon_{s}\Xi_{s}\right)\right)$ in Assumption \ref{ass:MCLT} is given in Table \ref{Table:cov} for various parameter vales and error types. The table shows clear evidence of decaying covariances.
	
	Finally, Assumption \ref{ass:autocovandcumulant} is rather trivially met given that $\varepsilon_{t}$ is mds and all the moments exist for $\varepsilon_{t}\ $ and $x_{it}$ for all $i=1,...,p$.

\subsection*{E2: Infinite-order AR}
A set of primitive conditions is given in the next Proposition. \cite{gonccalves2007bootstrapping} has emphasized the empirical relevance of allowing for conditional heteroskedasticity in autoregressive models, which is allowed below by relaxing \cite{berk1974consistent}'s condition of an iid error to an mds process. 
Let $L$ be the lag operator and $b\left( L\right) =\sum_{j=1}^{\infty}b_{j}L^{j}$ denote the lag polynomial. 
The MA(1) with the coefficient less than the unity in modulus
satisfies the conditions in the next Proposition. 

\begin{proposition}
\label{thm:Berk}Suppose that    
(1) $b\left( z\right) \neq 0$ for any $ \left\vert z\right\vert \leq 1$ and $ b^{-1}(e^{i\lambda})  $ exists and is nonzero for $ -\pi <\lambda \leq \pi $. 
(2) $\left\{ \varepsilon _{t}\right\}  $ is a stationary mds that possesses a density of bounded variation, $E\varepsilon_{t}=0\ $and $E\left\vert \varepsilon _{t}\right\vert ^{\kappa }<C$ for some $ \kappa \geq 4$ and $ E(\varepsilon_t^2 | \mathcal{F}_{t-1}) $ is bounded and bounded away from zero. 
(3) $p^{3}=o\left( n\right) $. 
(4) $\sum_{j=p}^{\infty}\left\vert b_{j}\right\vert =o\left( n^{-1/2}\right) $.
(5) $ (\sigma_t, y_{t-1} )$ is $ \rho $-mixing with $ \sum_{j=1}^{\infty} \rho (2^j) <\infty $. 
Then, Assumptions \ref{ass:errors} - \ref{ass:M_diff}  %\ref{ass:autocovandcumulant} 
are satisfied with $\varkappa _{p}=v_{p}= o (p^{-1/2}) $.
\end{proposition}
\begin{proof}
\cite{berk1974consistent} established that the minimum eigenvalue of the limiting autocovariance matrix $M$ is bounded away from zero, see its equation (2.7), and the deviation bound for its sample autocovariance is $ o(p^{-1/2}) $ in its Lemma 3. 
As for $\Omega$, note that for some $ c>0 $, which is an a.s. lower bound of $ E\left(\varepsilon_{t}^{2}|\mathcal{F}_{t-1}\right) $,  and any $\left|a\right|=1$
\[
a'\Omega a=E\left(a'x_{t}\right)^{2}\varepsilon_{t}^{2}=E\left(a'x_{t}\right)^{2}E\left(\varepsilon_{t}^{2}|\mathcal{F}_{t-1}\right)\geq cE\left(a'x_{t}\right)^{2},
\]
to conclude that the minimum eigenvalue of $\Omega$ is also bounded
away from zero. 

Lemma 3.4  of \cite{peligrad1982invariance} yields $ E| \sum_t^n ( z_t^2 -Ez_t^2) | \leq 
n \sum_i \rho(2^i ) Ez_t^4 $ for a $ \rho $-mixing sequence $ z_t $. Since  $ z_t = \sigma_t y_{t-j} $ for $ j=1,...,p $ in the current case and $ (\sigma_t , y_{t-1} )$ is $ \rho $-mixing, the bound may be set as  $ n \sum_i \rho(0 \vee(2^i -p) ) Ez_t^4 \leq n (\sum_i \rho(2^i ) + \log p ) E\sigma_t^4 y_{t-j}^4  $ for any $ j\leq p $ Then, $ \| \bar{\Omega}(\gamma)- \Omega(\gamma) \| =O_p (n^{-1}p^2\log p ) $.    
The rest of the proof is given in Lemma \ref{lemma:Omega_hat_true}.
\end{proof}

The same comments as in E1 apply for Assumptions \ref{ass:MCLT} and \ref{ass:autocovandcumulant}. 

\subsection*{E3: Sieve regression}
Let $\mathbb{Z}\subseteq \mathbb{R}^{k}$ denote the support of $z_{t}$ in E3. 
The following proposition provides some more primitive conditions for E3 as given by \cite{Chen2015}
and the dgp in our Monte Carlo simulation satisfies them with a bounded support by construction and geometric mixing rate due to \cite{Meyn1993} as discussed above.

\begin{proposition}     \label{thm:CC} 
Suppose that the following hold: 
(1)\label{ass:beta_mixing_Tropp} The sequence $\left\{
z_{t}\right\} $ is strictly stationary and $\beta $-mixing with $\beta $%
-mixing coefficient $\beta (\cdot )$. Let $q=q\left( n\right) $ be a
sequence of integers satisfying $\beta (q)n/q\rightarrow 0$ as $n\rightarrow
\infty $ and $q\leq n/2$; (2) $ \mathbb{Z} $ is compact and rectangular, and \label{ass:xnormactual} $\sup_{z\in {\mathbb{Z}}}\left\Vert
x_{nt}(z)\right\Vert =O\left( \vartheta _{p}\right) $; (3)  \label{ass:CC-minlambda} The $x_t$ are tensor-products of power series, univariate polynomial spline, trigonometric polynomial wavelet or orthogonal
polynomial bases. Then, Assumptions \ref{ass:errors} -\ref{ass:M_diff} are met with $\varkappa _{p}=\vartheta _{p}\sqrt{q(\log p) / n}$ and $ v_{p}=\min \left\{ {p^{3}}/{n},{\vartheta_{p}^{2}p}/{n}\right\} $.
\end{proposition}
\begin{proof} We
prove that Assumption \ref{ass:M_diff} is met for the partial sum only, with the result for the full
sum following from Corollary 4.2 of \cite{Chen2015}. By Assumption \ref{ass:M_diff}, we can normalize the $x_{t}$ so that $E\left(
x_{t}x_{t}^{\prime }\right) =I_{p}$ without loss of generality. The result
then follows by Corollary \ref{cor:fuk_nag} by taking $\Xi
_{t,n}=n^{-1}\left( x_{t}x_{t}^{\prime }-I_{p}\right) $, which implies that the terms in Theorem \ref{thm:fuk_nagaev_ineq} have bounds: $%
R_{n}\leq n^{-1}\left( C\vartheta _{p}^{2}+1\right) $ and $s_{n}^{2}\leq
n^{-2}\left( C\vartheta _{p}^{2}+1\right) $. The second claim follows similarly.
The rest of the proof is given in Lemma \ref{lemma:Omega_hat_true}.
\end{proof}

The permissible mixing decay rate depends on the dimension $p$ of $x_{t}$: larger $p$ requires faster mixing decay.
Both exponential and geometric decays are allowed. See the discussions of Assumption 4 and Remark 2.3 in \cite{Chen2015} for more detailed discussion in relation to the sieve basis functions. The sequence $q$ depends on the mixing decay rate. For instance, if $\beta \left( q\right) $ decays at an exponential rate, $q$ can be set as $\log n$.
If all elements of $x_t\left( \cdot \right) $ are bounded, then 
$\vartheta _{p}=p^{1/2}$. % for condition \eqref{ass:xnormactual}. 
Under suitable conditions, it can be shown that $\vartheta _{p}=p$ for power series or orthogonal polynomials and $\vartheta
_{p}=p^{1/2} $ for univariate
polynomial splines, trigonometric polynomials or wavelets, see \cite{Newey1997,Chen2015}.
%The condition \eqref{ass:CC-minlambda} is met for most widely used series.

The same comments as in E1 apply for Assumptions \ref{ass:MCLT} and \ref{ass:autocovandcumulant}.

\newpage
\setcounter{section}{0}
\renewcommand\thesection{S.\Alph{section}}
\setcounter{equation}{0}
\renewcommand\theequation{S\Alph{section}.\arabic{equation}}
\setcounter{lemma}{0}
\renewcommand\thelemma{SL.\Alph{section}.\arabic{lemma}}
\setcounter{theorem}{0}
\renewcommand\thetheorem{ST.\Alph{section}.\arabic{theorem}}
\setcounter{corollary}{0}
\renewcommand\thecorollary{SC.\Alph{section}.\arabic{corollary}}
\setcounter{proposition}{0}
\renewcommand\theproposition{SP.\Alph{section}.\arabic{proposition}}
\setcounter{table}{0}
\renewcommand\thetable{S.Tab.\Alph{section}.\arabic{table}}

	\begin{center}
		\large{Online supplement to ``Robust Inference on Infinite and Growing Dimensional Time Series Regression''}
	\end{center}
	
	\begin{center}
		Abhimanyu Gupta and Myung Hwan Seo
	\end{center}
	
	\label{sec:app_lemmas}
	
	\section{An exponential inequality for partial sums of weakly dependent
		random matrices}
	
	\label{sec:partial_sum} We develop a stochastic order for a matrix partial
	sum. Closely related results can be found in Theorems 4.1 and 4.2 of \cite%
	{Chen2015}, who establish such bounds for full matrix sums as opposed to
	partial sums. Our first theorem is a Fuk-Nagaev type inequality, using a
	coupling approach similar to \cite{Dedecker2004}, \cite{Chen2015} and \cite%
	{Rio2017}.
	
	\begin{theorem}
		\label{thm:fuk_nagaev_ineq} Let $\left \{\xi_i\right \}_{i\in {\mathbb{Z}}}$
		be a $\beta $-mixing sequence with support ${\mathfrak{X}}$ and $r$-th
		mixing coefficient $\beta (r)$ and let $\Xi _{i,n}=\Xi _{n}\left (\xi
		_i\right )$, for each $i$, where $\Xi _n:{\mathfrak{X}}\rightarrow {\mathbb{R%
		}}^{d_1\times d_2}$ is a sequence of measurable $d_1\times d_2$
		matrix-valued functions. Assume $E\left (\Xi _{i,n}\right )=0$ and $%
		\left
		\Vert \Xi _{i,n}\right \Vert \leq R_n$, for each $i$, set 
		\begin{equation*}
			s_n^2=\max _{1\leq i,j\leq n}\max \left \{\left \Vert E\left (\Xi _{i,n}\Xi
			_{j,n}^{\prime }\right )\right \Vert ,\left \Vert E\left (\Xi _{i,n}^{\prime
			}\Xi _{j,n}\right \Vert \right )\right \},
		\end{equation*}
		and define $S_k=\sum _{l=1}^k\Xi _{l,n}$. Then, for any integer $q$ such
		that $1<q\leq n/2$ and $\varrho \geq qR_n$, 
		\begin{equation*}
			P\left (\sup _{1\leq k\leq n}\left \Vert S_k\right \Vert >4\varrho \right
			)\leq \left (\left [\frac{n}{q}\right ]+1\right )\beta (q)+2\left
			(d_1+d_2\right ) \func{exp}\left (\frac{-\varrho ^2/2}{nqs_n^2+qR_n\varrho /3%
			}\right ).
		\end{equation*}
	\end{theorem}
	
	The required stochastic order now follows by a choice of $\varrho $ in
	Theorem \ref{thm:fuk_nagaev_ineq}:
	
	\begin{corollary}
		\label{cor:fuk_nag} Under the conditions of Theorem \ref{thm:fuk_nagaev_ineq}%
		, if $q$ is chosen as a function of $n$ such that $\left (n/q\right )\beta
		(q)=o(1)$ and $R_n\sqrt{q\log \left (d_1+d_2\right )}=o\left (s_n\sqrt{n}%
		\right )$ then 
		\begin{equation*}
			\sup _{1\leq k\leq n}\left \Vert S_k\right \Vert =O_p\left (s_n\sqrt{nq\log
				\left (d_1+d_2\right )}\right )
		\end{equation*}
	\end{corollary}

	\begin{proof}[Proof of Theorem \protect\ref{thm:fuk_nagaev_ineq}]
		For $i=1,\ldots ,[n/q]$, define $U_{i}=\sum_{j=iq-q+1}^{iq}\Xi _{j,n}$ and $%
		U_{[n/q]+1}=\sum_{j=[n/q]q}^{n}\Xi _{j,n}$. Now, for an integer $j$ that
		differs from an integer multiple of $q$ by at most $[q/2]$, we have 
		$
		\sup_{1\leq k\leq n}\left\Vert S_{k}\right\Vert \leq
		2[q/2]R_{n}+\sup_{j>0}\left\Vert \sum_{i=1}^{j}U_{i}\right\Vert .
		$
		If $q$ is even (respectively odd) then $q=2k$ (resp. $q=2k+1$) for some
		positive integer $k$, implying $[q/2]=[2k/2]=k$ (resp. $[q/2]=[(2k+1)/2]=k$)
		whence $2[q/2]R_{n}\leq qR_{n}$ (resp. $2[q/2]R_{n}\leq (q-1)R_{n}$). Thus,
		because $\varrho \geq qR_{n}$, 
		\begin{eqnarray}
			P\left( \sup_{1\leq k\leq n}\left\Vert S_{k}\right\Vert >4\varrho \right)
			&\leq &P\left( 2[q/2]R_{n}>\varrho \right) +P\left( \sup_{j>0}\left\Vert
			\sum_{i=1}^{j}U_{i}\right\Vert \geq 3\varrho \right)  \notag \\
			&=&P\left( \sup_{j>0}\left\Vert \sum_{i=1}^{j}U_{i}\right\Vert \geq 3\varrho
			\right) ,  \label{partial_sum1}
		\end{eqnarray}%
		so it suffices to prove that 
		\begin{equation*}
			P\left( \sup_{j>0}\left\Vert \sum_{i=1}^{j}U_{i}\right\Vert \geq 3\varrho
			\right) \leq \left( \left[ \frac{n}{q}\right] +1\right) \beta (q)+2\left(
			d_{1}+d_{2}\right) \func{exp}\left( \frac{-\varrho ^{2}/2}{%
				nqs_{n}^{2}+qR_{n}\varrho /3}\right) .
		\end{equation*}%
		Enlarging the probability space as needed, by Lemma 5.1 (Berbee's Lemma) of 
		\cite{Rio2017} there is a sequence $\xi _{i}^{\ast }$, $1\leq i\leq \lbrack
		n/q]+1$, such that
		
		\begin{enumerate}[(a)]

			\item The random variable $x_{i}^{\ast }$ is distributed as $x_{i}$ for each 
			$1\leq i\leq \lbrack n/q]+1$.
			
			\item The sequences $\xi _{2i}^{\ast }$, $1\leq 2i\leq \lbrack n/q]+1$, and $%
			\xi _{2i-1}^{\ast }$, $1\leq 2i-1\leq \lbrack n/q]+1$, comprise of
			independent random variables.
			
			\item $P\left( \xi _{i}\neq \xi _{i}^{\ast }\right) \leq \beta (q+p)$ for $%
			1\leq i\leq \lbrack n/q]+1$.
		\end{enumerate}
		
		Denote $\Xi _{i,n}^{\ast }=\Xi _{n}\left( \xi _{i}^{\ast }\right) $, and
		define $U_{i}^{\ast }$ in the obvious manner. Then, we have 
		\begin{equation}
			\sup_{j>0}\left\Vert \sum_{i=1}^{j}U_{i}\right\Vert \leq
			\sum_{i=1}^{[n/q]+1}\left\Vert U_{i}-U_{i}^{\ast }\right\Vert
			+\sup_{j>0}\left\Vert \sum_{i=1}^{j}U_{2i}^{\ast }\right\Vert
			+\sup_{j>0}\left\Vert \sum_{i=1}^{j}U_{2i-1}^{\ast }\right\Vert .
			\label{partial_sum2}
		\end{equation}
		
		Now, by (c), we have 
		\begin{eqnarray*}
			P\left( \sum_{i=1}^{[n/q]+1}\left\Vert U_{i}-U_{i}^{\ast }\right\Vert \geq
			\varrho \right) &=&P\left( \sum_{i=1}^{[n/q]+1}\left\Vert U_{i}-U_{i}^{\ast
			}\right\Vert \geq \left[ \sum_{i=1}^{[n/q]+1}\varrho /\left( [n/q]+1\right) %
			\right] \right) \\
			&\leq &\sum_{i=1}^{[n/q]+1}P\left( \left\Vert U_{i}-U_{i}^{\ast }\right\Vert
			\geq \varrho /\left( [n/q]+1\right) \right) \\
			&\leq &\left( [n/q]+1\right) \beta (q+p),
		\end{eqnarray*}%
		while for all $1\leq i\leq \lbrack n/q]+1$ the matrices $U_{i}^{\ast
		}=\sum_{j=iq-q+1}^{iq}\Xi _{j,n}^{\ast }$ satisfy $\left\Vert U_{i}^{\ast
		}\right\Vert \leq qR_{n}$ and 
		\begin{equation*}
			\max_{1\leq j\leq n}\max \left\{ \left\Vert E\left(
			\sum_{i=1}^{j}U_{i}U_{i}^{\ast ^{\prime }}\right) \right\Vert ,\left\Vert
			E\left( \sum_{i=1}^{j}U_{i}^{\ast ^{\prime }}U_{i}^{\ast }\right)
			\right\Vert \right\} \leq nqs_{n}^{2}.
		\end{equation*}%
		Furthermore, the sequence ${\mathcal{U}}_{j}=\sum_{i=1}^{j}U_{2i}^{\ast }$
		is a matrix martingale (because $U_{2i}^{\ast }$ is an independent sequence
		and $E{\mathcal{U}}_{j}=0$) with difference sequence ${\mathcal{U}}_{j}-{%
			\mathcal{U}}_{j-1}=U_{2j}^{\ast }$. Thus, by Corollary 1.3 of \cite%
		{Tropp2011}, 
		\begin{equation}
			P\left( \sup_{j>0}\left\Vert \sum_{i=1}^{j}U_{2i}^{\ast }\right\Vert \geq
			\varrho \right) \leq \left( d_{1}+d_{2}\right) \func{exp}\left( \frac{%
				-\varrho ^{2}/2}{nqs_{n}^{2}+qR_{n}\varrho /3}\right) .  \label{partial_sum3}
		\end{equation}%
		The third term on the RHS of (\ref{partial_sum2}) is bounded similarly,
		whence the claim follows.
	\end{proof}
	
	\begin{proof}[Proof of Corollary \protect\ref{cor:fuk_nag}]
		In Theorem \ref{thm:fuk_nagaev_ineq}, take $\varrho =Cs_{n}\sqrt{nq\log
			\left( d_{1}+d_{2}\right) }$ for a sufficiently large constant $C$. Then the
		claim follows by the condition $\left( n/q\right) \beta (q)=o(1)$ and
		because $R_{n}\sqrt{q\log \left( d_{1}+d_{2}\right) }=o\left( s_{n}\sqrt{n}%
		\right) $. To verify that $\varrho $ satisfies that requirement of Theorem %
		\ref{thm:fuk_nagaev_ineq}, note that the latter condition implies $Cs_{n}%
		\sqrt{n}\geq R_{n}\sqrt{q\log \left( d_{1}+d_{2}\right) }$ for sufficiently
		large $n$, so $\varrho \geq qR_{n}\log (d_{1}+d_{2})\geq qR_{n}$ for
		sufficiently large $n$, assuming $d_{1}+d_{2}\geq e\approx 2.72$. The latter
		condition fails only if the $\Xi _{i,n}$ are scalar.
	\end{proof}

	\section{For Section \ref{sec:theory}}\label{sec:weak convergence}
	
	We first present an initial approximation of ${\mathcal{Q}}_{n}(\gamma )$. 
	\begin{theorem}
		\label{theorem:struc_break_approx} Let Assumptions \ref{ass:errors}-\ref{ass:M_diff} hold, and 
		\begin{equation}
			\lambda_n^{-4}\sqrt{p}\left(\lambda_n^{-1} \varkappa _{p}+v_{p}\right) +\lambda_n^{-6}p^{-1}\rightarrow 0{\text{ as }}%
			n\rightarrow \infty , \label{rate:struc_break_approx}
		\end{equation}%
		Then, 
		${\mathcal{Q}}_{n}(\gamma )-\left({{\mathcal{R}_{n}}(\gamma )-p}\right)/{\sqrt{2p}} = o_{p}(1)$.   
	\end{theorem}
	\begin{proof}
		Much of the details are delegated to Lemmas \ref{lemma:Mhat_norm}-\ref{lemma:fancyS_lemma2}. In particular,
		we  show in Lemma \ref{lemma:wald_approx} that 
		\begin{equation}
			{\mathcal{Q}}_{n}(\gamma )=\frac{n\varepsilon ^{\prime }A(\gamma )^{\prime
				}B(\gamma )^{-1}A(\gamma )\varepsilon -p}{\sqrt{2p}}+o_{p}(1).
			\label{general_approx}
		\end{equation}%
		
		Then, note that 
		\begin{equation}
			\left( X^{\ast }(\gamma )^{\prime }M_{X}X^{\ast }(\gamma )\right)
			^{-1}=n^{-1}\left( I-\hat{M}^{-1}\hat{S}(\gamma )\right) ^{-1}\hat{S}(\gamma
			)^{-1},  \label{fancyR1}
		\end{equation}%
		and 
		\begin{equation}
			X^{\ast }(\gamma )^{\prime }M_{X}\varepsilon =X^{\ast }(\gamma )^{\prime
			}\varepsilon -\hat{S}(\gamma )\hat{M}^{-1}X^{\prime }\varepsilon ,
			\label{fancyR2}
		\end{equation}%
		because $n^{-1}X^{\ast }(\gamma )^{\prime }X=\hat{S}(\gamma )$. Using (\ref%
		{fancyR1}) and (\ref{fancyR2}), we may write $n\varepsilon ^{\prime
		}A(\gamma )^{\prime }B(\gamma )^{-1}A(\gamma )\varepsilon /\sqrt{2p}$ as 
		\begin{equation}
			\frac{n^{-1}R_{1}(\gamma )^{\prime }R_{2}(\gamma )^{\prime }B(\gamma
				)^{-1}R_{2}(\gamma )R_{1}(\gamma )}{\sqrt{2p}}  \label{fancyR3}
		\end{equation}%
		where $R_{1}(\gamma )=\hat{S}(\gamma )^{-1}X^{\prime \ast }(\gamma
		)\varepsilon -\gamma (1-\gamma )^{-1}\hat{M}^{-1}X^{\prime }\varepsilon $ and $R_{2}(\gamma
		)=\left( I-\hat{M}^{-1}\hat{S}(\gamma )\right) ^{-1}$. By adding and
		subtracting terms we can decompose (\ref{fancyR3}) as $\sum_{i=1}^{4}\Delta
		_{i}(\gamma )+\overline{{\mathcal{R}_{n}}}(\gamma )$, with%
		
		\begin{eqnarray*}
			\Delta _{1}(\gamma ) &=&\frac{\left( R_{1}(\gamma )-\overline{R}_{1}(\gamma
				)\right) ^{\prime }R_{2}(\gamma )^{\prime} B(\gamma )^{-1} R_{2}(\gamma )R_{1}(\gamma )}{n%
				\sqrt{2p}}, \\
			\Delta _{2}(\gamma ) &=&\frac{\overline{R}_{1}(\gamma )^{\prime
				}R_{2}(\gamma )^{\prime }B(\gamma )^{-1}R_{2}(\gamma )\left( R_{1}(\gamma )-\overline{R}%
				_{1}(\gamma )\right) }{n\sqrt{2p}}, \\
			\Delta _{3}(\gamma ) &=&\frac{\overline{R}_{1}(\gamma )^{\prime }\left(
				R_{2}(\gamma )-\gamma ^{-1}I\right) ^{\prime }B(\gamma )^{-1}R_{2}(\gamma )\overline{R}%
				_{1}(\gamma )}{n\sqrt{2p}}, \\
			\Delta _{4}(\gamma ) &=&\frac{\overline{R}_{1}(\gamma )^{\prime }B(\gamma)^{-1}\left(
				R_{2}(\gamma )-\gamma ^{-1}I\right) \overline{R}_{1}(\gamma )}{\gamma n\sqrt{%
					2p}},
		\end{eqnarray*}%
		where we write $\overline{R}_{1}(\gamma )=(1-\gamma )^{-1}\hat{M}^{-1}\left(
		\gamma \sum_{t=1}^{n}\varepsilon _{t}x_{t}-\sum_{t=1}^{[n\gamma
			]}\varepsilon _{t}x_{t}\right) $ and 
		\begin{equation}
			\overline{{\mathcal{R}_{n}}}(\gamma )=\frac{\left( \sum_{t=1}^{[n\gamma
					]}\varepsilon _{t}x_{t}-\gamma \sum_{t=1}^{n}\varepsilon _{t}x_{t}\right)
				^{\prime }\hat{M}^{-1}B(\gamma )^{-1}\hat{M}^{-1}\left( \sum_{t=1}^{[n\gamma
					]}\varepsilon _{t}x_{t}-\gamma \sum_{t=1}^{n}\varepsilon _{t}x_{t}\right) }{%
				\gamma ^{2}\left( 1-\gamma \right) ^{2}n\sqrt{2p}}.  \label{fancyR_bar}
		\end{equation}
		
		By (\ref{B_gamma1}), the term sandwiched between the parentheses in the
		numerator of (\ref{fancyR_bar}) is 
		\begin{equation}
			\left( \hat{M}^{-1}-M^{-1}\right) B(\gamma )^{-1}\hat{M}^{-1}+M^{-1}B(\gamma
			)^{-1}\left( \hat{M}^{-1}-M^{-1}\right) +\gamma (1-\gamma )\Omega ^{-1}.
			\label{midterm_norm}
		\end{equation}%
		Substituting (\ref{midterm_norm}) into (\ref{fancyR_bar}) yields three terms
		corresponding to the three terms in (\ref{midterm_norm}). The first of these, multiplied by the outside terms in the sandwich formula in \eqref{fancyR_bar},
		has modulus bounded by a constant times 
		\begin{equation*}
			\frac{n^{-1}\left( \left\Vert X^{\prime
				}\varepsilon \right\Vert ^{2}+\left\Vert X^{\ast }(\gamma )^{\prime
				}\varepsilon \right\Vert ^{2}\right) \left\Vert \hat{M}-M\right\Vert
				\left\Vert B(\gamma )^{-1}\right\Vert \left\Vert M^{-1}\right\Vert
				\left\Vert \hat{M}^{-1}\right\Vert ^{2}}{\sqrt{p}}=O_{p}\left(\lambda_n^{-4} \sqrt{p}\varkappa
			_{p}\right) ,
		\end{equation*}%
		by Assumption \ref{ass:M_diff} and Lemmas \ref%
		{lemma:Mhat_norm}, \ref{lemma:B_eigs}, and also (\ref{lemma2_first_expec}),
		while the second is similarly shown to be negligible also. By (\ref{rate:struc_break_approx}), we conclude that 
		\begin{equation*}
			\overline{{\mathcal{R}_{n}}}(\gamma )={{\mathcal{R}_{n}}}(\gamma )+o_{p}(1),
		\end{equation*}%
		indicating that the theorem is proved if $\Delta _{i}(\gamma )=o_{p}(1)$, $%
		i=1,2,3,4$. But by previously used techniques and Lemmas \ref%
		{lemma:fancyS_lemma1} and \ref{lemma:fancyS_lemma2}, we readily conclude
		that 
		\begin{equation*}
			\left( \Delta _{1}(\gamma ),\Delta _{2}(\gamma
			),\Delta _{3}(\gamma ),\Delta _{4}(\gamma )\right) =O_{p}\left( \lambda_n^{-5}{\sqrt{p}\varkappa
				_{p}}\right)
		\end{equation*}%
		which are all negligible by (\ref{rate:struc_break_approx}), proving the
		theorem.
	\end{proof}

	Write $\tilde{\Omega}(\gamma )=n^{-1}\sum_{t=1}^{n}x_{t}(\gamma
	)x_{t}^{\prime }(\gamma )\varepsilon _{t}^{2}$.
	
	\begin{lemma}
		\label{lemma:Omega_hat_true} Under Assumptions \ref{ass:errors}-\ref{ass:M_diff}, and the conditions of Propositions \ref{thm:CC} or \ref{thm:Berk} as applicable,
		\begin{align}
			\sup_{\gamma \in \Gamma }\left\Vert \hat{\Omega}(\gamma )-\tilde{\Omega}%
			(\gamma )\right\Vert & =O_{p}\left( \lambda_n^{-2}\min \left\{ \frac{p^{3}}{n},\frac{%
				\vartheta _{p}^{2}p}{n}\right\} \right) ,  \label{Omegadiff1} \\
			\sup_{\gamma \in \Gamma }\left\Vert \tilde{\Omega}(\gamma )-\bar{\Omega}%
			(\gamma )\right\Vert & =O_{p}\left( \frac{p}{\sqrt{n}}\right) .
			\label{Omegadiff2}
		\end{align}
	\end{lemma}
	
	\begin{proof}[Proof of Lemma \protect\ref{lemma:Omega_hat_true}]
		The matrix inside the norm on the LHS of (\ref{Omegadiff1}) can be
		decomposed as $\sum_{i=1}^{5}U_{i}(\gamma )$, with 
		\begin{align*}
			U_{1}(\gamma )& =n^{-1}\sum_{t=1}^{n}x_{t}(\gamma )x_{t}^{\prime }(\gamma ) 
			\left[ x_{t}^{\prime }(\gamma )\left( \delta -\hat{\delta}(\gamma )\right) %
			\right] ^{2}, \\
			U_{2}(\gamma )& =n^{-1}\sum_{t=1}^{n}x_{t}(\gamma )x_{t}^{\prime }(\gamma
			)r_{t}^{2}, \\
			U_{3}(\gamma )& =2n^{-1}\sum_{t=1}^{n}x_{t}(\gamma )x_{t}^{\prime }(\gamma ) 
			\left[ x_{t}^{\prime }(\gamma )\left( \delta -\hat{\delta}(\gamma )\right) %
			\right] \varepsilon _{t}, \\
			U_{4}(\gamma )& =2n^{-1}\sum_{t=1}^{n}x_{t}(\gamma )x_{t}^{\prime }(\gamma ) 
			\left[ x_{t}^{\prime }(\gamma )\left( \delta -\hat{\delta}(\gamma )\right) %
			\right] r_{t}, \\
			U_{5}(\gamma )& =2n^{-1}\sum_{t=1}^{n}x_{t}(\gamma )x_{t}^{\prime }(\gamma
			)r_{t}\varepsilon _{t}.
		\end{align*}
		
		Recall Lemma \ref{lemma:deltahat} for $\sup_{\gamma \in \Gamma }\left\Vert
		\delta -\hat{\delta}(\gamma )\right\Vert =O_{p}\left(\lambda_n^{-1} \sqrt{p/n}\right) .$
		Now, since the maximum eigenvalue of a non-negative definite symmetric
		matrix is less than equal to the trace, 
		\begin{eqnarray*}
			\left\Vert U_{1}(\gamma )\right\Vert &\leq &n^{-1}\sum_{t=1}^{n}\left(
			x_{t}^{\prime }(\gamma )x_{t}(\gamma )\right) ^{2}\left( \delta -\hat{\delta}%
			(\gamma )\right) ^{\prime }\left( \delta -\hat{\delta}(\gamma )\right) \\
			&\leq &2pn^{-1}\sum_{t=1}^{n}\sum_{j=1}^{p}x_{tj}^{4}\left\Vert \delta -\hat{%
				\delta}(\gamma )\right\Vert ^{2} = O_{p}\left(\lambda_n^{-2} p^{2}\right) O_{p}\left( p/n\right) ,
		\end{eqnarray*}%
		uniformly in $\gamma $, by the fact that $\sup_{t,j}Ex_{tj}^{4}<\infty $ and
		(\ref{delta_deltahat_order}). In a similar fashion, 
		\begin{equation*}
			E\left\Vert U_{2}(\gamma )\right\Vert \leq
			2En^{-1}\sum_{t=1}^{n}x_{t}^{\prime }x_{t}r_{t}^{2}\leq 2\left( E\left(
			x_{t}^{\prime }x_{t}\right) ^{2}Er_{t}^{4}\right) ^{1/2}=O\left( \lambda_n^{-2}p/\sqrt{n}%
			\right) .
		\end{equation*}%
		Similarly and using the fact that $E\left( \left\vert
		\varepsilon _{t}\right\vert |x_{t}\right) \leq \sqrt{E\left( \varepsilon
			_{t}^{2}|x_{t}\right) }=O\left( 1\right) ,$ we obtain 
		\begin{align*}
			\left\Vert U_{3}(\gamma )\right\Vert & \leq 4n^{-1}\sum_{t=1}^{n}\left(
			x_{t}^{\prime }x_{t}\right) ^{2}\left\vert \varepsilon _{t}\right\vert
			\left\Vert \delta -\hat{\delta}(\gamma )\right\Vert ^{2}=O_{p}\left(
			\lambda_n^{-2}p^{3}/n\right) , \\
			\left\Vert U_{4}(\gamma )\right\Vert & \leq 4n^{-1}\sum_{t=1}^{n}\left(
			x_{t}^{\prime }x_{t}\right) ^{3/2}\left\Vert \delta -\hat{\delta}(\gamma
			)\right\Vert \left\vert r_{t}\right\vert \\
			& \leq 4\left\Vert \delta -\hat{\delta}(\gamma )\right\Vert \left(
			n^{-1}\sum_{t=1}^{n}\left( x_{t}^{\prime }x_{t}\right) ^{2}\right)
			^{3/4}\left( n^{-1}\sum_{t=1}^{n}r_{t}^{4}\right) ^{1/4} 
			= O_p\left( \sqrt{\frac{p}{n}}p^{3/2}\frac{\lambda_n^{-1}}{n^{1/4}}\right) , \\
			\left\Vert U_{5}(\gamma )\right\Vert & =2n^{-1}\sum_{t=1}^{n}\left(
			x_{t}^{\prime }x_{t}\right) \left\vert r_{t}\varepsilon _{t}\right\vert
			=O_{p}\left( p/\sqrt{n}\right) ,
		\end{align*}%
		all uniformly in $\Gamma $. Thus (\ref{Omegadiff1}) is established.
		
		To show (\ref{Omegadiff2}), let $x_{it}$, $i=1,\ldots ,p$, be a typical
		element of $x_{t}$. Then any element of $\tilde{\Omega}(\gamma )-\bar{\Omega}%
		(\gamma )$ is of the form $n^{-1}\sum_{t=1}^{n}x_{it}(\gamma )x_{jt}(\gamma
		)\left( \varepsilon _{t}^{2}-\sigma _{t}^{2}\right) $, $i,j=1,\ldots ,p$,
		and $\varepsilon _{t}^{2}-\sigma _{t}^{2}$ is an MDS by construction. Thus,
		it has mean zero and variance $
		n^{-2}\sum_{t=1}^{n}Ex_{it}^{2}(\gamma )x_{jt}^{2}(\gamma )E\left( \left(
		\varepsilon _{t}^{2}-\sigma _{t}^{2}\right) ^{2}|\mathcal{F}_{t-1}\right)
		=O_{p}\left( n^{-1}\right)$, 
		by Assumption \ref{ass:errors} and the boundedness of $Ex_{it}^{4}$. Thus, 
		$
		E\left\Vert \tilde{\Omega}(\gamma )-\bar{\Omega}(\gamma )\right\Vert
		^{2}=O\left( p^{2}/n\right) ,
		$ and the claim in (\ref{Omegadiff2}) follows by Markov's inequality.
	\end{proof}
	
	We  establish asymptotic normality of 
	\begin{equation}
		{\mathcal{S}}_{n}(\gamma )=\frac{n^{-1}\sum_{s\neq t}g_{t}(\gamma )^{\prime
			}\Omega ^{-1}g_{s}(\gamma )\varepsilon _{t}\varepsilon _{s}}{\gamma \left(
			1-\gamma \right) \sqrt{2p}},  \label{dist_target}
	\end{equation}%
	recalling that $g_{t}(\gamma )=x_t 1\left\{t/n\leq\gamma\right\}-\gamma x_t$.
	
	\begin{theorem}
		\label{thm:T_weak_conv} Under Assumptions \ref{ass:errors}-\ref{ass:autocovandcumulant} and (\ref{rate:struc_break_approx}), $
		{\mathcal{S}}_{n}(\gamma )\overset{d}{\rightarrow} \sqrt{\mathcal{V}}\mathcal{Q}(\gamma ),{\text{ as }}%
		n\rightarrow \infty ,$ pointwise in $\gamma$.
	\end{theorem}
	
	\begin{proof}[Proof of Theorem \protect\ref{thm:T_weak_conv}]
		First, note that ${\mathcal{S}}_{n}(\gamma )$ equals $\left[ \gamma \left(
		1-\gamma \right) \sqrt{2p}\right] ^{-1}$ times 
		\begin{equation*}
			\frac{1}{n}\underset{s\neq t}{\sum_{s,t=1}^{[n\gamma ]}}x_{t}^{\prime
			}\Omega ^{-1}x_{s}\varepsilon _{t}\varepsilon _{s}-\frac{2\gamma }{n}%
			\underset{s\neq t}{\sum_{s=1}^{n}\sum_{t=1}^{[n\gamma ]}}x_{t}^{\prime
			}\Omega ^{-1}x_{s}\varepsilon _{t}\varepsilon _{s}+\frac{\gamma ^{2}}{n}%
			\underset{s\neq t}{\sum_{s,t=1}^{n}}x_{t}^{\prime }\Omega
			^{-1}x_{s}\varepsilon _{t}\varepsilon _{s}
		\end{equation*}%
		and thus 
		\begin{align*}
			{\mathcal{S}}_{n}(\gamma )& =\frac{\sqrt{2}}{\gamma \left( 1-\gamma \right) }%
			\left[ \mathcal{A}_{n}(\gamma )-\gamma \left[ \mathcal{A}_{n}(1)+\mathcal{A}_{n}(\gamma )-\bar{\mathcal{A}}%
			_{n}(\gamma )\right] +\gamma ^{2}\mathcal{A}_{n}(1)\right] , \\
			& =\sqrt{2}\left( \frac{\mathcal{A}_{n}(\gamma )}{\gamma }+\frac{\bar{\mathcal{A}}_{n}\left(
				\gamma \right) }{\left( 1-\gamma \right) }-\mathcal{A}_{n}(1)\right) ,
		\end{align*}%
		where 
		\begin{align*}
			\mathcal{A}_{n}(\gamma )& =\frac{1}{n\sqrt{p}}{\sum_{s=2}^{[n\gamma ]}}{%
				\sum_{t=1}^{s-1}}\xi _{t}^{\prime }\xi _{s}, \\
			\bar{\mathcal{A}}_{n}(\gamma )& =\frac{1}{n\sqrt{p}}{\sum_{s=[n\gamma
					]+1}^{n}\sum_{t=[n\gamma ]+1}^{s-1}}\xi _{t}^{\prime }\xi _{s},
		\end{align*}%
		and $\xi _{t}=\left\{ \xi _{ti}\right\} _{i=1}^{p}=\Omega
		^{-1/2}x_{t}\varepsilon _{t}$ being an mds.

		%Due to their symmetric nature, the tightness proof
		%is almost the same for both processes. We elaborate the tightness
		%of $\mathcal{A}_{n}\left(\gamma\right),$ for which we note that $\mathcal{A}_{n}\left(\gamma\right)$
		%is a partial sum process of a heterogeneous martingale difference
		%array $w_{ns}=\xi_{s}^{\prime}{\sum_{t=1}^{s-1}}\xi_{t}/\sqrt{np}$,
		%and thus it is sufficient to show 
		%\begin{eqnarray}
		%E\left\vert \mathcal{A}_{n}\left(\gamma_{1}\right)-\mathcal{A}_{n}\left(\gamma_{2}\right)\right\vert ^{4} & = & E\left\vert \frac{1}{\sqrt{n}}{\sum_{s=\left[n\gamma_{1}\right]+1}^{[n\gamma_{2}]}}w_{ns}\right\vert ^{4}\nonumber \\
		%& \leq & E\left(\sum_{s}\mathrm{E}\left(w_{ns}^{2}|\mathcal{F}_{s-1}\right)/n\right)^{2}+n^{-1}\max_{s}E\left\vert w_{ns}\right\vert ^{4}O\left(\left\vert \gamma_{2}-\gamma_{1}\right\vert \right)\label{eq:max_w_s}\\
		%& = & O\left(\left\vert \gamma_{2}-\gamma_{1}\right\vert \right),\nonumber 
		%\end{eqnarray}
		%where we apply the Rosenthal inequality, e.g. \cite{Hall1980}, for the inequality and a calculation similar to (\ref{clt:var_cond}) and \eqref{cltvarbd1}  for the last equality. Specifically, 
		%\begin{equation}\label{tight_argument}
		%n^{-1}\max_{s}E\left\vert w_{ns}\right\vert ^{4}\leq \max_s E\left( E((\xi_{s}'\xi_{s})^2 | \mathcal{F}_{s-1}) \left( \sum_{t_1 ,t_2 <s} \xi_{t_1}'\xi_{t_2}\right)^{2}\right)n^{-3}p^{-2}=o(1),
		%\end{equation} 
		%by Assumption \ref{ass:MCLT} and the same reasoning as for \eqref{target} and  \eqref{cltvarbd1}.

		Next, we check the conditions of Corollary 3.1 in \cite{Hall1980} for $\mathcal{A}_{n}(\gamma )$, with similar steps holding for $\bar{\mathcal{A}_{n}}(\gamma )$ due to the symmetric nature of the processes. Writing $w_{ns}=\xi_{s}^{\prime}{\sum_{t=1}^{s-1}}\xi_{t}/\sqrt{np}$ (a heterogeneous martingale difference array), we first check the second condition therein, viz.
		\begin{equation}\label{clt:var_cond}
			\sum_{s}E\left(w_{ns}^{2}|\mathcal{G}_{s-1}\right)/n-\mathcal{V}/2\overset{p}{\rightarrow}0.
		\end{equation}
		Let $\Delta_s=\Upsilon_s\Xi_s$. Then we want to show  $n^{-2}p^{-1}\sum_s\tr \Delta_s-\mathcal{V}/2\overset{p}{\rightarrow}0$
		but, because $n^{-2}p^{-1}\sum_s E\tr\Delta_s\rightarrow \mathcal{V}/2$, it suffices to show
		\begin{equation}\label{target}
			n^{-2}p^{-1}\sum_s\left(\tr\Delta_s-E\tr\Delta_s\right)\overset{p}{\rightarrow}0.
		\end{equation}
		By Assumption \ref{ass:MCLT}, $P\left\{ Z_{n}\leq n^\nu\right\} \to1$,
		where $Z_{n}=\max_{1\leq t\leq n}\overline{\lambda}\left(\Upsilon_t\right)$. Then, to prove (\ref{target}), let $d_{n}=I\left\{ Z_{n}\leq n^\nu\right\} $
		and $K_{n}$ denote the LHS in (\ref{target}). Write $K_{n}=K_{n}d_{n}+K_{n}\left(1-d_{n}\right)$ and note that
		$K_{n}\left(1-d_{n}\right)=o_{p}\left(1\right)$ since $P\left\{ \left\Vert Z_{n}K_{n}\left(1-d_{n}\right)\right\Vert >0\right\} \leq P\left\{ d_{n}\neq1\right\} \to0$. 
		
		Thus it suffices to show (\ref{target}) for $d_n=1$. The LHS of (\ref{target}) has variance
		\begin{equation}	\label{vartarget}
			n^{-4}p^{-2}\sum_s E\left(\tr\Delta_s-E\tr\Delta_s\right)^2
			+2n^{-4}p^{-2}\sum_{s_1< s_2}E\left(\left(\tr\Delta_{s_1}-E\tr\Delta_{s_1}\right)\left(\tr\Delta_{s_2}-E\tr\Delta_{s_2}\right)\right).
		\end{equation}
		The first term in (\ref{vartarget}) is bounded by $n^{-4}p^{-2}\sum_s E\left(\tr^2\Delta_s\right)$, and observe that
		\begin{eqnarray}
			\sum_s E\left(\tr^2\Delta_s\right)=\sum_s E\left(\mathrm{tr}^2\left(\Upsilon_s\Xi_s\right)\right) &\leq & \sum_s E\left\{\overline{\lambda}^2\left(\Upsilon_s\right)\mathrm{tr}^2\left(\Xi_s\right)\right\}\nonumber\\
			& \leq & Cn^{2\nu}E\left(\sum_s\mathrm{tr}^2\left(\Xi_s\right)\right).\label{cltvarbd1}
		\end{eqnarray}
		The above inequalities are obtained as follows: first, the matrix $\Xi_s=\sum_{t_1,t_2<s}\xi_{t_1}\xi_{t_2}'$ is symmetric and positive semidefinite as it equals $\left(\sum_{t<s}\Omega^{-1/2}x_t\varepsilon_{t}\right)\left(\sum_{t<s}\Omega^{-1/2}x_t\varepsilon_{t}\right)'$. Because $\Upsilon_s$ is also symmetric psd, Theorem 1 of \cite{Fang1994} yields $\mathrm{tr}\left(\Upsilon_s\Xi_s\right)\leq \overline{\lambda}\left(\Upsilon_s\right)\mathrm{tr}(\Xi_s)$, whence the remaining inequality follows by Assumption \ref{ass:MCLT}. 
		
		Because $\mathrm{tr}\left(\Omega^{-1/2}x_{t_1}x_{t_2}'\Omega^{-1/2}\right)=x_{t_1}'\Omega^{-1}x_{t_2}$,
		the right side of (\ref{cltvarbd1}) is
		\begin{equation}\label{neweq1}
			Cn^{2\nu}\sum_{s}\sum_{t_1,t_2<s;t_3,t_4<s}E\left(x_{t_1}'\Omega^{-1}x_{t_2}\varepsilon_{t_1}\varepsilon_{t_2}x_{t_3}'\Omega^{-1}x_{t_4}\varepsilon_{t_3}\varepsilon_{t_4}\right).
		\end{equation}
		The contribution to (\ref{neweq1}) when $t_1=t_2=t_3=t_4$ is 
		\[
		Cn^{2\nu}\sum_s\sum_{t<s}E\left(\left(x_t'\Omega^{-1}x_t\right)^2\varepsilon_t^4\right)\leq Cn^{2\nu}\sum_s\sum_{t<s}\sum_{i,j=1}^p E\left(x_{it}^2x_{jt}^2\right)=O\left(n^{2\nu+2}p^2\right),
		\]
		by Assumptions \ref{ass:errors} and \ref{ass:M_diff}. Thus, this case contributes $O\left(n^{2\nu-2}\right)=o(1)$ to (\ref{vartarget}). Next, the  contribution to (\ref{neweq1}) from the case $(t_1=t_2)\neq (t_3=t_4)$ is
		\begin{eqnarray*}
			&&Cn^{2\nu}\sum_s\sum_{t_1<t_2<s}E\left(x_{t_1}'\Omega^{-1}x_{t_1}\varepsilon_{t_1}^2E\left(\left.x_{t_2}'\Omega^{-1}x_{t_2}\varepsilon_{t_2}^2\right\vert \mathcal{G}_{t_2-1}\right)\right)\\
			&\leq&Cn^{2\nu}\sum_s\sum_{t_1<s}E\left(x_{t_1}'\Omega^{-1}x_{t_1}\varepsilon_{t_1}^2\sum_{t_2<s}\text{tr}\Upsilon_{t_2}\right)\\
			&\leq&Cn^{2\nu+1}p\sum_{t_1\leq n}E\left(x_{t_1}'\Omega^{-1}x_{t_1}\varepsilon_{t_1}^2\sum_{t_2\leq n}\overline{\lambda}\left(\Upsilon_{t_2}\right)\right)\\
			&\leq&Cn^{3\nu+2}p\sum_{t_1\leq n}\text{tr}\left(E\left(x_{t_1}x_{t_1}'\varepsilon_{t_1}^2\right)\Omega^{-1}\right)\\
			&=&O\left(n^{3\nu+3}p^2\right),
		\end{eqnarray*}
		by Assumption \ref{ass:MCLT}, and because $\text{tr}\Upsilon_{t_2}\leq p\overline{\lambda}\left(\Upsilon_{t_2}\right) $. Thus, this case contributes $O\left(n^{3\nu+3}p^2\right)$ to \ref{cltvarbd1}, and therefore $O\left(n^{3\nu-1}\right)$  to (\ref{vartarget}). 
		
		The cases $(t_1=t_3)\neq (t_2=t_4)$ and $(t_1=t_4)\neq (t_2=t_3)$ similarly contribute a constant times
		\begin{eqnarray}
			&&n^{2\nu}\sum_s\sum_{t_1\neq t_2}E\left(x_{t_1}'\Omega^{-1}x_{t_2}\varepsilon_{t_1}\varepsilon_{t_2}\right)^2\nonumber\\
			&\leq& n^{2\nu}\sum_s\sum_{t_1\neq t_2}\left(E\left(x_{t_1}'\Omega^{-1}x_{t_1}\varepsilon^2_{t_1}\right)^2\right)^{1/2}\left(E\left(x_{t_2}'\Omega^{-1}x_{t_2}\varepsilon^2_{t_2}\right)^2\right)^{1/2}\nonumber\\
			&=&O\left(n^{2\nu+3}p^2\right),\label{cltvarbd3} 
		\end{eqnarray} 
		to (\ref{neweq1}), using the Cauchy Schwarz inequality. This ensures a negligible contribution of $O\left(n^{2\nu-1}\right)$ to (\ref{vartarget}). Finally,
		\begin{eqnarray}
			&&Cn^{2\nu}\sum_{s}\sum^{\neq}_{t_1,t_2<s;t_3,t_4<s}E\left(x_{t_1}'\Omega^{-1}x_{t_2}\varepsilon_{t_1}\varepsilon_{t_2}x_{t_3}'\Omega^{-1}x_{t_4}\varepsilon_{t_3}\varepsilon_{t_4}\right)\nonumber\\
			&=&O\left(n^{2\nu}\sum_{s}\sum^{\neq}_{t_1,t_2<s;t_3,t_4<s}\sum_{i,j=1}^p\left\vert E\left(x_{t_1,i}\varepsilon_{t_1}x_{t_2,i}\varepsilon_{t_2}x_{t_3,j}\varepsilon_{t_3}x_{t_4,j}\varepsilon_{t_4}\right)\right\vert\right)\nonumber\\
			&=&O\left(n^{2\nu+1}p^2\max_{s}\sum^{\neq}_{t_1,t_2<s;t_3,t_4<s}\max_{i,j=1,\ldots,p}\left\vert E\left(x_{t_1,i}\varepsilon_{t_1}x_{t_2,i}\varepsilon_{t_2}x_{t_3,j}\varepsilon_{t_3}x_{t_4,j}\varepsilon_{t_4}\right)\right\vert\right),\nonumber\\
			\label{cltvarbd2}
		\end{eqnarray}
		where $\sum^{\neq}_{t_1,t_2<s;t_3,t_4<s}$ excludes all cases which were considered before. Therefore, in view of (\ref{cltvarbd1}) and (\ref{cltvarbd2}), to establish negligibility of the first term in (\ref{vartarget}) it suffices to show
		\begin{equation}\label{cltvarbd3}
			n^{2\nu-3}\max_{s}\max_{i,j=1,\ldots,p}\sum^{\neq}_{t_1,t_2<s;t_3,t_4<s}\left\vert E\left(x_{t_1,i}\varepsilon_{t_1}x_{t_2,j}\varepsilon_{t_2}x_{t_3,k}\varepsilon_{t_3}x_{t_4,l}\varepsilon_{t_4}\right)\right\vert=o(1).
		\end{equation}
		The summand on the LHS abve is bounded by
		\begin{eqnarray}
			&&\left\vert E\left(x_{t_1,i}\varepsilon_{t_1}x_{t_2,i}\varepsilon_{t_2}\right)\right\vert \left\vert E\left(x_{t_3,j}\varepsilon_{t_3}x_{t_4,j}\varepsilon_{t_4}\right)\right\vert+\left\vert E\left(x_{t_1,i}\varepsilon_{t_1}x_{t_3,j}\varepsilon_{t_3}\right)\right\vert\left\vert E\left(x_{t_2,i}\varepsilon_{t_2}x_{t_4,j}\varepsilon_{t_4}\right)\right\vert\nonumber
			\\&+&\left\vert E\left(x_{t_1,i}\varepsilon_{t_1}x_{t_4,j}\varepsilon_{t_4}\right)\right\vert\left\vert E\left(x_{t_2,i}\varepsilon_{t_2}x_{t_3,j}\varepsilon_{t_3}\right)\right\vert\nonumber\\
			&+&\left\vert\mathrm{cum}_{iijj}\left(x_{t_1,i}\varepsilon_{t_1},x_{t_2,i}\varepsilon_{t_2},x_{t_3,j}\varepsilon_{t_3},x_{t_4,j}\varepsilon_{t_4}\right)\right\vert\nonumber\\
			&=&\left\vert c_{ii}\left(t_1-t_2\right)\right\vert \left\vert c_{jj}\left(t_3-t_4\right)\right\vert+\left\vert c_{ij}\left(t_1-t_3\right)\right\vert \left\vert c_{ij}\left(t_2-t_4\right)\right\vert\nonumber
			\\&+&\left\vert c_{ij}\left(t_1-t_4\right)\right\vert \left\vert c_{ji}\left(t_2-t_3\right)\right\vert\nonumber\\
			&+&\left\vert\mathrm{cum}_{iijj}\left(x_{0,i}\varepsilon_{0},x_{t_2-t_1,i}\varepsilon_{t_2-t_1},x_{t_3-t_1,j}\varepsilon_{t_3-t_1},x_{t_4-t_1,j}\varepsilon_{t_4-t_1}\right)\right\vert.\nonumber\\
			\label{cumcovmixing}
		\end{eqnarray}
		Because $\sum_{t_1,t_2}\left\vert c_{ij}\left(t_1-t_2\right)\right\vert\leq n\sum_{t=-\infty}^\infty\left\vert c_{ij}\left(t\right)\right\vert$, by Assumption \ref{ass:autocovandcumulant} and (\ref{cumcovmixing}) the LHS of (\ref{cltvarbd3}) is $O\left(n^{2(\nu+1)-3}\right)=O\left(n^{2\nu-1}\right)=o(1)$, as desired. Thus the first term in (\ref{vartarget}) is negligible, and by Assumption \ref{ass:MCLT} we conclude the proof of (\ref{target}).
		
		We now check the conditional Lindeberg condition
		\begin{equation}\label{lindeberg}
			\text{For all } \eta>0,\;\; \sum_{s}E\left(\left. \left(w_{ns}/\sqrt{n}\right)^2 1\left(\left\vert w_{ns}\right\vert>\eta\right)\right\vert \mathcal{G}_{s-1}\right)\overset{p}\rightarrow 0,
		\end{equation}
		in \cite{Hall1980}, Corollary 3.1, for which we verify the sufficient Lyapunov condition
		\begin{equation}\label{lyapunov}
			\sum_{s}E\left(\left. \left(w_{ns}/\sqrt{n}\right)^4\right\vert \mathcal{G}_{s-1}\right)\overset{p}\rightarrow 0.
		\end{equation}
		The LHS of (\ref{lyapunov}) is positive and, by law of iterated expectations, has mean
		\[
		n^{-2}\sum_{s}Ew_{ns}^4\leq n^{-1}\max_s Ew_{ns}^4=o\left(1\right),
		\]
		the final bound coming due to a calculation similar to the proof of (\ref{clt:var_cond}). Specifically, 
		\begin{equation}\label{tight_argument}
			n^{-1}\max_{s}E\left\vert w_{ns}\right\vert ^{4}\leq \max_s E\left( E((\xi_{s}'\xi_{s})^2 | \mathcal{G}_{s-1}) \left( \sum_{t_1 ,t_2 <s} \xi_{t_1}'\xi_{t_2}\right)^{2}\right)n^{-3}p^{-2}=O\left(n^{\omega+\nu-1}\right),
		\end{equation} 
		by Assumption \ref{ass:MCLT} and because the steps involved in showing (\ref{clt:var_cond}) imply that $E\left(\sum_{t_1 ,t_2 <s} \xi_{t_1}'\xi_{t_2}\right)^{2}=O\left(n^{\nu+2}p^2\right)$. Thus, because Assumption \ref{ass:MCLT} also implies that $O\left(n^{\omega+\nu-1}\right)=o(1)$, (\ref{lindeberg}) is established. A similar proof holds for the asymptotic normality of $\bar{\mathcal{A}}_{n}\left( \gamma \right)$.

		We finally derive the
		limiting covariance of $\left( \mathcal{A}_{n}\left( \gamma
		\right) ,\bar{\mathcal{A}}_{n}\left( \gamma \right) \right) ^{\prime }$. Using Assumption \ref{ass:MCLT}, we first compute 
		\begin{align*}
			E\left\vert \mathcal{A}_{n}\left( \gamma \right) \right\vert ^{2}& =\frac{1}{n}%
			\sum_{s=1}^{\left[ n\gamma \right] }Ew_{s}^{2} \\
			& =\frac{1}{n^{2}}\sum_{s=1}^{\left[ n\gamma \right] }s\left( \frac{1}{sp}{%
				\mathrm{tr}}\sum_{t_{1},t_{2}=1}^{s-1}E\left( \xi_{s}\xi_{s}^{\prime
			}\xi_{t_{1}}\xi_{t_{2}}^{\prime }\right) \right) \\
			& =\frac{\left( \left[ n\gamma \right] +1\right) \left[ n\gamma \right] }{%
				2n^{2}}\lim_{s,p\rightarrow \infty }\left( \frac{1}{sp}{\mathrm{tr}}%
			\sum_{t_{1},t_{2}=1}^{s-1}E\left( \xi_{s}\xi_{s}^{\prime
			}\xi_{t_{1}}\xi_{t_{2}}^{\prime }\right) \right) +o\left( 1\right) \\
			& =\frac{\gamma ^{2}{\mathcal{V}}}{2}+o\left( 1\right) ,
		\end{align*}%
		where ${\mathcal{V}}$ is given in (\ref{V_partial_def}). Next,
		\begin{align*}
			E\left\vert \bar{\mathcal{A}}_{n}\left( \gamma \right) \right\vert ^{2}& =E\left\vert 
			\frac{1}{n\sqrt{p}}{\sum_{s=[n\gamma ]+1}^{n}\sum_{t=[n\gamma ]+1}^{s-1}}\xi
			_{t}^{\prime }\xi _{s}\right\vert ^{2} \\
			& =E\left\vert \frac{1}{n\sqrt{p}}{\sum_{s=1}^{n-\left[ n\gamma \right]
				}\sum_{t=1}^{s-1}}\xi _{t+\left[ n\gamma \right] }^{\prime }\xi _{s+\left[
				n\gamma \right] }\right\vert ^{2} \\
			& =\frac{\left( 1-\gamma \right) ^{2}{\mathcal{V}}}{2}+o\left( 1\right).
		\end{align*}%
		Finally,
		\[
		E\left( \mathcal{A}_{n}\left( \gamma\right) \bar{\mathcal{A}}_{n}\left( \gamma \right)
		\right) =0.
		\]
		Therefore, we conclude that 
		\begin{equation*}
			\left( 
			\begin{array}{c}
				\mathcal{A}_{n}\left( \gamma \right) \\ 
				\bar{\mathcal{A}}_{n}\left( \gamma \right)%
			\end{array}%
			\right) \overset{d}{\rightarrow} \sqrt{\frac{{\mathcal{V}}}{2}}\left( 
			\begin{array}{c}
				W\left( \gamma \right) \\ 
				\bar{W}\left( \gamma \right)%
			\end{array}%
			\right) ,
		\end{equation*}
		pointwise in $\gamma\in\Gamma$.
		
		Finally, apply the continuous mapping theorem to get, pointwise in $\gamma$, 
		\begin{align*}
			{\mathcal{S}}_{n}(\gamma )& =\sqrt{2}\left( \frac{\mathcal{A}_{n}(\gamma )}{\gamma }+%
			\frac{\bar{\mathcal{A}}_{n}\left( \gamma \right) }{\left( 1-\gamma \right) }%
			-\mathcal{A}_{n}(1)\right) \\
			& \overset{d}{\rightarrow} \sqrt{{\mathcal{V}}}\left( \frac{W(\gamma )}{\gamma }+\frac{%
				\bar{W}\left( \gamma \right) }{\left( 1-\gamma \right) }-W(1)\right) =\sqrt{\mathcal{V}}
			\mathcal{Q}(\gamma ),
		\end{align*}%
		where $\mathcal{Q}(\gamma )$ has a standard normal distribution for this given $\gamma$. 
	\end{proof}
	
	We record some preliminary calculations useful for the sequel. Note that 
	\begin{equation*}
		\hat{\delta}_{2}(\gamma )=A(\gamma )y=\delta _{2}+A(\gamma )e=\delta
		_{2}+A(\gamma )\varepsilon +A(\gamma )r.
	\end{equation*}%
	Because $\delta _{2}=0$ under $\mathcal{H}_{0}$, we have 
	\begin{equation}
		W_{n}(\gamma )=n\left( \varepsilon +r\right) ^{\prime }A^{\prime }(\gamma )%
		\hat{B}(\gamma )^{-1}A(\gamma )\left( \varepsilon +r\right) ,
		\label{wald_split}
	\end{equation}%
	where we recall that $\hat{B}(\gamma )=R\hat{M}(\gamma )^{-1}\hat{\Omega}(\gamma )\hat{M}%
	(\gamma )^{-1}R^{\prime }$.
	
	\begin{lemma}
		\label{lemma:Mhat_norm} Under the conditions of Theorem \ref%
		{theorem:struc_break_approx}, for all sufficiently large $n$, 
		\begin{equation*}
			\sup _{\gamma \in \Gamma }\left \Vert \hat {M}(\gamma )\right \Vert=O_p(1),\sup
			_{\gamma \in \Gamma }\left \Vert \hat {M}(\gamma )^{-1}\right \Vert
			=O_{p}(\lambda_n^{-1}).
		\end{equation*}
	\end{lemma}
	
	\begin{proof}
		Note that, by the triangle inequality, 
		\begin{equation*}
			\left\Vert \hat{M}(\gamma )^{-1}\right\Vert \leq \left\Vert \hat{M}(\gamma
			)^{-1}\right\Vert \left\Vert \hat{M}(\gamma )-M(\gamma )\right\Vert
			\left\Vert M(\gamma )^{-1}\right\Vert +\left\Vert M(\gamma )^{-1}\right\Vert
			,
		\end{equation*}%
		so 
		\begin{equation*}
			\left\Vert \hat{M}(\gamma )^{-1}\right\Vert \left( 1-\left\Vert \hat{M}%
			(\gamma )-M(\gamma )\right\Vert \left\Vert M(\gamma )^{-1}\right\Vert
			\right) \leq \left\Vert M(\gamma )^{-1}\right\Vert ,
		\end{equation*}%
		using the triangle inequality. Taking limits of the last displayed
		expression as $n\rightarrow \infty $ and using Assumption \ref{ass:M_diff},
		the rate condition (\ref{rate:struc_break_approx}) yields $\left\Vert \hat{M}(\gamma )^{-1}\right\Vert =O_{p}(\lambda_n^{-1})$.
		Next, noting that 
		\begin{equation*}
			\left\Vert \hat{M}(\gamma )\right\Vert \leq \left\Vert \hat{M}(\gamma
			)-M(\gamma )\right\Vert +\left\Vert M(\gamma )\right\Vert ,
		\end{equation*}%
		the lemma follows by using Assumption \ref{ass:M_diff}.
	\end{proof}
	
	It is useful to first establish the stochastic order of $\left\Vert \delta -%
	\hat{\delta}(\gamma )\right\Vert $.
	
	\begin{lemma}
		\label{lemma:deltahat} \sloppy Under the conditions of Theorem \ref%
		{theorem:struc_break_approx}, $\sup_{\gamma \in \Gamma }\left\Vert \delta -%
		\hat{\delta}(\gamma )\right\Vert =O_{p}\left( \lambda_n^{-1}\sqrt{p/n}\right) .$
	\end{lemma}
	
	\begin{proof}
		Note that $\delta -\hat{\delta}(\gamma )=\hat{M}(\gamma
		)^{-1}n^{-1}\sum_{t=1}^{n}x_{t}(\gamma )e_{t}$ and that 
		\begin{eqnarray*}
			\left\Vert \delta -\hat{\delta}(\gamma )\right\Vert ^{2} &=&O_{p}\left(
			\left\Vert \hat{M}(\gamma )^{-1}\right\Vert ^{2}n^{-2}\left\Vert
			\sum_{t=1}^{n}x_{t}(\gamma )e_{t}\right\Vert ^{2}\right) =\lambda_n^{-2}O_{p}\left(
			n^{-2}\left\Vert \sum_{t=1}^{n}x_{t}(\gamma )e_{t}\right\Vert ^{2}\right) \\
			&=&\lambda_n^{-2}O_{p}\left( n^{-2}\left\Vert \sum_{t=1}^{n}x_{t}(\gamma )\varepsilon
			_{t}\right\Vert ^{2}+n^{-2}\left\Vert X(\gamma )^{\prime }r\right\Vert
			^{2}\right) ,
		\end{eqnarray*}%
		uniformly in $\gamma $, by Lemma \ref{lemma:Mhat_norm}. Next, $E\left(
		n^{-2}\left\Vert \sum_{t=1}^{n}x_{t}(\gamma )\varepsilon _{t}\right\Vert
		^{2}\right) $ equals 
		\begin{equation}
			E\left( n^{-2}\sum_{s,t=1}^{n}x_{t}^{\prime }(\gamma )x_{s}(\gamma
			)\varepsilon _{s}\varepsilon _{t}\right) ,  \label{Omegadifflemma1}
		\end{equation}%
		which is 
		\begin{equation}
			n^{-2}\sum_{t=1}^{n}E\left\Vert x_{t}(\gamma )\right\Vert ^{2}\sigma
			_{t}^{2}+2n^{-2}\sum_{s<t}E\left(x_{t}^{\prime }(\gamma )x_{s}(\gamma )E\left(
			\varepsilon _{s}E\left( \varepsilon _{t}|\varepsilon _{r},r<t\right) \right)\right)
			=O_{p}\left( p/n\right) ,  \label{lemma2_first_expec}
		\end{equation}%
		by Assumptions \ref{ass:errors} and $Ex_{t}^{\prime }(\gamma )x_{t}(\gamma
		)=O\left( p\right) $. Finally, 
		\begin{equation}
			n^{-2}\left\Vert X(\gamma )^{\prime }r\right\Vert ^{2}\leq n^{-2}\left\Vert
			X(\gamma )\right\Vert ^{2}\left\Vert r\right\Vert ^{2}=\overline{\lambda }%
			\left( \hat{M}(\gamma )\right) n^{-1}\left\Vert r\right\Vert
			^{2}=O_{p}\left( 1/n\right) ,  \label{Xr_bound}
		\end{equation}%
		by (\ref{r_order}) and Lemma \ref{lemma:Mhat_norm}. Therefore, 
		\begin{equation}
			\sup_{\gamma \in \Gamma }\left\Vert \delta -\hat{\delta}(\gamma )\right\Vert
			=O_{p}\left(\lambda_n^{-1} \sqrt{p}/\sqrt{n}\right) ,  \label{delta_deltahat_order}
		\end{equation}%
		by Markov's inequality.
	\end{proof}
	
	Observe that because 
	\begin{equation}
		M(\gamma )^{-1}=\left[ 
		\begin{array}{cc}
			(1-\gamma )^{-1}M^{-1} & (1-\gamma )^{-1}M^{-1} \\ 
			(1-\gamma )^{-1}M^{-1} & \left[ \gamma (1-\gamma )\right] ^{-1}M^{-1}%
		\end{array}%
		\right] ,  \label{M_gammainv}
	\end{equation}%
	we have 
	\begin{equation}
		B(\gamma )^{-1}=\gamma (1-\gamma )M\Omega ^{-1}M.  \label{B_gamma1}
	\end{equation}
	
	\begin{lemma}
		\label{lemma:B_eigs} Under the conditions of Theorem \ref%
		{theorem:struc_break_approx}, 
		\begin{equation*}
			\sup_{\gamma \in \Gamma }\left\{ \underline{\lambda }\left( {B}(\gamma
			)\right) \right\} ^{-1} =O(\lambda_n^{-1})
			\quad \text{and} \quad
			\sup_{\gamma \in \Gamma }\overline{\lambda }\left(
			B(\gamma )\right) =O(\lambda_n^{-2}).
		\end{equation*}
	\end{lemma}
	
	\begin{proof}
		$\left \{\underline {\lambda }%
		\left ({B}(\gamma )\right )\right \}^{-1}=\overline {\lambda }\left ({B}%
		(\gamma )^{-1}\right )$, which, using (\ref{B_gamma1}), is bounded by 
		\begin{equation*}
			C\overline {\lambda }\left (M\Omega ^{-1}M\right )=C\left \Vert M\Omega
			^{-1}M\right \Vert \leq C\overline {\lambda }\left (M\right )^2\underline {%
				\lambda }\left (\Omega \right )^{-1}=O(\lambda_n^{-1}),
		\end{equation*}
		uniformly on the compact $\Gamma $, using Assumption \ref{ass:M_diff}$(ii)$. For the second part of the claim, because (\ref%
		{B_gamma1}) implies $B(\gamma )=\left [\gamma (1-\gamma
		)\right
		]^{-1}M^{-1}\Omega M^{-1}$, it follows similarly that $\overline {%
			\lambda }\left (B(\gamma )\right )$ is uniformly bounded by a constant times \footnote{If $ \underline{\lambda}(M\Omega^{-1}M) \geq \lambda_n $, the bound in this lemma becomes $ O(\lambda_n^{-1}) $.}
		\begin{equation*}
			\overline {\lambda }\left (M^{-1}\Omega M^{-1}\right )=\left \Vert
			M^{-1}\Omega M^{-1}\right \Vert \leq \underline {\lambda }\left (M\right
			)^{-2}\overline {\lambda }\left (\Omega \right )=O(\lambda_n^{-2}).
		\end{equation*}
	\end{proof}
	
	\begin{lemma}
		\label{lemma:Bhat_norm} Under the conditions of Theorem \ref{thm:T_weak_conv}, 
		\begin{equation*}
			\sup_{\gamma \in \Gamma }\left\Vert \hat{B}(\gamma )\right\Vert
			=O_p(\lambda_n^{-2}),\sup_{\gamma \in \Gamma }\left\Vert \hat{B}(\gamma )^{-1}\right\Vert
			=O_{p}(\lambda_n^{-1}).
		\end{equation*}
	\end{lemma}
	
	\begin{proof}
		We show the second claim, the first following easily by the definition of $\hat{B}(\gamma)$. First, define $\tilde{B}(\gamma )=R\hat{M}(\gamma )^{-1}\Omega (\gamma )\hat{%
			M}(\gamma )^{-1}R^{\prime }$. We will use uniform bounds in the calculations
		without explicitly mentioning this in each step to simplify notation.
		Proceeding as in the proof of Lemma \ref{lemma:Mhat_norm}, we can write 
		\begin{align}
			\left\Vert \hat{B}(\gamma )^{-1}\right\Vert \left( 1-\left\Vert \hat{B}%
			(\gamma )-\tilde{B}(\gamma )\right\Vert \right) & \leq \left\Vert \tilde{B}%
			(\gamma )^{-1}\right\Vert ,  \label{Bhat_til} \\
			\left\Vert \tilde{B}(\gamma )^{-1}\right\Vert \left( 1-\left\Vert \tilde{B}%
			(\gamma )-B(\gamma )\right\Vert \right) & \leq \left\Vert B(\gamma
			)^{-1}\right\Vert .  \label{Bhat_B}
		\end{align}%
		Next, Lemma \ref{lemma:Mhat_norm} implies 
		\begin{equation}
			\left\Vert \hat{B}(\gamma )-\tilde{B}(\gamma )\right\Vert \leq \left\Vert
			R\right\Vert ^{2}\left\Vert \hat{M}(\gamma )^{-1}\right\Vert ^{2}\left\Vert 
			\hat{\Omega}(\gamma )-\Omega (\gamma )\right\Vert =O_{p}\left( \lambda_n^{-2}v_p \right)=o_p(1) .
			\label{Bhat_til_diff}
		\end{equation}%
		On the other
		hand, $\tilde{B}(\gamma )-B(\gamma )$ equals 
		\begin{equation*}
			R\left[ \hat{M}(\gamma )^{-1}\Omega (\gamma )\hat{M}(\gamma )^{-1}-{M}%
			(\gamma )^{-1}\Omega (\gamma ){M}(\gamma )^{-1}\right] R^{\prime }.
		\end{equation*}%
		By adding and subtracting terms inside the square brackets, this can be
		written as 
		\begin{align}
			& R\left[ {M}(\gamma )^{-1}\left( \hat{M}(\gamma )-{M}(\gamma )\right) \hat{M%
			}(\gamma )^{-1}\Omega (\gamma )\hat{M}(\gamma )^{-1}\right] R^{\prime } 
			\notag \\
			& +R{M}(\gamma )^{-1}\Omega (\gamma ){M}(\gamma )^{-1}\left( \hat{M}(\gamma
			)-{M}(\gamma )\right) \hat{M}(\gamma )^{-1}R^{\prime }.  \label{Btil_diff_nn}
		\end{align}%
		By this fact, Assumption \ref{ass:M_diff}, Lemmas \ref{lemma:Omega_hat_true} and \ref{lemma:Mhat_norm},
		and (\ref{rate:struc_break_approx}), we deduce from (\ref{Btil_diff_nn}) that 
		\begin{equation}
			\left\Vert \tilde{B}(\gamma )-B(\gamma )\right\Vert =O_{p}\left( \lambda_n^{-3}\varkappa_p  \right) =o_{p}(1).
			\label{Bhat_diff}
		\end{equation}%
		The lemma now follows by taking limits of (\ref{Bhat_til}) and (\ref{Bhat_B}%
		), and using (\ref{Bhat_til_diff}), (\ref{Bhat_diff}) and Lemma \ref%
		{lemma:B_eigs}.
	\end{proof}
	
	\begin{lemma}
		\label{lemma:wald_approx} Under the conditions of Theorem \ref{thm:T_weak_conv}  and $\mathcal{H}_{0}$, 
		\begin{equation*} 
			\frac{W_{n}(\gamma )}{\sqrt{2p}}=\frac{n\varepsilon ^{\prime }A(\gamma
				)^{\prime }B(\gamma )^{-1}A(\gamma )\varepsilon }{\sqrt{2p}}+o_{p}(1).
		\end{equation*}
	\end{lemma}
	
	\begin{proof}
		Recall the notation $\hat{M}=n^{-1}X^{\prime }X$ and $\hat{S}(\gamma
		)=n^{-1}X^{\prime \ast }(\gamma )X(\gamma )$. Notice that from (\ref%
		{wald_split}) we obtain 
		\begin{equation}
			\frac{W_{n}(\gamma )}{\sqrt{2p}}=\frac{n\varepsilon ^{\prime }A(\gamma
				)^{\prime }\hat{B}(\gamma )^{-1}A(\gamma )\varepsilon }{\sqrt{2p}}+\frac{%
				2n\varepsilon ^{\prime }A(\gamma )^{\prime }\hat{B}(\gamma )^{-1}A(\gamma )r%
			}{\sqrt{2p}}+\frac{nr^{\prime }A(\gamma )^{\prime }\hat{B}(\gamma
				)^{-1}A(\gamma )r}{\sqrt{2p}},  \label{wald_split_proof} 
		\end{equation}%
		with $r$ the $n\times 1$ vector with elements $ r_t $.
		Begin with the modulus of the last term on the RHS of (\ref{wald_split_proof}). 
		Recalling the relation in \eqref{fancyR1} and \eqref{fancyR2} for $ A(\gamma) r $,
		we bound it by ${Cn}/{\sqrt{2p}}$ times 
		\begin{eqnarray}
			\label{wald_approx_1}\\
			&&
			O_p\left(\left\Vert n^{-1}X^{\prime }r\right\Vert^2 \left\Vert I - \hat{S}(\gamma)\hat{M}^{-1} \right\Vert ^{2}
			\left\Vert \left(
			n^{-1}X^{\ast }(\gamma )^{\prime }M_{X}X^{\ast }(\gamma )\right)
			^{-1}\right\Vert ^{2}\left\Vert \hat{B}(\gamma )^{-1}\right\Vert\right) .
			\notag\\
			&&= O_p(\lambda_n^{-3} n^{-1}) , \notag
		\end{eqnarray}%
		where Assumption \ref{ass:aprx0} bounds the first term, Lemma \ref{lemma:fancyS_lemma1} yields a bound for the second and third terms after expanding the third term by \eqref{fancyR1}, and the last term is $ O_p \left(\lambda_n^{-1}\right) $ Lemma \ref{lemma:Bhat_norm}. 
		Thus (\ref{wald_approx_1}) implies that the third term on the RHS of (\ref%
		{wald_split_proof}) is $o_{p}\left( 1\right) $.
		
		We now show that the first term on the RHS of (\ref{wald_split_proof}) is 
		\begin{equation}
			\frac{n\varepsilon ^{\prime }A(\gamma )^{\prime }B(\gamma )^{-1}A(\gamma
				)\varepsilon }{\sqrt{2p}}+o_{p}(1).  \label{wald_approx_2}
		\end{equation}%
		Indeed, as above,
		\begin{eqnarray}
			\frac{n\varepsilon ^{\prime }A(\gamma )^{\prime }\left( \hat{B}(\gamma
				)^{-1}-B(\gamma )^{-1}\right) A(\gamma )\varepsilon }{\sqrt{2p}} &=&\frac{n}{%
				\sqrt{p}}O_{p}\left( \left\Vert n^{-1}X^{\prime }\varepsilon \right\Vert
			^{2}\left\Vert \hat{B}(\gamma )^{-1}-B(\gamma )^{-1}\right\Vert \right) 
			\notag \\
			&=&\sqrt{p} O_{p}\left(\left\Vert B(\gamma )^{-1}\right\Vert\left\Vert B(\gamma )-\hat{B}(\gamma
			)\right\Vert \left\Vert\hat{B}(\gamma )^{-1}\right\Vert\right)  \notag \\
			&=&\lambda_n^{-4}\sqrt{p}O_{p}\left(\lambda_n^{-1} \left\Vert \hat{M}(\gamma )-M(\gamma )\right\Vert
			+\right.\\
			&&\left.\left\Vert \hat{\Omega}(\gamma )-\Omega (\gamma )\right\Vert \right)  \notag
			\\
			&=&O_{p}\left(\lambda_n^{-4} \sqrt{p}\left( \lambda_n^{-1}\varkappa _{p}+v_{p}\right) \right) ,
			\label{Omegadifflater}
		\end{eqnarray}%
		using equations (\ref{Bhat_til_diff}) and (\ref{Bhat_diff}). This is
		negligible by (\ref{rate:struc_break_approx}).
		
		For the second term  on the RHS of (\ref{wald_split_proof}), apply the Cauchy-Schwarz inequality and the preceding two results. Then, the second term becomes $ o_p(1) $, establishing the lemma. 
	\end{proof}
	Denote, for convenience, $C(\gamma )=\left[ \gamma \left( 1-\gamma \right) %
	\right] ^{-1}n^{-1}\Sigma ^{\frac{1}{2}}G(\gamma )\Omega ^{-1}G(\gamma
	)\Sigma^{ \frac{1}{2}}$, where $\Sigma =diag\left[ \sigma _{1}^{2},\ldots ,\sigma _{n}^{2}\right] $.

	\begin{lemma}
		\label{lemma:asy_idem_weight_matrix} Under the conditions of Theorem \ref{thm:T_weak_conv}, any eigenvalue $\lambda $ of $C(\gamma )$ satisfies 
		\begin{equation*}
			P\left( \left\vert \lambda (\lambda -1)\right\vert <\eta \right) \rightarrow
			1,
		\end{equation*}%
		as $n\rightarrow \infty $, for any $\eta >0$.
	\end{lemma}
	
	\begin{proof}
		We have 
		\begin{eqnarray}
			C(\gamma )^2&=&\left [\gamma \left (1-\gamma \right )\right
			]^{-1}n^{-1}\Sigma ^{\frac{1}{2}}G(\gamma )\Omega ^{-1}\left [\gamma \left
			(1-\gamma \right )\right ]^{-1}n^{-1}G(\gamma )'\Sigma G(\gamma)\Omega^{-1}G(\gamma)'\Sigma^{\frac{1}{2}}  \notag \\
			&=&\left [\gamma \left (1-\gamma \right )\right ]^{-1}n^{-1}\Sigma ^{\frac{1%
				}{2}}G(\gamma )\Omega ^{-1}\Omega \Omega ^{-1}G(\gamma )' \Sigma^{\frac{1}{2}}
			\notag \\
			&+&\left [\gamma \left (1-\gamma \right )\right ]^{-1}n^{-1}\Sigma ^{\frac{1%
				}{2}}G(\gamma )\Omega ^{-1}\left \{\left [\gamma \left (1-\gamma \right
			)\right ]^{-1}n^{-1}G(\gamma )^{\prime }\Sigma G(\gamma )-\Omega \right
			\}\notag\\
			&\times&\Omega ^{-1}G(\gamma )' \Sigma^{\frac{1}{2}}  \notag \\
			&=&C(\gamma )+D(\gamma ) ,  \notag
		\end{eqnarray}
		say. We now prove that 
		\begin{equation}  \label{D_negl}
			\left \Vert D(\gamma )\right \Vert =o_p(1){\text { as }}n\rightarrow \infty .
		\end{equation}
		In view of Assumptions \ref{ass:errors} and \ref{ass:M_diff}$(i)$, to prove (\ref{D_negl}) it suffices to show that 
		\begin{equation}  \label{D_negl2}
			\left \Vert \left [\gamma \left (1-\gamma \right )\right
			]^{-1}n^{-1}G(\gamma )^{\prime }\Sigma G(\gamma )-\Omega \right \Vert
			=o_p(1).
		\end{equation}
		But 
		\begin{eqnarray*}
			&&n^{-1}G(\gamma )'\Sigma G(\gamma)=n^{-1}(1-2\gamma )\sum _{t=1}^{[n\gamma
				]}x_tx_t^{\prime }\sigma _t^2+\gamma ^2\Omega  \notag \\
			&=&(1-2\gamma )\left (n^{-1}\sum _{t=1}^{[n\gamma ]}x_tx_t^{\prime }\sigma
			_t^2-\gamma \Omega \right )+\left [\gamma (1-\gamma )\right ]\Omega ,
		\end{eqnarray*}
		so (\ref{D_negl2}) follows if $\left \Vert n^{-1}\sum _{t=1}^{[n\gamma
			]}x_tx_t^{\prime }\sigma _t^2-\gamma \Omega \right \Vert =o_p(1)$, which is
		true by Assumption \ref{ass:M_diff}. Thus (\ref{D_negl}) is established.
		
		Let $\lambda $ be any eigenvalue of $C(\gamma )$ and $w$ be the
		corresponding eigenvector, normalised to $\left \Vert w\right \Vert =1$.
		Because $\lambda w=C(\gamma )w$, we have $\lambda C(\gamma )w=C(\gamma
		)^2w=\left [C(\gamma )+D(\gamma )\right ]w=\lambda w+D(\gamma )w$, implying $%
		\lambda (\lambda -1)w=D(\gamma )w$. Thus 
		\begin{equation}  \label{C_eig_1}
			\left \vert \lambda (\lambda -1)\right \vert =\left \Vert D(\gamma )w\right
			\Vert \leq \left \Vert D(\gamma )\right \Vert .
		\end{equation}
		Then, for arbitrary $\eta >0$, 
		\begin{equation*}
			P\left (\left \vert \lambda (\lambda -1)\right \vert <\eta \right )=P\left
			(\left \Vert D(\gamma )w\right \Vert <\eta \right )\geq P\left (\left \Vert
			D(\gamma )\right \Vert <\eta \right )\rightarrow 1,{\text { as }}%
			n\rightarrow \infty ,
		\end{equation*}
		by (\ref{D_negl}). This completes the proof.
	\end{proof}
	
	We have ${\mathcal{R}_{n}}(\gamma )=\left[ \gamma \left( 1-\gamma \right) \right]
	^{-1}n^{-1}\varepsilon' G(\gamma )^{\prime} \Omega^{-1}G(\gamma )^{\prime }\varepsilon$,
	which in turn equals 
	\begin{equation}
		\left[ \gamma \left( 1-\gamma \right) \right] ^{-1}n^{-1}%
		\sum_{t,s=1}^{n}g_{t}(\gamma )'\Omega^{-1}g_{s}(\gamma )\varepsilon
		_{t}\varepsilon _{s}.  \label{S_sum}
	\end{equation}Note that $tr\left \{C(\gamma )\right \}$ is the sum of the eigenvalues of $%
	C(\gamma )$, which is a symmetric matrix with rank $p$. Thus, in view of
	Lemma \ref{lemma:asy_idem_weight_matrix} it has $p$ eigenvalues that
	approach 1 in probability, with the remainder approaching 0. Thus, 
	\begin{equation}
		\frac{{\mathcal{R}_n}(\gamma )-tr\left (C(\gamma )\right
			)}{\sqrt{2p}}=\frac{{\mathcal{R}_n}(\gamma )-p}{\sqrt{2p}}+o_{p}(1),
		\label{struc_break_approx_trace}
	\end{equation}
	whence using (\ref{S_sum}) we deduce that (\ref{struc_break_approx_trace})
	equals 
	\begin{equation}
		\frac{n^{-1}\sum _{t=1}^ng_t(\gamma )'\Omega^{-1}g_t(\gamma )\left
			(\varepsilon _t^2-\sigma _t^2\right )+n^{-1}\sum _{s\neq t}g_t(\gamma
			)'\Omega^{-1}g_s(\gamma )\varepsilon _t\varepsilon _s}{\gamma \left
			(1-\gamma \right )\sqrt{2p}}.  \label{quad_form_full}
	\end{equation}
	
	\begin{lemma}
		\label{lemma:diag_terms_neg} Under the conditions of Theorem \ref{thm:T_weak_conv}, 
		\begin{equation}
			\sup _{\gamma \in \Gamma }n^{-1}\sum _{t=1}^ng_t(\gamma )^{\prime
			}\Omega^{-1} g_t(\gamma )\left (\varepsilon _t^2-\sigma _t^2\right )=o_p(1){\text { as 
			}}n\rightarrow \infty .  \label{diag_terms_neg_tgt}
		\end{equation}
	\end{lemma}
	
	\begin{proof}
		Conditional on $x_{t}$, the LHS of (\ref{diag_terms_neg_tgt}) has mean zero
		and variance 
		\begin{eqnarray}
			&&n^{-2}\sum_{t=1}^{n}\left( g_{t}(\gamma )'\Omega^{-1}g_{t}(\gamma )\right)
			^{2}E\left[ \left( \varepsilon _{t}^{2}-\sigma _{t}^{2}\right) ^{2}\right]
			\label{diag_terms_neg1} \\
			&+&2n^{-2}\sum_{s<t}g_{s}(\gamma )'\Omega^{-1}g_{s}(\gamma )g_{t}(\gamma
			)'\Omega^{-1}g_{t}(\gamma )E\left[ \left( \varepsilon _{t}^{2}-\sigma
			_{t}^{2}\right) \left( \varepsilon _{s}^{2}-\sigma _{s}^{2}\right) \right] .
			\label{diag_terms_neg2}
		\end{eqnarray}%
		The expectation in (\ref{diag_terms_neg2}) equals $E\left[ \left(
		\varepsilon _{t}^{2}-\sigma _{t}^{2}\right) E\left( \left( \varepsilon
		_{s}^{2}-\sigma _{s}^{2}\right) |\varepsilon _{s}\right) \right] =0$, by
		Assumption \ref{ass:errors}. Also by Assumption \ref{ass:errors}, (\ref%
		{diag_terms_neg1}) is bounded by a constant times 
		\begin{equation*}
			n^{-2}\left\Vert \Omega ^{-1}\right\Vert ^{2}\sum_{t=1}^{n}\left\Vert
			g_{t}(\gamma )\right\Vert ^{4}\leq n^{-2}\left\Vert \Omega ^{-1}\right\Vert
			^{2}\sum_{t=1}^{n}\left( \left\Vert x_{t}(\gamma )\right\Vert ^{4}+\gamma
			^{4}\left\Vert x_{t}\right\Vert ^{4}\right) =O_{p}\left( \lambda_n^{-2}\frac{p^{2}}{n}%
			\right) ,
		\end{equation*}%
		uniformly in $\gamma $, the last equality following by Assumption \ref{ass:M_diff}$(i)$.
	\end{proof}
	
	\begin{lemma}
		\label{lemma:fancyS_lemma1} Under the conditions of Theorem \ref{thm:T_weak_conv}, as $n\rightarrow \infty $,%
		\begin{equation*}
			\left\Vert \left( I-\hat{M}^{-1}\hat{S}(\gamma )\right) ^{-1}-\gamma
			^{-1}I\right\Vert =O_{p}\left(\lambda_n^{-1} \varkappa _{p}\right) .
		\end{equation*}
	\end{lemma}
	
	\begin{proof}
		First note that $\left \Vert \left (I-\hat {M}^{-1}\hat {S}(\gamma
		)\right
		)-\gamma I\right \Vert $ equals 
		\begin{align}
			&\left \Vert (1-\gamma )I-\hat {M}^{-1}\left (\hat {S}(\gamma )-(1-\gamma
			)M\right )-(1-\gamma )\hat {M}^{-1}M\right \Vert  \notag \\
			&\leq C\left \Vert \hat {M}^{-1}\right \Vert \left (\left \Vert 
			\hat {M}-M\right \Vert +\left \Vert \hat {S}(\gamma )-(1-\gamma )M\right
			\Vert \right )  \notag \\
			&=O_{p}\left (\lambda_n^{-1}\varkappa _p\right ),  \notag
		\end{align}
		by Assumptions \ref{ass:M_diff}. Since 
		\begin{equation*}
			\left (I-\hat {M}^{-1}\hat {S}(\gamma )\right )^{-1}-\gamma ^{-1}I=-\gamma
			^{-1}\left (I-\hat {M}^{-1}\hat {S}(\gamma )\right )^{-1}\left \{\left (I-%
			\hat {M}^{-1}\hat {S}(\gamma )\right )-\gamma I\right \}
		\end{equation*}
		and $\left \Vert \left (I-\hat {M}^{-1}\hat {S}(\gamma
		)\right
		)^{-1}\right
		\Vert  =O_{p}\left(1\right)$, the lemma is established.
	\end{proof}
	
	\begin{lemma}
		\label{lemma:fancyS_lemma2} Under the conditions of Theorem \ref{thm:T_weak_conv}, as $n\rightarrow \infty $, 
		\begin{eqnarray}  \label{fancyS_lemma2_eq}
			&&\left\Vert \left( \hat{S}(\gamma )^{-1}X^{\prime \ast }(\gamma
			)\varepsilon -\gamma (1-\gamma )^{-1}\hat{M}^{-1}X^{\prime }\varepsilon
			\right) -\hat{M}^{-1}\frac{\left( \sum_{t=1}^{[n\gamma
					]}\varepsilon _{t}x_{t}-\gamma \sum_{t=1}^{n}\varepsilon _{t}x_{t}\right)}{1-\gamma }
			\right\Vert  \notag \\
			&=&O_{p}\left( \lambda_n^{-2}\sqrt{np}\varkappa _{p}\right) .  \notag \\
		\end{eqnarray}
	\end{lemma}
	
	\begin{proof}
		First note that 
		\begin{equation*}
			(1-\gamma )^{-1}\left( \sum_{t=1}^{[n\gamma ]}\varepsilon _{t}x_{t}-\gamma
			\sum_{t=1}^{n}\varepsilon _{t}x_{t}\right) =(1-\gamma )^{-1}X^{\prime \ast
			}(\gamma )\varepsilon -\gamma (1-\gamma )^{-1}X^{\prime }\varepsilon ,
		\end{equation*}%
		so the term inside the norm in (\ref{fancyS_lemma2_eq}) equals 
		\begin{equation}
			\left( \hat{S}(\gamma )^{-1}-(1-\gamma )^{-1}\hat{M}^{-1}\right) X^{\prime
				\ast }(\gamma )\varepsilon =(1-\gamma )^{-1}\hat{M}^{-1}\left( (1-\gamma )%
			\hat{M}-\hat{S}(\gamma )\right) \hat{S}(\gamma )^{-1}X^{\prime \ast }(\gamma
			)\varepsilon .  \label{fancyS_lemma2_eq2}
		\end{equation}%
		The norm of the RHS of (\ref{fancyS_lemma2_eq2}) is bounded by a constant
		times 
		\begin{equation*}
			\left\Vert \hat{M}^{-1}\right\Vert \left\Vert \hat{S}(\gamma
			)^{-1}\right\Vert \left( \left\Vert n^{-1}\sum_{t=1}^{[n\gamma
				]}x_{t}x_{t}^{\prime }-\gamma M\right\Vert +\left\Vert \hat{M}-M\right\Vert
			\right) \left\Vert X^{\prime \ast }(\gamma )\varepsilon \right\Vert
			=O_{p}\left( \lambda_n^{-2}\sqrt{np}\varkappa _{p}\right),
		\end{equation*}%
		the last equality following from Assumptions \ref{ass:M_diff}, Lemma \ref{lemma:Mhat_norm}, and also (\ref%
		{lemma2_first_expec}).
	\end{proof}
	\section{Proof of Theorem \ref{thm:bootstrap}}\label{sec:bootstrap_proof}
	\begin{proof} It is sufficient to check (\ref{bootthmorig}), whence (\ref{bootthmnew}) follows. Let $ \varepsilon^{\star} $ denote the vector collecting $ \varepsilon_t^{\star} = \hat{e}_t(\gamma) \xi_t  $, where $ \xi_t $ is an iid sequence of Rademacher variables. Then, 
		\begin{equation*}
			\hat{\delta}_{2}^{\star}(\gamma )= A(\gamma )\varepsilon^{\star}, 
		\end{equation*}%
		since $\delta _{2}=0$ under $\mathcal{H}_{0}$. Also, we have 
		\begin{equation}
			W_{n}^{\star}(\gamma )=n\left( \varepsilon^{\star} \right) ^{\prime }A^{\prime }(\gamma )%
			\hat{B}^{\star}(\gamma )^{-1}A(\gamma ) \varepsilon^{\star} ,
			\label{wald_*}
		\end{equation}%
		where $\hat{B}^{\star}(\gamma )=R\hat{M}(\gamma )^{-1}\hat{\Omega}^{\star}(\gamma )\hat{M}(\gamma )^{-1}R^{\prime }$ and $ \hat{\Omega}^{\star}(\gamma ) $ is constructed as $ \hat{\Omega}(\gamma ) $ with the bootstrap sample. 
		
		We begin with 
		\begin{eqnarray}
			E^{\star} \bar{W}_{n}^{\star}(\gamma ) &=& n tr A^{\prime }(\gamma )		\hat{B}(\gamma )^{-1}A(\gamma ) E^{\star} \varepsilon^{\star} (\varepsilon^{\star})' ,\notag \\
			&=& n tr A^{\prime }(\gamma )		\hat{B}(\gamma )^{-1}A(\gamma ) diag\left[\hat{e}_1(\gamma)^2,...,\hat{e}_n(\gamma)^2\right], \label{eq:EWbar}
		\end{eqnarray} 
		where $ \bar{W}_{n}^{\star}(\gamma )=n\left( \varepsilon^{\star} \right) ^{\prime }A^{\prime }(\gamma )%
		{B}(\gamma )^{-1}A(\gamma ) \varepsilon^{\star} $. Note that the term in \eqref{eq:EWbar} subtracted by $ p $ is $ o_p (p^{1/2}) $ uniformly in $ \gamma $ due to Lemma \ref{lemma:asy_idem_weight_matrix}, Lemma \ref{lemma:diag_terms_neg}, and Lemma \ref{lemma:deltahat}.  
		
		Next, we show that the order of the difference between $ E^{\star} \bar{W}_{n}^{\star}(\gamma ) $ and $ E^{\star} W_{n}^{\star}(\gamma ) $ is $ o_p (p^{1/2}) $. Following \eqref{Omegadifflater}, write
		\begin{eqnarray}
			E^{\star}| \bar{W}_{n}^{\star}(\gamma ) -  W_{n}^{\star}(\gamma )| \leq   
			E^{\star} \left( \left\Vert n^{-1/2}A^{\prime }(\gamma)\varepsilon^{\star} \right\Vert
			^{2}\left\Vert \hat{B}(\gamma )^{-1}-\hat{B}^{\star}(\gamma )^{-1}\right\Vert \right). 
			\notag  
		\end{eqnarray}
		To apply the Cauchy-Schwarz inequality, and to bound $ E^{\star}\left\Vert \hat{B}(\gamma )^{-1}-\hat{B}^{\star}(\gamma )^{-1}\right\Vert^2	 $, we derive bounds for $ E^{\star}\left\Vert \hat{B}^{\star}(\gamma )^{-1}\right\Vert^4	 $ and $ E^{\star} \left\Vert \hat{B}(\gamma )-\hat{B}^{\star}(\gamma )\right\Vert^4 $. Since both are similar to the derivations for the sample counterparts in Lemmas  \ref{lemma:Omega_hat_true} and \ref{lemma:Bhat_norm}, we only illustrate the latter. 
		Recall $ \hat{B}(\gamma )-\hat{B}^{\star}(\gamma ) = R' \hat{M}(\gamma)^{-1} \left(\hat{\Omega}(\gamma )-\hat{\Omega}^{\star}(\gamma )\right)\hat{M}(\gamma)^{-1}R $ and $ \sup_{\gamma \in \Gamma} \left\Vert \hat{M}(\gamma)^{-1} \right\Vert = O_p (\lambda_n^{-1})$ by Lemma \ref{lemma:Mhat_norm}. 
		Following the steps in the proof of Lemma \ref{lemma:Omega_hat_true}, the term $ \hat{\Omega}(\gamma )-\hat{\Omega}^{\star}(\gamma )$ is given by the sum of $ U_1^{\star}(\gamma) $ and $U_3^{\star} (\gamma) $ therein. Due to the triangle inequality and $ c_r $ inequality, we only show $ E^{\star} \left\Vert U_j^{\star}(\gamma) \right\Vert^4 = O_p (\lambda_n^{-8} p^{12}/n^{4}) $, for $ j=1,3 $. Note that by the independence of the sequence $ \xi_{t} $
		\begin{eqnarray*}
			E^{\star} \left\Vert U_{1}^{\star}(\gamma )\right\Vert^4 &\leq& \left(n^{-1}\sum_{t=1}^{n}\left(
			x_{t}^{\prime }(\gamma )x_{t}(\gamma )\right) ^{2} \right)^4 
			E^{\star} \left(\left( \delta^{\star} -\hat{\delta}^{\star}(\gamma )\right) ^{\prime }\left( \delta^{\star} -\hat{\delta}^{\star}(\gamma )\right) \right)^4 \\
			&\leq & 
			O_p (p^8) \left\Vert \hat{M}(\gamma )^{-1}\right\Vert^8 n^{-8}\sum_{t_1,t_2,t_3,t_4} \hat{e}_{t_1}^2 x_{t_1}'x_{t_1} \cdots \hat{e}_{t_4}^2 x_{t_4}'x_{t_4}, 
		\end{eqnarray*} 
		to yield the desired result and the bound for $ U_3^{\star} $ is similarly obtained. Putting these together yields $ E^{\star}\left\Vert \hat{B}(\gamma )^{-1}-\hat{B}^{\star}(\gamma )^{-1}\right\Vert^2	=O_p (\lambda_n^{-10} p^{12}/n^{4}) $.
		
		Next, similar to the preceding bound, 
		\begin{eqnarray}
			E^{\star} \left\Vert n^{-1/2}A^{\prime }(\gamma)\varepsilon^{\star} \right\Vert
			^{4} = 
			O_p(\lambda_n^{-4}) \left( n^{-1} \sum_{t=1}^n x_t 'x_t \hat{e}_t^2 \right)^2 = 
			O_p (\lambda_n^{-4} p^2), \notag
		\end{eqnarray}
		as $ \xi_t $ is an iid Rademacher sequence. 
		Then, under the condition \eqref{rate:Q_weak_conv}, $ \lambda_n^{-14} p^{14}/n^{4} = o (p)$ and this completes the proof. 
	\end{proof}
	\newpage
	\section{Verification of covariance decay in Assumption \ref{ass:MCLT}}
	\begin{table}[ht]
		\centering 
		\caption{$(n^{4}p^{2})^{-1}\sum_{t=1}^{n}\sum_{s=1}^{t-1}cov\left(tr\left(\Upsilon_{t}\Xi_{t}\right),tr\left(\Upsilon_{s}\Xi_{s}\right)\right)$ with $ n=999,...,9999$ for multiple regression. }
		\begin{tabular}{llllllll}
			\hline
			type & $ \alpha_x $ & $ \alpha $ & 999 & 3249 & 5499 & 7749 & 9999 \\
			\hline
			1 & 0.1 & 0.3 & 0.0082 & 0.0058 & 0.0047 & 0.0043 & 0.0043 \\ 1 & 0.1 & 0.4 & 0.0086 & 0.0053 & 0.0044 & 0.0043 & 0.0038 \\ 1 & 0.1 & 0.5 & 0.0086 & 0.0055 & 0.0046 & 0.0041 & 0.0037 \\ 1 & 0.1 & 0.55 & 0.008 & 0.0054 & 0.0042 & 0.0037 & 0.0034 \\ 1 & 0.5 & 0.3 & 0.0377 & 0.0264 & 0.0233 & 0.0219 & 0.0212 \\ 1 & 0.5 & 0.4 & 0.0405 & 0.0251 & 0.0225 & 0.0197 & 0.0191 \\ 1 & 0.5 & 0.5 & 0.0357 & 0.0217 & 0.0174 & 0.0147 & 0.0138 \\ 1 & 0.5 & 0.55 & 0.038 & 0.0221 & 0.0171 & 0.0155 & 0.0138 \\ 1 & 0.7 & 0.3 & 0.2278 & 0.1735 & 0.1388 & 0.137 & 0.1354 \\ 1 & 0.7 & 0.4 & 0.2615 & 0.1858 & 0.1544 & 0.1413 & 0.1375 \\ 1 & 0.7 & 0.5 & 0.2757 & 0.1561 & 0.1416 & 0.1359 & 0.1313 \\ 1 & 0.7 & 0.55 & 0.2062 & 0.1484 & 0.1352 & 0.1149 & 0.1002 \\ 2 & 0.1 & 0.3 & 0.0155 & 0.0075 & 0.0054 & 0.0047 & 0.0044 \\ 2 & 0.1 & 0.4 & 0.011 & 0.0061 & 0.0054 & 0.0046 & 0.0042 \\ 2 & 0.1 & 0.5 & 0.0202 & 0.0108 & 0.0122 & 0.0069 & 0.006 \\ 2 & 0.1 & 0.55 & 0.0355 & 0.095 & 0.0199 & 0.0091 & 0.006 \\ 2 & 0.5 & 0.3 & 0.0518 & 0.0334 & 0.0265 & 0.0241 & 0.0219 \\ 2 & 0.5 & 0.4 & 0.0614 & 0.0307 & 0.0255 & 0.0241 & 0.0223 \\ 2 & 0.5 & 0.5 & 0.0703 & 0.0498 & 0.0302 & 0.0256 & 0.0238 \\ 2 & 0.5 & 0.55 & 1.6204 & 0.0992 & 0.0373 & 0.0313 & 0.0326 \\ 2 & 0.7 & 0.3 & 0.3164 & 0.203 & 0.188 & 0.1722 & 0.1536 \\ 2 & 0.7 & 0.4 & 0.3501 & 0.2266 & 0.2175 & 0.1972 & 0.183 \\ 2 & 0.7 & 0.5 & 0.9339 & 0.2196 & 0.8319 & 0.3447 & 0.3014 \\ 2 & 0.7 & 0.55 & 1.4368 & 0.5232 & 0.2518 & 0.1936 & 0.167 \\ 3 & 0.1 & 0.3 & 0.0078 & 0.0039 & 0.0035 & 0.0032 & 0.0028 \\ 3 & 0.1 & 0.4 & 0.0061 & 0.0035 & 0.0028 & 0.0024 & 0.0021 \\ 3 & 0.1 & 0.5 & 0.004 & 0.0023 & 0.0019 & 0.0018 & 0.0017 \\ 3 & 0.1 & 0.55 & 0.0035 & 0.0019 & 0.0014 & 0.0013 & 0.0012 \\ 3 & 0.5 & 0.3 & 0.0363 & 0.0232 & 0.0161 & 0.0141 & 0.0135 \\ 3 & 0.5 & 0.4 & 0.0229 & 0.0169 & 0.0134 & 0.0122 & 0.0108 \\ 3 & 0.5 & 0.5 & 0.0202 & 0.0122 & 0.0119 & 0.0099 & 0.0092 \\ 3 & 0.5 & 0.55 & 0.0172 & 0.01 & 0.0075 & 0.0067 & 0.006 \\ 3 & 0.7 & 0.3 & 0.2202 & 0.1438 & 0.1218 & 0.1062 & 0.0963 \\ 3 & 0.7 & 0.4 & 0.1665 & 0.107 & 0.0919 & 0.0849 & 0.0789 \\ 3 & 0.7 & 0.5 & 0.1273 & 0.0807 & 0.077 & 0.0696 & 0.062 \\ 3 & 0.7 & 0.55 & 0.1106 & 0.0556 & 0.0522 & 0.0453 & 0.0438 \\
			\hline
		\end{tabular}
		\label{Table:cov}
	\end{table}

\end{document}